\documentclass{LMCS}

\usepackage{amsmath}
\usepackage{amssymb}
\usepackage{amsfonts}
\usepackage{stmaryrd}
\usepackage{epsfig}
\usepackage{xspace}
\usepackage{algorithm}
\usepackage[noend]{algorithmic}
\usepackage{verbatim}
\usepackage{url}
\usepackage[inference]{semantic}
\usepackage{multirow}
\usepackage{rotating}
\usepackage{array}
\usepackage{enumerate,hyperref}

\theoremstyle{plain}
\def\eg{{\em e.g.}}

\def\PEG{PEG\xspace}
\def\PEGs{PEGs\xspace}
\def\EPEG{E-PEG\xspace}
\def\EPEGs{E-PEGs\xspace}
\def\Peggy{Peggy\xspace}

\def\param{{\it param}}
\def\eval{{\it eval}}
\def\pass{{\it pass}}
\def\peel{{\it peel}}
\def\suminner{{\it sum_{inner}}}
\def\sumouter{{\it sum_{outer}}}
\def\context{\Psi}

\newcommand{\B}{\mathbb{B}}
\newcommand{\N}{\mathbb{N}}
\newcommand{\W}{\mathbb{N}}
\newcommand{\I}{\mathbb{I}}
\renewcommand{\i}{\mathbf{i}}
\newcommand{\D}{\tau}
\newcommand{\U}[1]{#1_\bot}
\newcommand{\Lifted}[1]{\widetilde{#1}}
\newcommand{\set}[1]{\mathcal{#1}}

\newcommand{\sliwrt}[2]{\invariant_{#2}({#1})}
\newcommand{\corresponds}[3]{\semvalue{#1}(#3)=#2}

\newcommand{\true}{\mathbf{true}}
\newcommand{\false}{\mathbf{false}}

\newcommand{\undef}{\bot}

\newcommand{\base}{\mathit{base}}
\newcommand{\loopit}{\mathit{loop}}
\newcommand{\idx}{\mathit{idx}}
\newcommand{\cond}{\mathit{cond}}
\newcommand{\monotonize}{\mathit{monotonize}}

\newcommand{\sem}[1]{#1}
\newcommand{\semvalue}[1]{\llbracket #1 \rrbracket}

\DeclareMathOperator{\FixpointTemporaries}{FixpointTemps}
\DeclareMathOperator{\TranslateStmt}{TS}
\DeclareMathOperator{\TranslateExpr}{TE}
\DeclareMathOperator{\TranslateProg}{TranslateProg}
\DeclareMathOperator{\InitMap}{InitMap}

\DeclareMathOperator{\ParemeterNode}{\constructor{param}}
\DeclareMathOperator{\TemporaryNode}{TemporaryNode}
\DeclareMathOperator{\Combine}{Combine}
\DeclareMathOperator{\Keys}{Keys}

\DeclareMathOperator{\THETA}{THETA}
\DeclareMathOperator{\EVAL}{EVAL}
\DeclareMathOperator{\PHI}{PHI}
\newcommand{\params}{\mathit{params}}
\newcommand{\retvar}{\mathtt{retvar}}
\newcommand{\loopdepth}{\ell}
\newcommand{\Stmt}{\mathit{Stmt}}
\newcommand{\Prog}{\mathit{Prog}}
\newcommand{\Expr}{\mathit{Expr}}

\newcommand{\maptype}[2]{#1 \rightarrow #2}
\renewcommand{\maptype}[2]{\mathit{map}[#1,#2]}
\newcommand{\listtype}[1]{{\it list}[#1]}
\renewcommand{\listtype}[1]{#1^*}
\newcommand{\tab}{$\quad$}

\newcommand{\vars}{\mathit{vars}}
\newcommand{\nextnodes}{\mathit{next\_nodes}}
\renewcommand{\nextnodes}{\context'}
\newcommand{\varnodes}{\mathit{tmp\_nodes}}
\renewcommand{\varnodes}{\context_t}
\newcommand{\thetanodes}{\context_\theta}

\newcommand{\op}{\mathit{op}}

\newcommand{\constructor}[1]{\overline{#1}}
\newcommand{\NodeType}{N}
\newcommand{\MustEval}{\text{MustEval}}
\newcommand{\EvalCond}{\text{EvalCond}}

\DeclareMathOperator{\Z}{Z}
\let\S=\undefined
\DeclareMathOperator{\S}{S}

\DeclareMathOperator{\Optimize}{Optimize}
\DeclareMathOperator{\CfgToIr}{ConvertToIR}

\DeclareMathOperator{\CreateInitialEPEG}{CreateInitialEPEG}
\DeclareMathOperator{\IrToCfg}{ConvertToCFG}
\DeclareMathOperator{\Saturate}{Saturate}

\DeclareMathOperator{\SelectBest}{SelectBest}

\DeclareMathOperator{\Match}{Match}

\DeclareMathOperator{\AddEdges}{AddEqualities}

\newcommand{\CFG}{\mathit{CFG}}
\newcommand{\EPEGType}{\mathit{EPEG}}
\newcommand{\PEGType}{\mathit{PEG}}
\newcommand{\cfg}{\mathit{cfg}}
\newcommand{\ir}{\mathit{ir}}
\newcommand{\best}{\mathit{best}}
\newcommand{\eir}{\mathit{ir}}
\newcommand{\satir}{\mathit{saturated\_ir}}

\newcommand{\Trigger}{p}
\newcommand{\Search}{f}
\newcommand{\subst}{\mathit{subst}}
\newcommand{\SubstType}{\mathit{S}}

\newcommand{\invariant}{\mathit{invariant}}

\newcommand{\Var}{V}

\newcommand{\maxvariance}{\mathit{depth}}
\newcommand{\var}{\mathit{var}}
\newcommand{\bcost}{\mathit{basic\_cost}}
\newcommand{\peg}{\mathit{peg}}
\newcommand{\epeg}{\mathit{epeg}}

\newcommand{\body}{\mathit{body}}

\newcommand{\valueit}{\mathit{value}}

\newcommand{\Node}{\mathit{N}}
\newcommand{\Label}{\mathit{L}}
\newcommand{\Param}{\mathit{C}}
\newcommand{\Func}{\mathit{F}}
\newcommand{\DomainFunc}{\mathit{Domain}}
\newcommand{\Primitives}{\mathit{Prims}}
\newcommand{\DomainOperators}{\mathit{DomainOp}}
\newcommand{\ParamLabels}{\mathit{P}}

\newcommand{\linespace}{\vspace{3pt}}
\newcommand{\thickstraightline}{\linespace\hrule height .8pt \linespace}

\newcommand{\declarefunction}[1]{
   \thickstraightline
   \STATE \textbf{function} $#1$
   }
\newcommand{\declaresmallfunction}[1]{
   \STATE \textbf{function} $#1$
   }

\newcommand{\finishfunction}{\vspace{9pt}}

\newcommand{\mypara}[1]{\subsection*{#1}}

\newcommand{\tvsuccess}{98}
\newcommand{\tvnummethods}{3,416\xspace}
\newcommand{\compilemethods}{2,461\xspace}
\newcommand{\heapsize}{200\xspace}
\newcommand{\completionrate}{84}
\newcommand{\enginebound}{500\xspace}

\newcommand{\lowerbound}{$2^{103}$\xspace}

\newcounter{linenumber}
\newcommand{\clearlineno}{\setcounter{linenumber}{0}}
\newcommand{\lineno}{\addtocounter{linenumber}{1} \arabic{linenumber}.}
\newcommand{\savelineno}[1]{\newcounter{#1}\setcounter{#1}{\value{linenumber}}}
\newcommand{\linelabel}[1]{\savelineno{#1}}
\newcommand{\reflineno}[1]{\arabic{#1}}
\newcommand{\startrange}[1]{\savelineno{#1-begin}\addtocounter{#1-begin}{1}}
\newcommand{\finishrange}[1]{\savelineno{#1-end}}
\newcommand{\refrange}[1]{\reflineno{#1-begin} through \reflineno{#1-end}}

\newcommand{\trans}{\rightarrow}
\newcommand{\atrans}[1]{\stackrel{#1}{\trans}}
\newtheorem{theorem}{Theorem}
\newtheorem{lemma}{Lemma}
\newtheorem{defn}{Definition}

\newcommand{\SIMPLE}{SIMPLE\xspace}
\newcommand{\simpleite}[3]{\textbf{if}~(#1)~\{#2\}~\textbf{else}~\{#3\}}
\newcommand{\simplewhile}[2]{\textbf{while}~(#1)~\{#2\}}

\renewcommand{\simpleite}[3]{\texttt{if}~(#1)~\{#2\}~\texttt{else}~\{#3\}}
\renewcommand{\simplewhile}[2]{\texttt{while}~(#1)~\{#2\}}

\newcommand{\typestmnt}[3]{#1 \vdash #2 : #3}
\newcommand{\typeexpr}[3]{#1 \vdash #2 : #3}

\newcommand{\transstmnt}[4][\ell]{#2 \vdash #3 \mapsto_{#1} #4}
\newcommand{\transexpr}[3]{#1 \vdash #2 \mapsto #3}
\renewcommand{\transstmnt}[4][\ell]{#2 \leadsto_{#1} #4}
\renewcommand{\transexpr}[3]{#1 \leadsto #3}

\newcommand{\typetransstmnt}[6][\ell]{\typestmnt{#3}{#2}{#4} \; \Longrightarrow \; \transstmnt[#1]{#5}{#2}{#6}}
\newcommand{\typetransexpr}[5]{\typeexpr{#2}{#1}{#3} \; \Longrightarrow \; \transexpr{#4}{#1}{#5}}
\renewcommand{\typetransstmnt}[6][\ell]{\typestmnt{#3}{#2}{#4} \;\; \vartriangleright \;\; \transstmnt[#1]{#5}{#2}{#6}}
\renewcommand{\typetransexpr}[5]{\typeexpr{#2}{#1}{#3} \;\; \vartriangleright \;\; \transexpr{#4}{#1}{#5}}

\def\doi{7 (1:10) 2011}
\lmcsheading%
{\doi}
{1--37}
{}
{}
{Oct.~12, 2009}
{Mar.~28, 2011}
{}

\begin{document}

\title{Equality Saturation: a New Approach to Optimization}

\thanks{Supported in part by NSF CAREER grant CCF-0644306.}

\author[R.~Tate]{Ross Tate}
\author[M.~Stepp]{Michael Stepp}
\author[Z.~Tatlock]{Zachary Tatlock}
\author[S.~Lerner]{Sorin Lerner}
\address{Department of Computer Science and Engineering,
         University of California, San Diego}
\email{\{rtate,mstepp,ztatlock,lerner\}@cs.ucsd.edu}

\keywords{Compiler Optimization, Equality Reasoning, Intermediate Representation}
\subjclass{D.3.4}
\titlecomment{An earlier version of this work appeared at the 36th
Annual ACM SIGPLAN - SIGACT Symposium on Principles of Programming
Languages (POPL 2009)}

\begin{abstract}

Optimizations in a traditional compiler are applied sequentially, with
each optimization destructively modifying the program to produce a
transformed program that is then passed to the next optimization. We
present a new approach for structuring the optimization phase of a
compiler. In our approach, optimizations take the form of equality
analyses that add equality information to a common intermediate
representation. The optimizer works by repeatedly applying these
analyses to infer equivalences between program fragments, thus
saturating the intermediate representation with equalities. Once
saturated, the intermediate representation encodes multiple optimized
versions of the input program. At this point, a profitability
heuristic picks the final optimized program from the various programs
represented in the saturated representation. Our proposed way of
structuring optimizers has a variety of benefits over previous
approaches: our approach obviates the need to worry about optimization
ordering, enables the use of a global optimization heuristic that
selects among fully optimized programs, and can be used to perform
translation validation, even on compilers other than our own. We
present our approach, formalize it, and describe our choice of
intermediate representation. We also present experimental results
showing that our approach is practical in terms of time and space
overhead, is effective at discovering intricate optimization
opportunities, and is effective at performing translation validation
for a realistic optimizer.

\end{abstract}

\maketitle

\section*{Introduction}

In a traditional compilation system, optimizations are applied
sequentially, with each optimization taking as input the program
produced by the previous one. This traditional approach to compilation
has several well-known drawbacks. One of these drawbacks is that the
order in which optimizations are run affects the quality of the
generated code, a problem commonly known as the \emph{phase ordering
problem}. Another drawback is that profitability heuristics, which
decide whether or not to apply a given optimization, tend to make
their decisions one optimization at a time, and so it is difficult for
these heuristics to account for the effect of future transformations.

In this paper, we present a new approach for structuring optimizers
that addresses the above limitations of the traditional approach, and
also has a variety of other benefits. Our approach consists of
computing a set of optimized versions of the input program and then
selecting the best candidate from this set. The set of candidate
optimized programs is computed by repeatedly inferring equivalences
between program fragments, thus allowing us to represent the effect of
many possible optimizations at once. This, in turn, enables the
compiler to delay the decision of whether or not an optimization is
profitable until it observes the full ramifications of that decision.
Although related ideas have been explored in the context of
super-optimizers, as Section~\ref{sec:rel} on related work will point
out, super-optimizers typically operate on straight-line code, whereas
our approach is meant as a general-purpose compilation paradigm that
can optimize complicated control flow structures.

At its core, our approach is based on a simple change to the
traditional compilation model: whereas traditional optimizations
operate by destructively performing transformations, in our approach
optimizations take the form of \emph{equality analyses} that simply
add equality information to a common intermediate representation (IR),
without losing the original program. Thus, after each
equality analysis runs, both the old program and the new program are
represented.

The simplest form of equality analysis looks for ways to instantiate
equality axioms like $a * 0 = 0$, or $a * 4 = a~\text{\tt
<<}~2$. However, our approach also supports arbitrarily complicated
forms of equality analyses, such as inlining, tail recursion
elimination, and various forms of user defined axioms. The flexibility
with which equality analyses are defined makes it easy for compiler
writers to port their traditional optimizations to our equality-based
model: optimizations can work as before, except that when the
optimization would have performed a transformation, it now simply
records the transformation as an equality.

The main technical challenge that we face in our approach is that the
compiler's IR must now use equality information to represent not just
one optimized version of the input program, but multiple versions at
once. We address this challenge through a new IR that compactly
represents equality information, and as a result can simultaneously
store multiple optimized versions of the input program. After a
program is converted into our IR, we repeatedly apply equality
analyses to infer new equalities until no more equalities can be
inferred, a process known as saturation. Once saturated with
equalities, our IR compactly represents the various possible ways of
computing the values from the original program modulo the given set of
equality analyses (and modulo some bound in the case where applying
equality analyses leads to unbounded expansion).

Our approach of having optimizations add equality information to a
common IR until it is saturated with equalities has a variety of
benefits over previous optimization models.

\mypara{Optimization order does not matter.} The first benefit of our
approach is that it removes the need to think about optimization
ordering. When applying optimizations sequentially, ordering is a
problem because one optimization, say $A$, may perform some
transformation that will irrevocably prevent another optimization, say
$B$, from triggering, when in fact running $B$ first would have
produced the better outcome. This so-called \emph{phase ordering
problem} is ubiquitous in compiler design. In our approach, however,
the compiler writer does not need to worry about ordering, because
optimizations do not destructively update the program -- they simply
add equality information. Therefore, after an optimization $A$ is
applied, the original program is still represented (along with the
transformed program), and so any optimization $B$ that could have been
applied before $A$ is still applicable after $A$. Thus, there is no
way that applying an optimization $A$ can irrevocably prevent another
optimization $B$ from applying, and so there is no way that applying
optimizations will lead the search astray. As a result, compiler
writers who use our approach do not need to worry about the order in
which optimizations run. Better yet, because optimizations are allowed
to freely interact during equality saturation, without any
consideration for ordering, our approach can discover intricate
optimization opportunities that compiler writers may not have
anticipated, and hence would not have implemented in a general purpose
compiler.

\mypara{Global profitability heuristics.} The second benefit of our
approach is that it enables \emph{global profitability
heuristics}. Even if there existed a perfect order to run
optimizations in, compiler writers would still have to design
profitability heuristics for determining whether or not to perform
certain optimizations such as inlining. Unfortunately, in a
traditional compilation system where optimizations are applied
sequentially, each heuristic decides in isolation whether or not to
apply an optimization at a particular point in the compilation
process. The local nature of these heuristics makes it difficult to
take into account the effect of future optimizations.

Our approach, on the other hand, allows the compiler writer to design
profitability heuristics that are global in nature. In particular,
rather than choosing whether or not to apply an optimization locally,
these heuristics choose between fully optimized versions of the input
program. Our approach makes this possible by separating the decision
of whether or not a transformation is \emph{applicable} from the
decision of whether or not it is \emph{profitable}. Indeed, using an
optimization to add an equality in our approach does not indicate a
decision to perform the transformation -- the added equality just
represents the \emph{option} of picking that transformation later. The
actual decision of which transformations to apply is performed by a
global heuristic \emph{after} our IR has been saturated with
equalities. This global heuristic simply chooses among the various
optimized versions of the input program that are represented in the
saturated IR, and so it has a global view of all the transformations
that were tried and what programs they generated.

There are many ways to implement this global profitability heuristic,
and in our prototype compiler we have chosen to implement it using a
Pseudo-Boolean solver (a form of Integer Linear Programming
solver). In particular, after our IR has been saturated with
equalities, we use a Pseudo-Boolean solver and a static cost model for
every node to pick the lowest-cost program that computes the same
result as the original program.

\mypara{Translation validation.} The third benefit of our approach is
that it can be used not only to optimize programs, but also to prove
programs equivalent: intuitively, if during saturation an equality
analysis finds that the return values of two programs are equal, then
the two programs are equivalent. Our approach can therefore be used to
perform \emph{translation validation}, a technique that consists of
automatically checking whether or not the optimized version of an
input program is semantically equivalent to the original program. For
example, we can prove the correctness of optimizations performed by
existing compilers, even if our profitability heuristic would not have
selected those optimizations.

\mypara{Contributions.} The contributions of this paper can therefore
be summarized as follows:

\begin{enumerate}[$\bullet$]

\item We present a new approach for structuring optimizers. In our
approach optimizations add equality information to a common IR that
simultaneously represents multiple optimized versions of the input
program. Our approach obviates the need to worry about optimization
ordering, it enables the use of a global optimization heuristic (such
as a Pseudo-Boolean solver), and it can be used to perform translation
validation for any compiler.
Sections~\ref{sec:overview} and~\ref{sec:loops} present an overview of
our approach and its capabilities, Section~\ref{sec:formal} makes our
approach formal, and Sections~\ref{sec:semantics}
through~\ref{sec:peg2cfg} describe the new IR that allows our
approach to be effective.

\item We have instantiated our approach in a new Java bytecode
  optimizer called
\Peggy (Section~\ref{sec:implementation}). \Peggy uses our approach
not only to optimize Java methods, but also to perform translation
validation for existing compilers. Our experimental results
(Section~\ref{sec:eval}) show that our approach (1) is practical both
in terms of time and space overhead, (2) is effective at discovering
both simple and intricate optimization opportunities and (3) is
effective at performing translation validation for a realistic
optimizer -- \Peggy is able to validate \tvsuccess\% of the runs of
the Soot optimizer~\cite{vall99soot}, and within the remaining 2\% it
uncovered a bug in Soot.

\end{enumerate}

\section{Overview}
\label{sec:overview}

Our approach for structuring optimizers is based on the idea of having
optimizations propagate equality information to a common IR that
simultaneously represents multiple optimized versions of the input
program. The main challenge in designing this IR is that it must make
equality reasoning \emph{effective} and \emph{efficient}.

To make equality reasoning \emph{effective}, our IR needs to support
the same kind of basic reasoning that one would expect from simple
equality axioms like \hbox{$a*(b+c) = a*b + a*c$}, but with more complicated
computations such as branches and loops. We have designed a
representation for computations called Program Expression Graphs
(\PEGs) that meets these requirements. Similar to the \emph{gated SSA}
representation~\cite{GSSA, ThinGSSA}, \PEGs are \emph{referentially
  transparent}, which intuitively means that the value of an
expression depends only on the value of its constituent expressions,
without any side-effects. As has been observed previously in many
contexts, referential transparency makes equality reasoning simple and
effective. However, unlike previous SSA-based representations, \PEGs
are also \emph{complete}, which means that there is no need to
maintain any additional representation such as a control flow graph
(CFG). Completeness makes it easy to use equality for performing
transformations: if two \PEG nodes are equal, then we can pick either
one to create a program that computes the same result, without
worrying about the implications on any underlying representation.

In addition to being effective, equality reasoning in our IR must be
\emph{efficient}. The main challenge is that each added equality can
potentially double the number of represented programs, thus making the
number of represented programs exponential in the worst case. To
address this challenge, we record equality information of PEG nodes by
simply merging PEG nodes into equivalence classes. We call the
resulting equivalence graph an \EPEG, and it is this \EPEG
representation that we use in our approach. Using equivalence classes
allows \EPEGs to efficiently represent exponentially many ways of
expressing the input program, and it also allows the equality
saturation engine to efficiently take into account previously
discovered equalities. Among existing IRs, \EPEGs are unique in their
ability to represent multiple optimized versions of the input
program. A more detailed discussion of how \PEGs and \EPEGs relate to
previous IRs can be found in Section~\ref{sec:rel}.

We illustrate the main features of our approach by showing how it can
be used to implement loop-induction-variable strength reduction. The
idea behind this optimization is that if all assignments to a variable
\verb-i- in a loop are increments, then an expression \verb-i * c- in
the loop (with \verb-c- being loop invariant) can be replaced with
\verb-i-, provided all the increments of \verb-i- in the loop are
appropriately scaled by \verb-c-.

\begin{figure}
\begin{minipage}[t]{1.8in}
\begin{verbatim}
  i := 0;
  while (...) {
     use(i * 5);
     i := i + 1;
     if (...) {
        i := i + 3;
     }
  }
\end{verbatim}
\end{minipage}
\begin{minipage}[t]{1.3in}
\begin{verbatim}
i := 0;
while (...) {
   use(i);
   i := i + 5;
   if (...) {
      i := i + 15;
   }
}
\end{verbatim}
\end{minipage} \\
\begin{minipage}[t]{1.8in}
\begin{center}
(a)
\end{center}
\end{minipage}
\begin{minipage}[t]{1.3in}
\begin{center}
(b)
\end{center}
\end{minipage}
\caption{Loop-induction-variable strength reduction: (a) shows
  the original code, and (b) shows the optimized code.}
\label{fig:strength-reduction-code}
\end{figure}

As an example, consider the code snippet from
Figure~\ref{fig:strength-reduction-code}(a). The use of
\verb-i*5- inside the loop can be replaced with \verb-i- as long as
the two increments in the loop are scaled by \verb-5-. The resulting
code is shown in Figure~\ref{fig:strength-reduction-code}(b).

\subsection{Program Expression Graphs}

\begin{figure}[!t]
\begin{center}
\includegraphics[width=6in,clip]{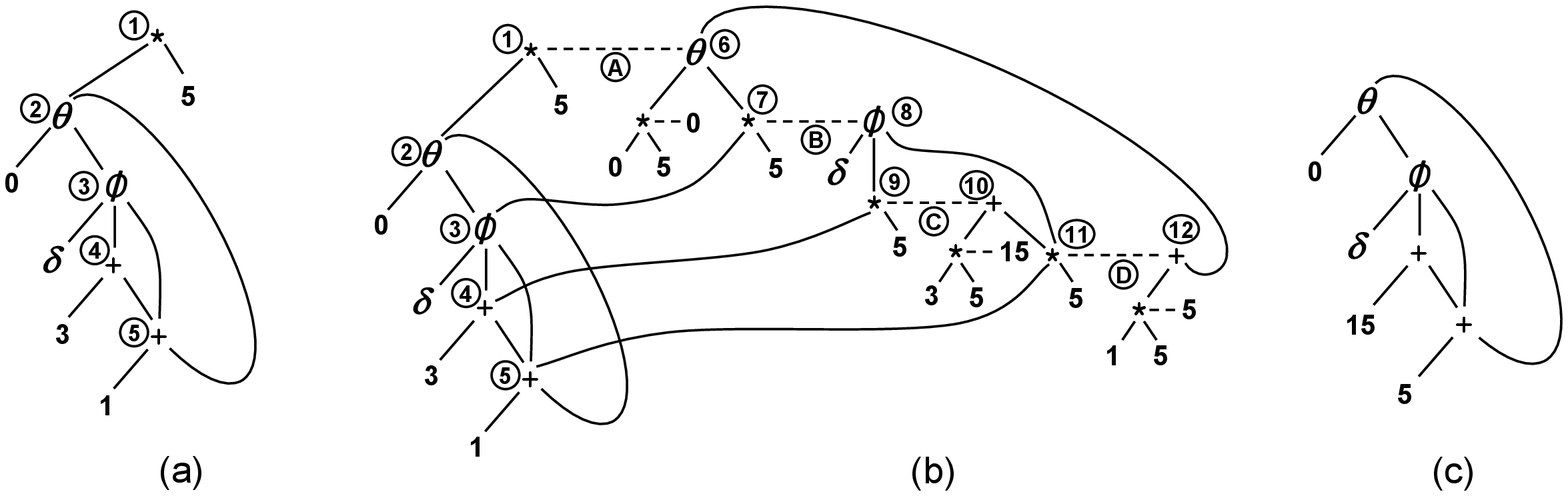}
\end{center}
\caption{Loop-induction-variable Strength Reduction using \PEGs: 
  (a) shows the original \PEG, (b) shows the \EPEG that our
  engine produces from the original \PEG and (c) shows the
  optimized \PEG, which results by choosing nodes 6,
  8, 10, and 12 from (b).}
\label{fig:strength-reduction-veg}
\end{figure}

A Program Expression Graph (\PEG) is a graph containing: (1) operator
nodes, for example ``plus'', ``minus'', or any of our built-in nodes
for representing conditionals and loops, and (2) ``dataflow'' edges
that specify where operator nodes get their arguments from.  As an
example, consider the ``use'' statement in
Figure~\ref{fig:strength-reduction-code}(a). This is meant as a
placeholder for any kind of use of the value \texttt{i*5}; it is used
to mark the specific location inside the loop where we examine this
value. The \PEG for the value \texttt{i*5} is shown in
Figure~\ref{fig:strength-reduction-veg}(a).  At the very top of the
\PEG we see node 1, which represents the \verb-i*5- multiply operation
from inside the loop. Each \PEG node represents an operation, with the
children nodes being the arguments to the operation. The links from
parents to children are shown using solid (non-dashed) lines. For
example, node 1 represents the multiplication of node 2 by the
constant 5. \PEGs follow the notational convention used in
E-graphs~\cite{Nelson:1979:SCD, NelsonOpenCongruence80, simplify} and
Abstract Syntax Trees (ASTs) of displaying operators above the
arguments that flow into them, which is the opposite convention
typically used in Dataflow Graphs~\cite{ssa,ZadeckVariableEquality}.
We use the E-graph/AST orientation because we think of \PEGs as
recursive expressions.

Node 2 in our \PEG represents the value of variable \verb-i- inside
the loop, right before the first instruction in the loop is
executed. We use $\theta$ nodes to represent values that vary inside
of a loop. A \PEG contains one $\theta$ node per variable that is live
in the loop, and a variable's $\theta$ node represents the entire
sequence of values that the variable takes throughout the
loop. Intuitively, the left child of a $\theta$ node computes the
initial value, whereas the right child computes the value at the
current iteration in terms of the value at the previous iteration. In
our example, the left child of the $\theta$ node is the constant 0,
representing the initial value of \verb-i-. The right child of the
$\theta$ node uses nodes 3, 4, and 5 to compute the value of \verb-i-
at the current iteration in terms of the value of \verb-i- from the
previous iteration. The two plus nodes (nodes 4 and 5) represent the
two increments of \verb-i- in the loop, whereas the $\phi$ node (node
3) represents the merging of the two values of \verb-i- produced by
the two plus nodes. In traditional SSA, a $\phi$ node has only two
inputs (the true value and the false value) and as a result the node
itself does not know which of the two inputs to select, relying
instead on an explicit control-flow join to know at run-time which
case of the branch was taken. In contrast, our $\phi$ nodes are like
those in \emph{gated} SSA~\cite{GSSA, ThinGSSA}: they take an
additional parameter (the first left-most one) which is used to select
between the second and the third parameter. As a result, our $\phi$
nodes are executable by themselves, and so there is no need to
explicitly encode a control-flow join. Our example doesn't use the
branch condition in an interesting way, and so we just let $\delta$
represent the \PEG sub-graph that computes the branch condition.
Furthermore, since this \PEG represents the value of \texttt{i} {\em
inside} the loop, it does not contain any operators to describe the
\texttt{while}-condition, since this information is only relevant for
computing the value of \texttt{i} after the loop has terminated. We
show how to express the value of variables after a loop in
Section~\ref{sec:loops}.

From a more formal point of view, each $\theta$ node produces a
\emph{sequence} of values, one value for each iteration of the
loop. The first argument of a $\theta$ node is the value for the first
iteration, whereas the second argument is a sequence that represents
the values for the remaining iterations. For example, in
Figure~\ref{fig:strength-reduction-veg}, the nodes labeled 3 through 5
compute this sequence of remaining values in terms of the sequence
produced by the $\theta$ node. In particular, nodes 3, 4 and 5 have
been implicitly lifted to operate on this sequence. The fact that a
single $\theta$ node represents the entire sequence of values that a
loop produces allows us to represent that two loops compute the same
sequence of values with a single equality between two $\theta$ nodes.

\PEGs are well-suited for equality reasoning because all \PEG
operators, even those for branches and loops, are mathematical
functions with no side effects. As a result, \PEGs are
\emph{referentially transparent}, which allows us to perform the same
kind of equality reasoning that one is familiar with from
mathematics. Though \PEGs are related to functional programs, in our
work we have used \PEGs to represent intra-procedural imperative code
with branches and looping constructs. Furthermore, even though all
\PEG operators are pure, \PEGs can still represent programs with state
by using heap summary nodes: stateful operations, such as heap reads
and writes, can take a heap as an argument and return a new heap. This
functional representation of stateful programs allows our \Peggy
compiler to use \PEGs to reason about Java programs. The heap summary
node can also be used to encode method/function calls in an
intra-procedural setting by simply threading the heap summary node
through special nodes representing method/function calls. There is
however one big challenge with heap summary nodes, which we have not
yet fully addressed yet. Although in the \PEG domain, heap summary
nodes can be reasoned about as if the heap can be duplicated, when a
\PEG is converted back to a CFG, heap summary nodes must obey a
linear-typing discipline. We have developed a simple constraint
solving technique for finding a linearization of heap operations in a
\PEG, but this technique is not complete (as in, even if there is a
linearization, we are not guaranteed to find it). In our \Peggy
implementation, after optimizations have been applied, this
incompleteness affects 3\% of Java methods (in which case we do not
optimize the method). Section~\ref{sec:implementation} explains in
more detail how we represent several features of Java programs in
\PEGs (including the heap and method calls) and what the issues are
with our linearization incompleteness. We also present some ideas on
how to address this incompleteness in future work.

\subsection{Encoding equalities using \EPEGs}
\label{sec:reasoning-with-a-veg}

A \PEG by itself can only represent a single way of expressing the
input program. To represent \emph{multiple} optimized versions of the
input program, we need to encode equalities in our representation. To
this end, an \EPEG is a graph that groups together \PEG nodes that are
equal into equivalence classes. As an example,
Figure~\ref{fig:strength-reduction-veg}(b) shows the \EPEG that our
engine produces from the \PEG of
Figure~\ref{fig:strength-reduction-veg}(a). We display equalities
graphically by adding a dashed edge between two nodes that have become
equal. These dashed edges are only a visualization mechanism. In
reality, \PEG nodes that are equal are grouped together into an
equivalence class.

Reasoning in an \EPEG is done through the application of
optimizations, which in our approach take the form of equality
analyses that add equality information to the \EPEG. An equality
analysis consists of two components: a trigger, which is an expression
pattern stating the kinds of expressions that the analysis is
interested in, and a callback function, which should be invoked when
the trigger pattern is found in the \EPEG. The saturation engine
continuously monitors all the triggers simultaneously, and invokes the
necessary callbacks when triggers match. When invoked, a callback
function adds the appropriate equalities to the \EPEG.

The simplest form of equality analysis consists of instantiating
axioms such as $a*0 = 0$. In this case, the trigger would be $a*0$,
and the callback function would add the equality $a*0=0$. Even though
the vast majority of our reasoning is done through such declarative
axiom application, our trigger and callback mechanism is much more
general, and has allowed us to implement equality analyses such as
inlining, tail-recursion elimination, and constant folding.

The following three axioms are the equality analyses required to perform
loop-induction-variable strength reduction. They state that
multiplication distributes over addition, $\theta$, and $\phi$:
\begin{align}
(a + b) * m & = a * m + b * m \label{mult-distr-plus} \\
\theta(a,b) * m & = \theta(a * m, b * m) \label{mult-distr-theta} \\
\phi(a,b,c) * m& = \phi(a, b * m, c * m) \label{mult-distr-phi} 
\end{align}

After a program is converted to a \PEG, a saturation engine repeatedly
applies equality analyses until either no more equalities can be
added, or a bound is reached on the number of expressions that have
been processed by the engine. As Section~\ref{sec:eval} will describe
in more detail, our experiments show that \completionrate\% of
methods can be completely saturated, without any bounds being imposed.

Figure~\ref{fig:strength-reduction-veg}(b) shows the
saturated \EPEG that results from applying the above distributivity
axioms, along with a simple constant folding equality analysis. In
particular, distributivity is applied four times:
axiom~\eqref{mult-distr-theta} adds equality edge A,
axiom~\eqref{mult-distr-phi} edge B, axiom~\eqref{mult-distr-plus}
edge C, and axiom~\eqref{mult-distr-plus} edge D. Our engine also
applies the constant folding equality analysis to show that $0 * 5 =
0$, $3 * 5 = 15$ and $1 * 5 = 5$. Note that when
axiom~\eqref{mult-distr-theta} adds edge A, it also adds node 7, which
then enables axiom~\eqref{mult-distr-phi}. Thus, equality analyses
essentially communicate with each other by propagating equalities
through the \EPEG. Furthermore, note that the instantiation of
axiom~\eqref{mult-distr-plus} adds node 12 to the \EPEG, but it does
not add the right child of node 12, namely $\theta(\ldots) * 5$,
because it is already represented in the \EPEG.

Once saturated with equalities, an \EPEG compactly represents multiple
optimized versions of the input program -- in fact, it compactly
represents all the programs that could result from applying the
optimizations in any order to the input program. For example, the
\EPEG in Figure~\ref{fig:strength-reduction-veg}(b) encodes 128 ways
of expressing the original program (because it encodes 7 independent
equalities, namely the 7 dashed edges). In general, a single \EPEG can
efficiently represent exponentially many ways of expressing the input
program.

After saturation, a global profitability heuristic can pick which
optimized version of the input program is best. Because this
profitability heuristic can inspect the entire \EPEG at once, it has a global
view of the programs produced by various optimizations, \emph{after}
all other optimizations were also run. In our example, starting at
node 1, by choosing nodes 6, 8, 10, and 12, we can construct the graph
in Figure~\ref{fig:strength-reduction-veg}(c), which corresponds
exactly to performing loop-induction-variable strength reduction in
Figure~\ref{fig:strength-reduction-code}(b).

More generally, when optimizing an entire function, one has to pick a
node for the equivalence class of the return value and nodes for all
equivalence classes that the return value depends on. There are many
plausible heuristics for choosing nodes in an \EPEG. In our \Peggy
implementation, we have chosen to select nodes using a Pseudo-Boolean
solver, which is an Integer Linear Programming solver where variables
are constrained to 0 or 1. In particular, we use a Pseudo-Boolean
solver and a static cost model for every node to compute the
lowest-cost program that is encoded in the \EPEG. In the example from
Figure~\ref{fig:strength-reduction-veg}, the Pseudo-Boolean solver
picks the nodes described above. Section~\ref{sec:heuristic} describes
our technique for selecting nodes in more detail.\vfill\eject

\subsection{Benefits of our approach}

\mypara{Optimization order does not matter.} To understand how our
approach addresses the phase ordering problem, consider a simple
peephole optimization that transforms \verb-i * 5- into
\verb-i << 2 + i-. On the surface, one may think that this
transformation should always be performed if it is applicable -- after
all, it replaces a multiplication with the much cheaper shift and
add. In reality, however, this peephole optimization may disable other
more profitable transformations. The code from
Figure~\ref{fig:strength-reduction-code}(a) is such an example:
transforming \verb-i * 5- to \verb-i << 2 + i- disables
loop-induction-variable strength reduction, and therefore
generates code that is worse than the one from
Figure~\ref{fig:strength-reduction-code}(b).

The above example illustrates the ubiquitous \emph{phase ordering
problem}. In systems that apply optimizations sequentially, the
quality of the generated code depends on the order in which
optimizations are applied. Whitfield and Soffa~\cite{WhitfieldSoffa97}
have shown experimentally that enabling and disabling interactions
between optimizations occur frequently in practice, and furthermore
that the patterns of interaction vary not only from program to
program, but also within a single program. Thus, no one order is best
across all compilation.

A common partial solution consists of carefully considering all the
possible interactions between optimizations, possibly with the help of
automated tools, and then coming up with a carefully tuned sequence
for running optimizations that strives to enable most of the
beneficial interactions.  This technique, however, puts a heavy burden
on the compiler writer, and it also does not account for the fact that
the best order may vary between programs.

At high levels of optimizations, some compilers may even run
optimizations in a loop until no more changes can be made. Even so, if
the compiler picks the wrong optimization to start with, then no
matter what optimizations are applied later, in any order, any number
of times, the compiler will not be able to reverse the disabling
consequences of the first optimization.

In our approach, the compiler writer does not need to worry about the
order in which optimizations are applied. The previous peephole
optimization would be expressed as the axiom
\verb-i * 5 = i << 2 + i-. However, unlike in a traditional
compilation system, applying this axiom in our approach does not
remove the original program from the representation --- it only adds
information --- and so it cannot disable other
optimizations. Therefore, the code from
Figure~\ref{fig:strength-reduction-code}(b) would still be discovered,
even if the peephole optimization was run first. In essence, our
approach is able to simultaneously explore all possible sequences of
optimizations, while sharing work that is common across the various
sequences.

In addition to reducing the burden on compiler writers, removing the
need to think about optimization ordering has two additional
benefits. First, because optimizations interact freely with no regard
to order, our approach often ends up combining optimizations in
unanticipated ways, leading to surprisingly complicated optimizations
given how simple our equality analyses are --- Section~\ref{sec:loops}
gives such an example. Second, it makes it easier for end-user
programmers to add domain-specific axioms to the compiler, because
they don't have to think about where exactly in the compiler the axiom
should be run, and in what order relative to other optimizations.

\mypara{Global profitability heuristics.} Profitability heuristics in
traditional compilers tend to be local in nature, making it difficult to
take into account the effect of future optimizations. For example,
consider inlining. Although it is straightforward to estimate the
\emph{direct cost} of inlining (the code-size increase) and the
\emph{direct benefit} of inlining (the savings from removing the call
overhead), it is far more difficult to estimate the potentially larger
\emph{indirect benefit}, namely the additional optimization
opportunities that inlining exposes.

To see how inlining would affect our running example, consider again
the code from Figure~\ref{fig:strength-reduction-code}(a), but assume
that instead of \verb-use(i * 5)-, there was a call to a function
\verb-f-, and the use of \verb-i*5- occurred \emph{inside}
\verb-f-. If \verb-f- is sufficiently large, a traditional inliner
would not inline \verb-f-, because the code bloat would outweigh the
call-overhead savings. However, a traditional inliner would miss the
fact that it may still be worth inlining \verb-f-, despite its size,
because inlining would expose the opportunity for
loop-induction-variable strength reduction. One solution to this
problem consists of performing an \emph{inlining
trial}~\cite{inlining-trials}, where the compiler simulates the
inlining transformation, along with the effect of subsequent
optimizations, in order to decide whether or not to actually
inline. However, in the face of multiple inlining decisions (or more
generally multiple optimization decisions), there can be exponentially
many possible outcomes, each one of which has to be compiled
separately.

In our approach, on the other hand, inlining simply adds an equality
to the \EPEG stating that the call to a given function is equal to its
body instantiated with the actual arguments. The resulting \EPEG
simultaneously represents the program where inlining is performed and
where it is not. Subsequent optimizations then operate on both of
these programs at the same time.  More generally, our approach can
simultaneously explore exponentially many possibilities in parallel,
while sharing the work that is redundant across these various
possibilities. In the above example with inlining, once the \EPEG is
saturated, a global profitability heuristic can make a more informed
decision as to whether or not to pick the inlined version, since it
will be able to take into account the fact that inlining enabled
loop-induction-variable strength reduction.

\mypara{Translation Validation.} Unlike traditional compilation
frameworks, our approach can be used not only to optimize programs,
but also to establish equivalences between programs. In particular, if we
convert two programs into an \EPEG, and then saturate it with
equalities, then we can conclude that the two programs are equivalent
if they belong to the same equivalence class in the saturated
\EPEG. In this way, our approach can be used to perform translation
validation for any compiler (not just our own), by checking
that each function in the input program is equivalent to the
corresponding optimized function in the output program.

For example, our approach would be able to show that the two program
fragments from Figure~\ref{fig:strength-reduction-code} are
equivalent. Furthermore, it would also be able to validate a
compilation run in which \verb-i * 5 = i << 2 + i- was applied first
to Figure~\ref{fig:strength-reduction-code}(a). This shows that we are
able to perform translation validation regardless of what optimized
program our own profitability heuristic would choose.

Although our translation validation technique is intraprocedural, we
can use interprocedural equality analyses such as inlining to enable a
certain amount of interprocedural reasoning. This allows us to reason
about transformations like reordering function calls.

\section{Reasoning about loops}
\label{sec:loops}

This section shows how our approach can be used to reason across
nested loops. The example highlights the fact that a simple axiom set
can produce unanticipated optimizations which traditional
compilers would have to explicitly search for.

We start in Sections~\ref{sec:single-loop} and~\ref{sec:nested-loops}
by describing all \PEG constructs used to represent loops. We then
show in Section~\ref{sec:inter-loop-strength-reduction} how our
approach can perform an inter-loop strength reduction optimization.

\subsection{Single loop}
\label{sec:single-loop}

\begin{figure}[t]
\begin{center}
\includegraphics[width=6.0in]{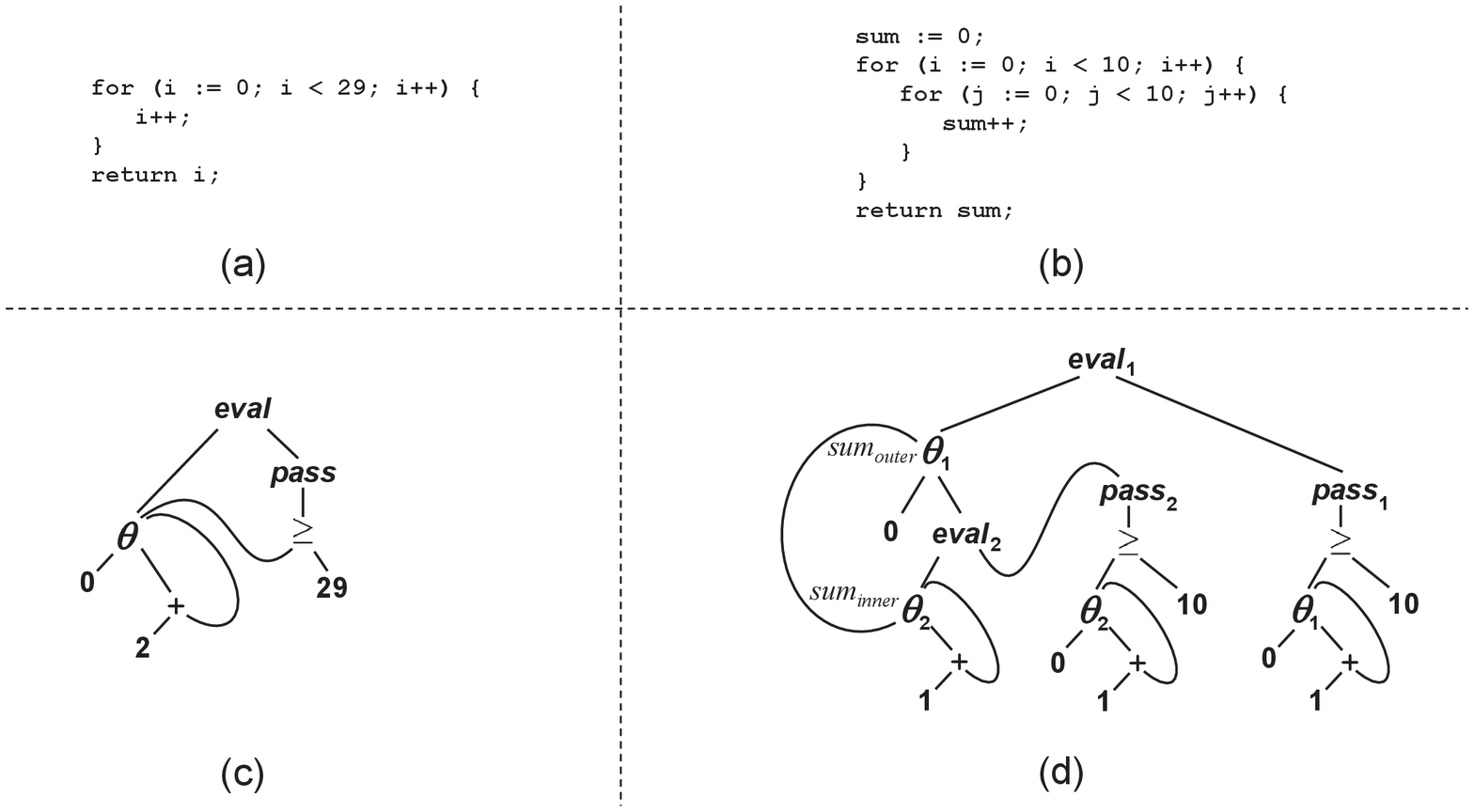}
\end{center}
\caption{Two loops and their \PEG representations.}
\label{fig:loop-single-double}
\end{figure}

Consider the simple loop from Figure~\ref{fig:loop-single-double}(a).
This loop iterates 15 times, incrementing the value of \verb-i- each time by 2. 
The final value of \verb-i- is then returned at the end of the function.
The \PEG for this code is shown in Figure~\ref{fig:loop-single-double}(c).
The value of \verb-i- inside the loop is represented by a $\theta$ node.
Intuitively, this $\theta$ node produces the sequence of values that \verb-i- takes throughout the loop, in this case $[0, 2, 4, \ldots]$.
The value of $\verb-i-$ after the loop is represented by the $\eval$ node at the top of the \PEG.
Given a sequence $s$ and an index $n$, $\eval(s,n)$ produces the $n^{\it th}$ element of sequence $s$.
To determine which element to select from a sequence, our \PEG representation uses $\pass$ nodes.
Given a sequence $s$ of booleans, $\pass(s)$ returns the index of the first element in the sequence that is true.
In our example, the $\geq$ node uses the result of the $\theta$ node to produce the sequence of values taken on by the boolean expression ${\tt i} \geq 29$ throughout the loop.
This sequence is then sent to $\pass$, which in this case produces the value $15$, since the $15^{\it th}$ value (counting from 0) of \verb-i- in the loop (which is $30$) is the first one to make ${\tt i} \geq 29$ true.
The $\eval$ node then selects the $15^{\it th}$ element of the sequence produced by the $\theta$ node, which is $30$.
In our previous example from Section~\ref{sec:overview}, we omitted $\eval$/$\pass$ from the \PEG for clarity -- because we were not interested in any of the values after the loop, the $\eval$/$\pass$ nodes would not have been used in any reasoning.

Note that every loop-varying value will be represented by its own
$\theta$ node, and so there will be one $\theta$ node in the \PEG per
live variable in the loop. Also, every variable that is live after the
loop has its own $\eval$ node, which represents the value after the
loop. However, there is only one $\pass$ node per loop, which
represents the iteration at which the loop terminates. Thus, there can
be many $\theta$ and $\eval$ nodes per loop, but only one $\pass$
node.

Since the $\eval$ and $\pass$ operators are often paired together, it
is natural to consider merging them into a single operator.  However,
we have found that the separation is useful.  For one, although there
will be many $\eval$ nodes corresponding to a single loop, each loop
has only one corresponding $\pass$ node.  Having this single node to
represent each loop is useful in many of the compilation stages for
\PEGs.  Also, $\pass$ nodes are not the only nodes we will use as the
second argument to an $\eval$ node. For example, to accomplish loop
peeling (as shown in Section~\ref{sec:loop-peeling}) we use $\phi$
nodes and other special-purpose nodes as the second
argument. Furthermore, Section~\ref{evalpass} will present a more
detailed reflection on our design choice after we have shown how the
$\eval$ and $\pass$ operators are used in our various compilation
stages.

\subsection{Nested loops}
\label{sec:nested-loops}

We now illustrate, through an example, how nested loops can be encoded in our \PEG representation.
Consider the code snippet from Figure~\ref{fig:loop-single-double}(b), which has two nested loops.
The \PEG for this code snippet is shown in Figure~\ref{fig:loop-single-double}(d).
Each $\theta$, $\eval$ and $\pass$ node is labeled with a subscript indicating what loop depth it operates on (we previously omitted these subscripts for clarity).
The topmost $\eval_1$ node represents the final value of \verb-sum-.
The node labeled $\suminner$ represents the value of \verb-sum- at the beginning of the inner loop body.
Similarly, $\sumouter$ is the value of \verb-sum- at the beginning of the outer loop body.
Looking at $\suminner$, we can see that: (1) on the first iteration (the left child of $\suminner$), $\suminner$ gets the value of \verb-sum- from the outer loop; (2) on other iterations, it gets one plus the value of \verb-sum- from the previous iteration of the inner loop.
Looking at $\sumouter$, we can see that: (1) on the first iteration, $\sumouter$ gets $0$; (2) on other iterations, it gets the value of \verb-sum-
right after the inner loop terminates.
The value of \verb-sum- after the inner loop terminates is computed using a similar $\eval$/$\pass$ pattern as in Figure~\ref{fig:loop-single-double}(c), as is the value of \verb-sum- after the outer loop terminates.

\subsection{Inter-loop strength reduction}
\label{sec:inter-loop-strength-reduction}

\begin{figure}[t]
\begin{center}
\includegraphics[width=6.0in]{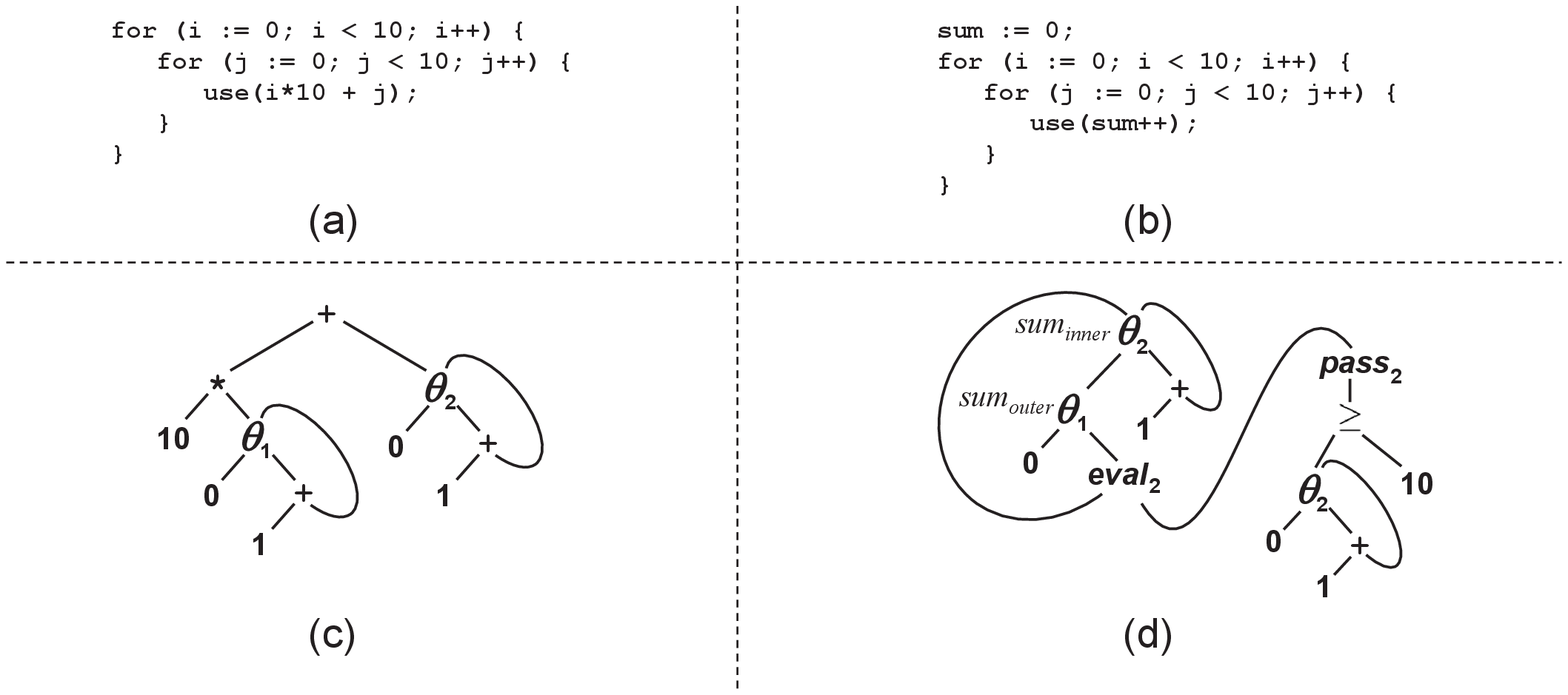} 
\end{center}
\caption{Two equivalent loops and their \PEG representations. The
  \PEGs for the expressions inside the {\tt use} statements in (a)
  and (b) are shown in (c) and (d), respectively.}
\label{fig:loop-nested-ilsr}
\end{figure}

Our approach allows an optimizing compiler to perform intricate optimizations of looping structures.
We present such an example here, with a kind of inter-loop strength reduction.
Consider the code snippets from Figure~\ref{fig:loop-nested-ilsr}(a) and (b).
The code in Figure~\ref{fig:loop-nested-ilsr}(b) is equivalent to the code in Figure~\ref{fig:loop-nested-ilsr}(a), but it is faster because \verb-sum++- is cheaper than $\verb-i- * 10 + j$.
We show how our approach can transform the code in Figure~\ref{fig:loop-nested-ilsr}(a) to the code in Figure~\ref{fig:loop-nested-ilsr}(b).

The \PEGs for the code from parts (a) and (b) are shown in parts (c)
and (d), respectively.  We do not show the entire \PEGs, but only the
parts that are relevant to the optimization -- namely the \PEGs for
the expressions inside the {\tt use} statements. More specifically,
Figure~\ref{fig:loop-nested-ilsr}(c) shows the \PEG for
\verb-i*10 + j-, which is the \PEG that our optimization will apply
to. The top-level \verb-+- node occurs in some larger \PEG context
which includes $\eval$ and $\pass$ nodes, but we do not show the
larger context (\textit{i.e.}: the parents of \verb-+-), because they
are not used in this example, except in one step that we will make
explicit. The result of the optimization, in \PEG form, is shown in
Figure~\ref{fig:loop-nested-ilsr}(d). This is the \PEG for the
\texttt{sum++} expression from Figure~\ref{fig:loop-nested-ilsr}(b).
Note that the code snippet in Figure~\ref{fig:loop-nested-ilsr}(b) is
the same as Figure~\ref{fig:loop-single-double}(b), and as a result
Figure~\ref{fig:loop-nested-ilsr}(d) is just the $\suminner$ node from
Figure~\ref{fig:loop-single-double}(d), along with its children. To summarize,
in terms of \PEGs, our optimization will replace the \verb-+- node
from Figure~\ref{fig:loop-nested-ilsr}(c), which occurs in some larger
\PEG context, with the $\suminner$ node from
Figure~\ref{fig:loop-nested-ilsr}(d). The surrounding \PEG context,
which we do not show, remains unchanged.

\begin{figure}[!t]
\begin{center}
\includegraphics[width=3.5in]{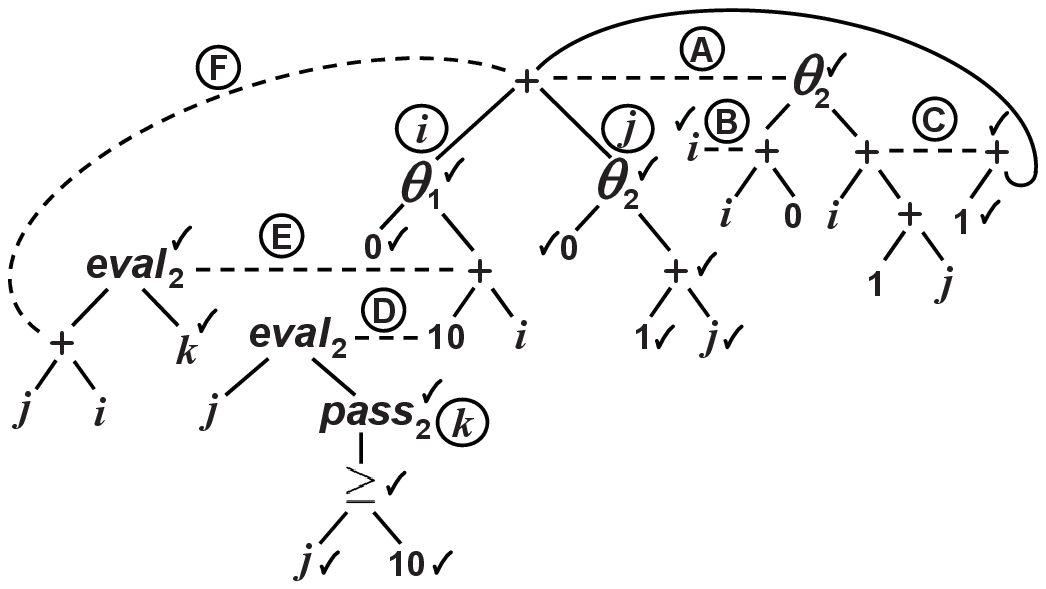}
\end{center}
\caption{
\EPEG that results from running the saturation engine on the
  \PEG from Figure~\ref{fig:loop-nested-ilsr}(c). By picking the nodes
  that are checkmarked, we get the \PEG from
  Figure~\ref{fig:loop-nested-ilsr}(d).  To make the graph more
  readable, we sometimes label nodes, and then connect an edge
  directly to a label name, rather than connecting it to the node with
  that label. For example, consider node $j$ in the \EPEG, which reads
  as $\theta_2(0, 1 + j)$. Rather than explicitly drawing an edge from
  $+$ to $j$, we connect $+$ to a new copy of label $j$. 
}
\label{fig:inter-loop-strength-reduction}
\end{figure}

Figure~\ref{fig:inter-loop-strength-reduction} shows the saturated
\EPEG that results from running the saturation engine on the \PEG from
Figure~\ref{fig:loop-nested-ilsr}(c). The checkmarks indicate which
nodes will eventually be selected -- they can be ignored for now. In
drawing Figure~\ref{fig:inter-loop-strength-reduction}, we have
already performed loop-induction variable strength reduction on the
left child of the topmost $+$ from
Figure~\ref{fig:loop-nested-ilsr}(c). In particular, this left child
has been replaced with a new node $i$, where $i =
\theta_1(0,10+i)$. We skip the steps in doing this because they are
similar to the ones described in
Section~\ref{sec:reasoning-with-a-veg}.

Figure~\ref{fig:inter-loop-strength-reduction} shows the relevant
equalities that our saturation engine would add. We describe each in
turn.
\begin{enumerate}[$\bullet$]
\item Edge A is added by distributing $+$ over $\theta_2$:
  $$i + \theta_2(0, 1 + j) = \theta_2(i+0, i + (1 + j))$$

\item Edge B is added because $0$ is the identity of $+$:
  $$i + 0 = i$$

\item Edge C is added because addition is associative and commutative:
  $$i + (1 + j) = 1 + (i + j)$$

\item Edge D is added because $0$, incremented $n$ times, produces $n$:
$$\eval_\ell(id_\ell, \pass_\ell(id_\ell \geq n)) = n \mbox{ where }id_\ell =
  \theta_\ell(0, 1 + id_\ell)$$
This is an example of a loop optimization expressible as a simple \PEG axiom.

\item Edge E is added by distributing $+$ over the first
  child of $\eval_2$:
  $$\eval_2(j, k) + i = \eval_2(j + i, k)$$

\item Edge F is added because addition is commutative:
  $$j + i = i + j$$
\end{enumerate}

We use checkmarks in Figure~\ref{fig:inter-loop-strength-reduction} to
highlight the nodes that \Peggy would select using its Pseudo-Boolean
profitability heuristic. These nodes constitute exactly the \PEG from
Figure~\ref{fig:loop-nested-ilsr}(d), meaning that \Peggy optimizes
the code in Figure~\ref{fig:loop-nested-ilsr}(a) to the one in
Figure~\ref{fig:loop-nested-ilsr}(b).

\mypara{Summary.} This example illustrates several
points. First, it shows how a transformation that locally seems
undesirable, namely transforming the constant 10 into an expensive
loop (edge D), in the end leads to much better code. Our global
profitability heuristic is perfectly suited for taking advantage of
these situations. Second, it shows an example of an
\emph{unanticipated optimization}, namely an optimization that we did
not realize would fall out from the simple equality analyses we
already had in place. In a traditional compilation system, a
specialized analysis would be required to perform this optimization,
whereas in our approach the optimization simply happens without any
special casing. In this way, our approach essentially allows a few
general equality analyses to do the work of many specialized
transformations.  Finally, it shows how our approach is able to reason
about complex loop interactions, something that is beyond the
reach of current super-optimizer-based techniques.

\section{Local changes have non-local effects}
\label{sec:local-nonlocal}

The axioms we apply during our saturation phase tend to be simple and
local in nature. It is therefore natural to ask how such axioms can
perform anything more than peephole optimizations. The examples shown
so far have already given a flavor of how local reasoning on a \PEG
can lead to complex optimizations. In this section, we show additional
examples of how \Peggy is capable of making significant changes in the
program using its purely local reasoning. We particularly emphasize
how local changes in the \PEG representation can lead to large changes
in the CFG of the program. We conclude the section by describing some
loop optimizations that we have not fully explored using \PEGs, and
which could pose additional challenges.

\subsection{Loop-based code motion}
\label{sec:loop-based-code-motion}

\begin{figure}[t]
  \begin{center}
    \includegraphics[width=4.0in]{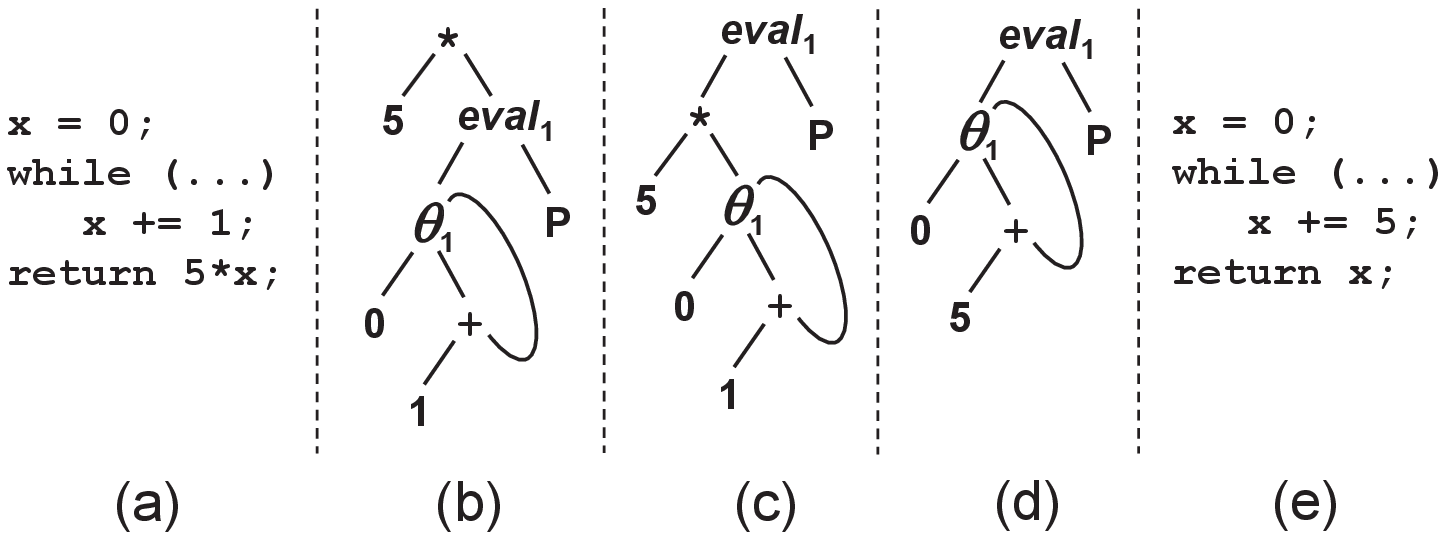}
  \end{center}
  \caption{An example of loop-based code motion from simple axiom
    applications; 
    (a) the original source code, 
    (b) the original \PEG,
    (c) the \PEG after distributing $*$ through $\eval_1$,
    (d) the \PEG after performing loop-induction-variable strength reduction,
    (e) the resulting source code.
  }
  \label{fig:loop-code-motion}
\end{figure}

We start with an example showing how \Peggy can use simple local
axioms to achieve code motion through a loop. Consider the program in
Figure~\ref{fig:loop-code-motion}. Part (a) shows the source code for
a loop where the counter variable is multiplied by 5 at the end, and
part (e) shows equivalent code where the multiplication is removed and
the increment has been changed to 5. Essentially, this optimization
moves the $(*5)$ from the end of the loop and applies it to the
increment and the initial value instead. This constitutes code motion
into a loop, and is a non-local transformation in the CFG.

\Peggy can perform this optimization using local axiom applications,
without requiring any additional non-local
reasoning. Figure~\ref{fig:loop-code-motion}(b) shows the \PEG for the
expression \verb-5*x- in the code from part (a). Parts (c) and (d)
show the relevant pieces of the \EPEG used to optimize this
program. The \PEG in part (c) is the result of distributing
multiplication through the $\eval$ node. The \PEG in part (d) is the
result of applying loop-induction-variable strength reduction to part
(c) (the intermediate steps are omitted for brevity since they are
similar to the earlier example from
Section~\ref{sec:overview}). Finally, the code in part (e) is
equivalent to the \PEG in part (d).

Our mathematical representation of loops is what makes this
optimization so simple. Essentially, when an operator distributes
through $\eval$ (a local transformation in the \PEG), it enters the
loop (leading to code motion). Once inside the loop, distributing it
through $\theta$ makes it apply separately to the initial value and
the inductive value. Then, if there are axioms to simplify those two
expressions, an optimization may result. This is exactly what happened
to the multiply node in the example. In this case, only a simple
operation $(*5)$ was moved into the loop, but the same set of axioms
would allow more complex operations to do the same, using the same
local reasoning.

\subsection{Restructuring the CFG}
\label{sec:restructuring-the-cfg}

In addition to allowing non-local optimizations, small changes in the
\PEG can cause large changes in the program's CFG. Consider the
program in Figure~\ref{fig:phi-on-phi}. Parts (a) and (f) show two CFGs
that are equivalent but have very different structure.  \Peggy can use
several local axiom applications to achieve this same
restructuring. Figure~\ref{fig:phi-on-phi}(b) shows the \PEG version
of the original CFG, and parts (c)-(e) show the relevant portions of
the \EPEG used to optimize it. Part (c) results from distributing the
multiply operator through the left-hand $\phi$ node. Similarly, part
(d) results from distributing each of the two multiply operators
through the bottom $\phi$ node. Part (e) is simply the result of
constant folding, and is equivalent to the CFG in part (f).

\begin{figure}[t]
  \begin{center}
    \includegraphics[width=6.0in]{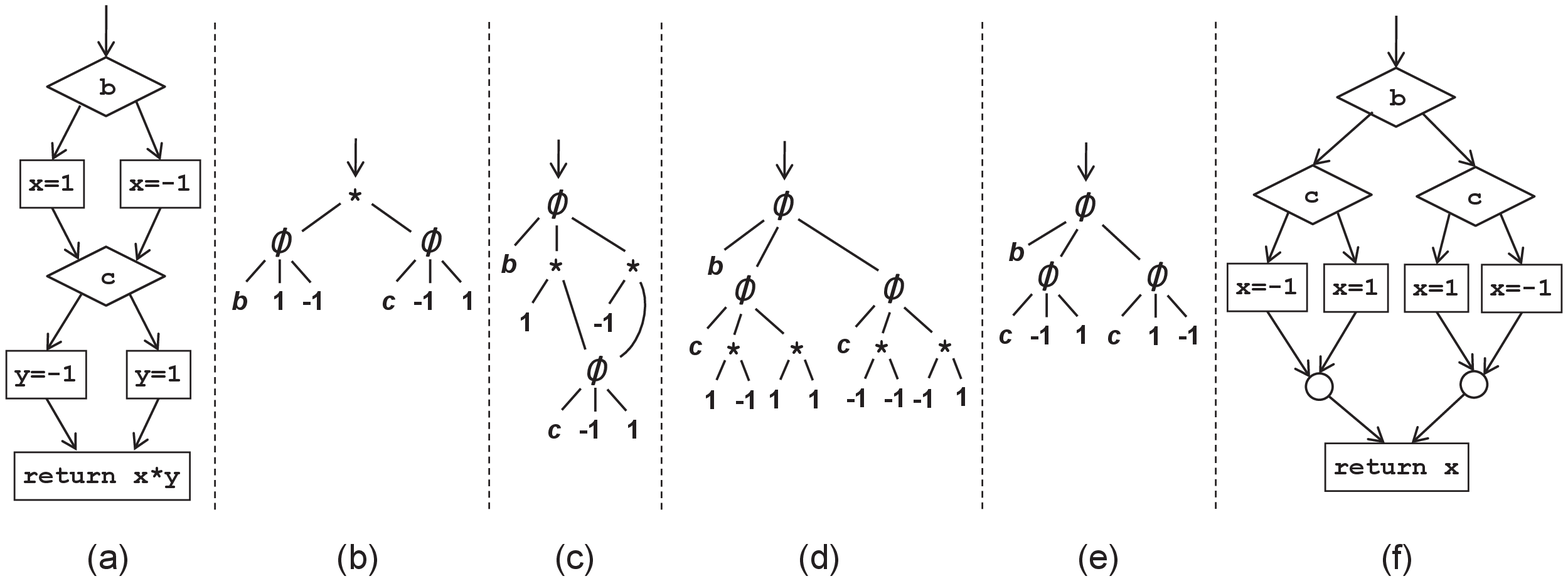}
  \end{center}
  \caption{
    An example of how local changes in the \PEG can cause large changes in the CFG:
    (a) the original CFG,
    (b) the original \PEG,
    (c) the \PEG after distributing $*$ through the left-hand $\phi$,
    (d) the \PEG after distributing $*$ through the bottom $\phi$,
    (e) the \PEG after constant folding,
    (f) the resulting CFG.
  }
  \label{fig:phi-on-phi}
\end{figure}

By simply using the local reasoning of distributing multiplications
through $\phi$ nodes, we have radically altered the branching
structure of the corresponding CFG. This illustrates how small, local
changes to the \PEG representation can have large, far-reaching
effects on the program.

\subsection{Loop Peeling}
\label{sec:loop-peeling}

\begin{figure}[t]
  \begin{center}
    \includegraphics[width=6.0in]{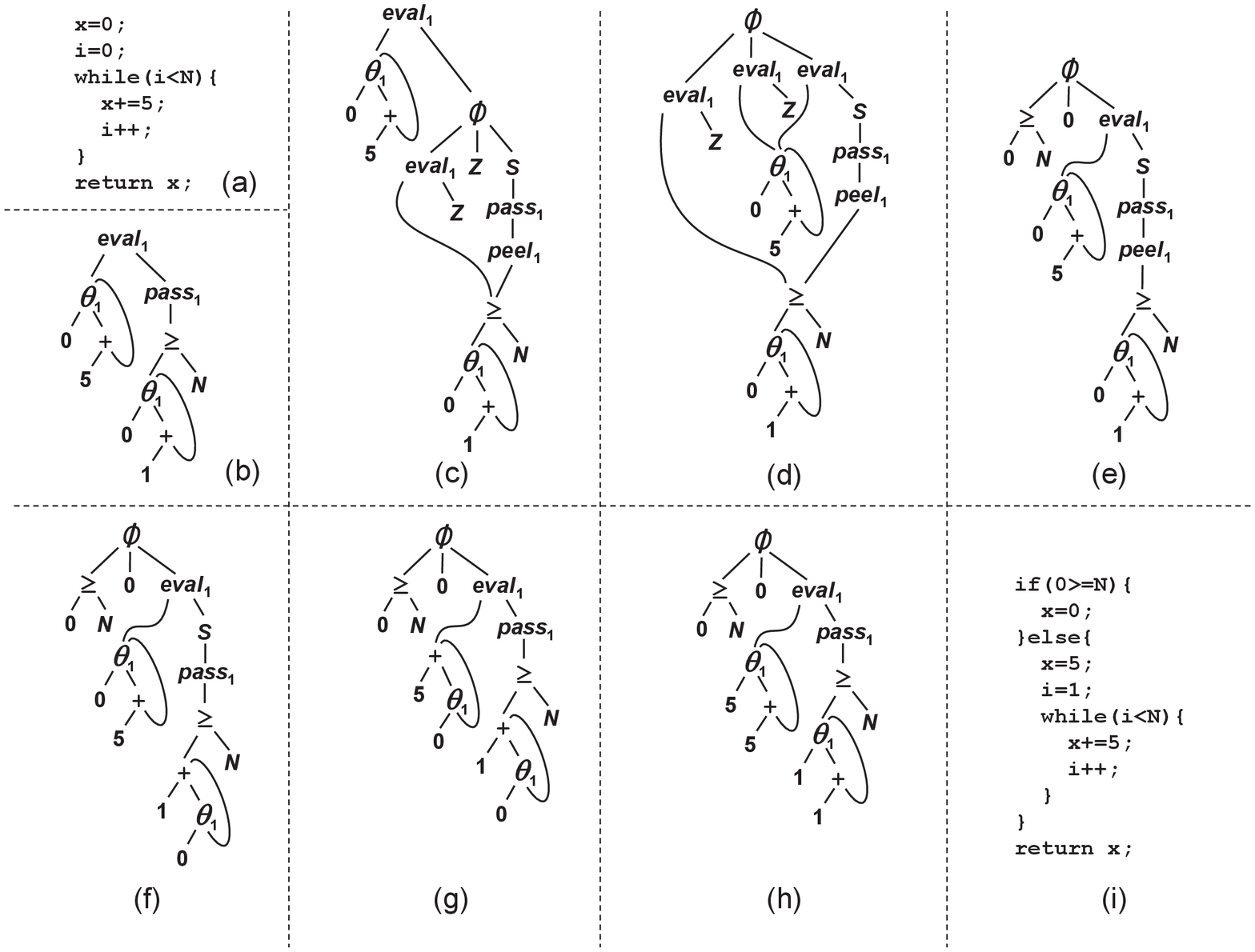}
  \end{center}
  \caption{ 
    An example of axiom-based loop peeling: 
    (a) the original loop, 
    (b) the \PEG for part (a),
    (c)-(h) intermediate steps of the optimization,
    (i) the final peeled loop, which is equivalent to (h).
  }
  \label{fig:loop-peeling}
\end{figure}

Here we present an in-depth example to show how loop peeling is
achieved using equality saturation. Loop peeling essentially takes the
first iteration from a loop and places it before the loop. Using very
simple, general-purpose axioms, we can peel a loop of any type and
produce code that only executes the peeled loop when the original
would have iterated at least once. Furthermore, the peeled loop will
also be a candidate for additional peeling.

Consider the source code in Figure~\ref{fig:loop-peeling}(a). We want
to perform a loop peeling on this code, which will result in the code
shown in Figure~\ref{fig:loop-peeling}(i). This can be done through
axiom application through the following steps, depicted in
Figure~\ref{fig:loop-peeling} parts (c) through (h).  

Starting from the \PEG for the original code, shown in part (b), the
first step transforms the $\pass_1$ node using the axiom
$\pass_1(C) = \phi(\eval_1(C,\Z),\Z,\S(\pass_1(\peel_1(C))))$, yielding the \PEG in part (c). In this axiom, $\Z$ is the zero iteration
count value, $\S$ is a function that takes an iteration count and
returns its successor (i.e. $S = \lambda x.x+1$), and $\peel$ takes a
sequence and strips off the first element (i.e. $\peel(C)[i] =
C[i+1]$). This axiom is essentially saying that the iteration where a
loop stops is equal to one plus where it would stop if you peeled off
the first iteration, but only if the loop was going to run at least
one iteration.

The second step, depicted in part (d), involves distributing the
topmost $\eval_1$ through the $\phi$ node using the axiom
$\op(\phi(A,B,C),D) = \phi(A,\op(B,D),\op(C,D))$. Note that $\op$ only
distributes on the second and third children of the $\phi$ node,
because the first child is the condition.

The third step, shown in part (e), is the result of propagating the
two $\eval_1(\cdot,\Z)$ expressions downward, using the axiom
$\eval_1(\op(a_1,\ldots,a_k),\Z) =
\op(\eval_1(a_1,\Z),\ldots,\eval_1(a_k,\Z))$ when $\op$ is a domain
operator, such as $+, *,$ or $S$. When the $\eval$ meets a $\theta$, it
simplifies using the following axiom: $\eval_1(\theta_1(A,B),\Z) =
A$. Furthermore, we also use the axiom that $\eval_1(C,\Z) = C$ for any
constant or parameter $C$, which is why $\eval_1(\texttt{N},\Z) =
\texttt{N}$.

The fourth step, shown in part (f), involves propagating the $\peel_1$
operator downward, using the axiom $\peel_1(\op(a_1,\ldots,a_k)) =
\op(\peel_1(a_1),\ldots,\peel_1(a_k))$ when $\op$ is a domain
operator. When the $\peel$ operator meets a $\theta$, it simplifies
with the axiom \hbox{$\peel_1(\theta_1(A,B)) = B$}. Furthermore, we also use
the axiom that $\peel_1(C) = C$ for any constant or parameter $C$,
which is why $\peel_1(\texttt{N}) = \texttt{N}$.

The fifth step, shown in part (g), involves removing the $\S$ node
using the axiom $\eval_1(\theta_1(A,B),\S(C)) = \eval_1(B,C)$.

The final step (which is not strictly necessary, as the peeling is
complete at this point) involves distributing the two plus operators
through their $\theta$'s and doing constant folding afterward, to
yield the \PEG in part (h). This \PEG is equivalent to the final
peeled source code in part (i).

It is interesting to see that this version of loop peeling includes
the conditional test to make sure that the original loop would iterate
at least once, before executing the peeled loop. Another way to
implement loop peeling is to exclude this test, opting only to peel
when the analysis can determine statically that the loop will always
have at least one iteration. This limits peeling to certain types of
loops, those with guards that fit a certain pattern. This can both
increase the analysis complexity and reduce the applicability of the
optimization. In the \PEG-based loop peeling, not only do we use the
more applicable version of peeling, but the loop guard expression is
immaterial to the optimization. 

The resulting \PEG shown in Figure~\ref{fig:loop-peeling}(h) is
automatically a candidate for another peeling, since the original
axiom on $\pass$ can apply again. Since we separate our profitability
heuristic from the saturation engine, \Peggy may attempt any number of
peelings. After saturation has completed, the global profitability
heuristic will determine which version of the \PEG is best, and hence
what degree of peeling yields the best result.

\subsection{Branch Hoisting}
\label{sec:branch-hoisting}

\begin{figure}[t]
  \begin{center}
    \includegraphics[width=6.0in]{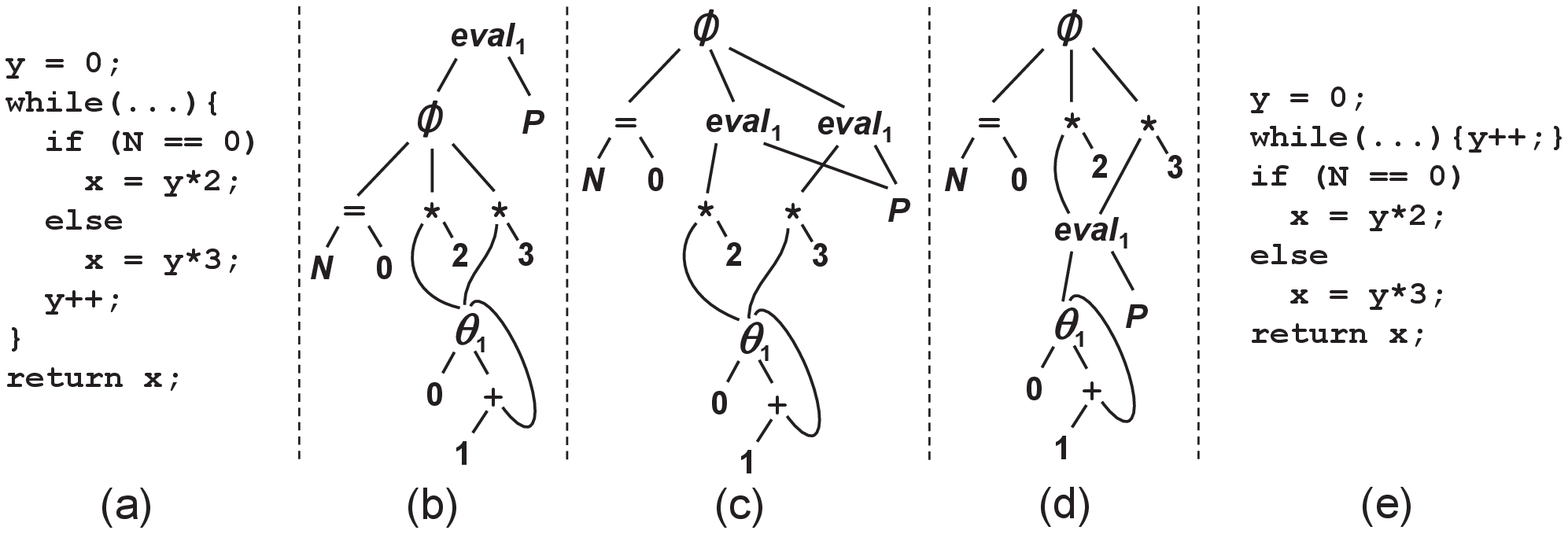}
  \end{center}
  \caption{ 
    An example of branch hoisting:
    (a) the original program,
    (b) the \PEG for part (a),
    (c) the \PEG after distributing $\eval$ through $\phi$,
    (d) the \PEG after distributing $\eval$ through $*$,
    (e) the code resulting from (d).
  }
  \label{fig:branch-hoisting}
\end{figure}

We now examine an example of branch hoisting, where a conditional
branch is moved from inside the loop to after the loop. This is
possible when the condition of the branch is loop-invariant, and hence
is not affected by the loop it's in. This is another example of code
motion, and is an optimization because the evaluation of the branch no
longer happens multiple times inside the loop, but only once at the
end.

Consider the code in Figure~\ref{fig:branch-hoisting}(a). We assume
that \texttt{N} is a parameter or a variable initialized elsewhere,
and is clearly not altered inside the loop. Hence the condition on the
if-statement is loop-invariant. Also we see that \texttt{x} is never
read inside the loop, so the value it holds at the end of the loop can
be expressed entirely in terms of the final values of the other
variables (i.e. \texttt{y}). Hence, this code is equivalent to the
code seen in part (e), where the branch is moved outside the loop and
\texttt{x} is assigned once, using only the final value of \texttt{y}.

Our saturation engine can perform this optimization using simple
axioms, starting with the \PEG shown in part~(b) corresponding
to the code in part~(a). In part~(b), we display the pass condition as
$P$, since we never need to reason about it. Parts~(c)~and~(d) depict
the relevant intermediate steps in the optimization. Part~(c) results
from distributing the $\eval$ operator through the $\phi$ operator
using the axiom \hbox{$\op(\phi(A,B,C),D) = \phi(A,\op(B,D),\op(C,D))$} with $\op = \eval_1$. Part~(d) comes from
distributing the two $\eval$ nodes through the multiplication
operator, using the axiom \hbox{$\eval_1(\op(A,B),P) =
\op(\eval_1(A,P),\eval_1(B,P))$} where $\op$ is any domain operator. Part~(e) is the final code, which is equivalent to the \PEG in part~(d).

Our semantics for $\phi$ nodes allows the $\eval$ to distribute
through them, and hence the loop moves inside the conditional in one
axiom. Since we can further factor the $*$'s out of the $\eval$'s, all
of the loop-based operations are joined at the ``bottom'' of the \PEG,
which essentially means that they are at the beginning of the
program. Here we again see how a few simple axioms can work together
to perform a quite complex optimization that involves radical
restructuring of the program.

\subsection{Limitations of \PEGs}
\label{sec:limitations}

The above examples show how local changes to a \PEG lead to non-local
changes to the CFG. There are however certain kinds of more advanced
loop optimizations that we have not yet fully explored. Although we
believe that these optimizations could be handled with equality
saturation, we have not worked out the full details, and there could
be additional challenges in making these optimizations work in
practice. One such optimization would be to fuse loops from different
nesting levels into a single loop. For example, in the inter-loop
strength reduction example from Section~\ref{sec:loops}, the ideal
output would be a single loop that increments the \verb-sum-
variable. One option for doing this kind of optimization is to add
build-in axioms for fusing these kinds of loops together into
one. Another optimization that we have not fully explored is loop
unrolling. By adding a few additional higher-level operators to our
\PEGs, we were able to perform loop unrolling on paper using just
equational reasoning. Furthermore, using similar higher-level
operators, we believe that we could also perform loop interchange
(which changes a loop \verb-for i in R1, for j in R2- into
\verb-for j in R2 for i in R1-). However, both of these optimizations
do require adding new operators to the \PEG, which would require
carefully formalizing their semantics and axioms that govern
them. Finally, these more sophisticated loop optimizations would also
require a more sophisticated cost model. In particular, because our
current cost model does not take into account loop bounds (only loop
depth), it has only a coarse approximation of the number of times a
loop executes. As a result, it would assign the same cost to the loop
before and after interchange, and it would assign a higher cost to an
unrolled loop than the original. For our cost model to see these
optimizations as profitable, we would have to update it with more
precise information about loop bounds, and a more precise modeling of
various architectural effects like caching and scheduling. We leave
all of these explorations to future work.

\section{Formalization of our Approach}
\label{sec:formal}

\begin{figure}
\begin{algorithmic}[1]
\declarefunction{\Optimize(\cfg:\CFG):\CFG}
  \STATE {\bf let } $\ir = \CfgToIr(\cfg)$
  \STATE {\bf let } $\satir = \Saturate(\ir,A)$
  \STATE {\bf let } $\best = \SelectBest(\satir)$
  \STATE {\bf return } $\IrToCfg(\best)$
\thickstraightline
\end{algorithmic}
\caption{Optimization phase in our approach. We assume a global set $A$ of
  equality analyses to be run.}
\label{fig:optimize}
\end{figure}

Having given an intuition of how our approach works through examples,
we now move to a formal description. Figure~\ref{fig:optimize} shows
the $\Optimize$ function, which embodies our approach. $\Optimize$
takes four steps: first, it converts the input CFG into an internal
representation of the program; second, it saturates this internal
representation with equalities; third, it uses a global profitability
heuristic to select the best program from the saturated
representation; finally, it converts the selected program back to a
CFG.

An instantiation of our approach therefore consists of three
components: (1) an IR where equality reasoning is effective, along
with the translation functions $\CfgToIr$ and $\IrToCfg$, (2) a
saturation engine $\Saturate$, and (3) a global profitability
heuristic $\SelectBest$. Future sections will show how we instantiate
these three components in our \Peggy compiler.

\mypara{Saturation Engine.} The saturation engine $\Saturate$ infers
equalities by repeatedly running a set $A$ of equality analyses. Given
an equality analysis $a \in A$, we define $\eir_1 \atrans{a} \eir_2$
to mean that $\eir_1$ produces $\eir_2$ when the equality analysis $a$
runs and adds some equalities to $\eir_1$. If $a$ chooses not to add
any equalities, then $\eir_2$ is simply the same as $\eir_1$. Note
that $a$ is not required to be deterministic: given a single $\eir_1$,
there may be many $\eir_2$ such that $\eir_1 \atrans{a} \eir_2$. This
non-determinism gives equality analyses that are applicable in multiple
locations in the \EPEG the choice of where to apply. For example, the
distributivity of an operator could apply in many locations, and the
non-determinism allows the distributivity analysis the flexibility of
choosing which instances of distributivity to apply.  Note also that
at this point in the presentation, when saying $\eir_1 \atrans{a}
\eir_2$ we keep $\eir_1$ and $\eir_2$ as abstract representations of
an \EPEG (which includes a \PEG graph and a set of equalities over
\PEG nodes). Later in Section~\ref{sec:semantics} we will formally
define what \PEGs and \EPEGs are.

We define a partial order $\sqsubseteq$ on IRs, based on the nodes and
equalities they encode: $\eir_1 \sqsubseteq \eir_2$ iff the nodes in
$\eir_1$ are a subset of the nodes in $\eir_2$, and the equalities in
$\eir_1$ are a subset of the equalities in $\eir_2$. Immediately from
this definition, we get:
\begin{equation}
(\eir_1 \atrans{a} \eir_2) \Rightarrow \eir_1
\sqsubseteq \eir_2
\label{eq:simplesubset}
\end{equation}

We define an equality analysis $a$ to be monotonic iff:
\begin{equation}
(\eir_1 \sqsubseteq \eir_2) \wedge 
  (\eir_1 \atrans{a} \eir_1') \Rightarrow
  \exists \eir_2'. [(\eir_2 \atrans{a} \eir_2') \wedge (\eir_1' \sqsubseteq
  \eir_2')]
\label{eq:monotonicity}
\end{equation}
This basically states that if $a$ is able to apply to $\eir_1$ to
produce $\eir_1'$ and $\eir_1 \sqsubseteq \eir_2$, then there is a way
to apply $a$ on $\eir_2$ to get some $\eir_2'$ such that $\eir_1'
\sqsubseteq \eir_2'$

If $a$ is monotonic, properties~\eqref{eq:simplesubset}
and~\eqref{eq:monotonicity} immediately imply the following
property:
\begin{equation}
 (\eir_1 \atrans{a} \eir_1') \wedge (\eir_1 \atrans{b} \eir_2)
 \Rightarrow \exists \eir_2'.[(\eir_2 \atrans{a} \eir_2') \wedge
(\eir_1'\sqsubseteq \eir_2')]
\label{eq:non-interference}
\end{equation}
Intuitively, this simply states that applying an equality analysis $b$
before $a$ cannot make $a$ less effective.

We now define $\eir_1 \trans \eir_2$ as:
$$\eir_1 \trans \eir_2 \iff \exists a\in A~.~ (\eir_1 \atrans{a}
\eir_2 \wedge \eir_1 \neq \eir_2)$$ 
The $\trans$ relation formalizes one step taken by the saturation
engine. We also define $\trans^*$ to be the reflexive transitive
closure of $\trans$. The $\trans^*$ relation formalizes an entire run
of the saturation engine. We call a sequence $\eir_1 \atrans{a} \eir_2
\atrans{b} \ldots$ a trace through the saturation engine. We define
$\eir_2$ to be a normal form of $\eir_1$ if $\eir_1 \trans^* \eir_2$
and there is no $\eir_3$ such that $\eir_2 \trans \eir_3$.  It is
straightforward to show the following property:
\begin{equation}
\begin{array}{l}
\mbox{\textit{Given a set $A$ of monotonic equality analyses, 
if $\eir_2$ is a normal formal of $\eir_1$,}}\\
\mbox{\textit{then any other normal form of $\eir_1$ is equal to $\eir_2$.}}
\end{array}
\label{eq:normal-unique}
\end{equation}
In essence, property~\eqref{eq:normal-unique} states that if one trace
through the saturation engine leads to a normal form (and thus a
saturated IR), then any other trace that also leads to a normal form
results in the same saturated IR. In other words, if a given $\eir$
has a normal form, it is unique.

If the set $A$ of analyses makes the
saturation engine terminate on all inputs, then
property~\eqref{eq:normal-unique} implies that the engine is
convergent, meaning that every $\eir$ has a unique normal form.
In general, however, equality saturation may not terminate. For a
given $\eir$ there may not be a normal form, and even if there is a
normal form, some traces may not lead to it because they run forever.
Non-termination occurs when the saturation engine never runs out of
equality analyses that can match in the \EPEG and produce new nodes
and new equalities. For example, the axiom $A = (A+1)-1$ used in the
direction from left to right can be applied an unbounded number of
times, producing successively larger and larger expressions (\verb|x|,
\verb|(x+1)-1|, \verb|(((x+1)-1)+1)-1|, and so on). An inlining axiom
applied to a recursive function can also be applied an unbounded
number of times.

Because unrestricted saturation may not terminate, we bound the number
of times that individual analyses can run, thus ensuring that the
$\Saturate$ function will always halt. In the case when the saturation
engine is stopped early, we cannot provide the same convergence
property, but property~\eqref{eq:non-interference} still implies that
no area of the search space can be made unreachable by applying an
equality analysis (a property that traditional compilation systems
lack).

\section{\PEGs and \EPEGs}
\label{sec:semantics}

The first step in instantiating our approach from the previous section
is to pick an appropriate IR. To this end, we have designed a new IR
called the \EPEG which can simultaneously represent multiple optimized
versions of the input program. We first give a formal description of
our IR (Section~\ref{sec:formalization}), then we present its benefits
(Section~\ref{sec:benefits}), and finally we give a detailed
description of how to translate from CFGs to our IR and back
(Sections~\ref{sec:cfg2peg} and \ref{sec:peg2cfg}).

\subsection{Formalization of \PEGs}
\label{sec:formalization}

A \PEG is a triple $\langle \Node, \Label, \Param\rangle$, where
$\Node$ is a set of nodes, $\Label:\Node \rightarrow \Func$ is a
labeling that maps each node to a semantic function from a set of
semantic functions $\Func$, and $\Param:\Node \rightarrow
\listtype{\Node}$ is a function that maps each node to its children
(i.e. arguments). For a given node $n$, if $\Label(n) = f$, we say
that $n$ is \emph{labeled} with $f$. We say that a node $n'$ is a
\emph{child} of node $n$ if $n'$ is an element of
$\Param(n)$. Finally, we say that $n_k$ is a \emph{descendant} of
$n_0$ if there is a sequence of nodes $n_0, n_1, \ldots, n_k$ such
that $n_{i+1}$ is a child of $n_i$ for $0 \leq i < k$.

\mypara{Types.}  Before giving the definition of semantic functions,
we first define the types of values that these functions operate
over. Values that flow through a \PEG are lifted in two ways. First,
they are $\bot$-lifted, meaning that we add the special value $\bot$
to each type domain. The $\bot$ value indicates that the computation
fails or does not terminate. Formally, for each type $\tau$, we define
$\U{\D} = \tau \cup \{ \undef \}$.

Second, values are loop-lifted, which means that instead of
representing the value at a particular iteration, \PEG nodes represent
values for all iterations at the same time. Formally, we let $\set{L}$
be a set of loop identifiers, with each $\ell \in \set{L}$
representing a loop from the original code (in our previous examples
we used integers). We assume a partial order $\leq$ that represents
the loop nesting structure: $\ell < \ell'$ means that $\ell'$ is
nested within $\ell$. An iteration index $\i$ captures the iteration
state of all loops in the \PEG. In particular, $\i$ is a function that
maps each loop identifier $\ell \in \set{L}$ to the iteration that
loop $\ell$ is currently on. Suppose for example that there are two
nested loops in the program, identified as $\ell_1$ and $\ell_2$. Then
the iteration index $\i = [ \ell_1 \mapsto 5, \ell_2 \mapsto 3 ]$
represents the state where loop $\ell_1$ is on the $5^{th}$ iteration
and loop $\ell_2$ is on the $3^{rd}$ iteration.  We let $\I = \set{L}
\to \W$ be the set of all loop iteration indices (where $\W$ denotes
the set of non-negative integers). For $\i \in \I$, we use the
notation $\i[\ell \mapsto v]$ to denote a function that returns the
same value as $\i$ on all inputs, except that it returns $v$ on input
$\ell$. The output of a \PEG node is a map from loop iteration indices
in $\I$ to values. In particular, for each type $\tau$, we define a
loop-lifted version $\Lifted{\tau} = \I \to \U{\D}$. \PEG nodes
operate on these loop-lifted types.

\mypara{Semantic Functions.}
The semantic functions in $\Func$ actually implement the operations
represented by the \PEG nodes. Each function $f \in F$ has type
$\Lifted{\tau_1} \times \ldots \times \Lifted{\tau_k} \rightarrow
\Lifted{\tau}$, for some $k$. Such an $f$ can be used as the label for
a node that has $k$ children. That is to say, if $\Label(n) = f$,
where $f: \Lifted{\tau_1} \times \ldots \times \Lifted{\tau_k}
\rightarrow \Lifted{\tau}$, then $\Param(n)$ must be a list of $k$
nodes.

The set of semantic functions $\Func$ is divided into two: $\Func =
\Primitives \cup \DomainFunc$. $\Primitives$ contains the
\emph{primitive functions} like $\phi$ and $\theta$, which are built
into the \PEG representation, whereas $\DomainFunc$ contains semantic
functions for particular domains like arithmetic.

\begin{figure}
$$
\begin{array}{ll}

\fbox{\raisebox{0pt}[8pt][1pt]{
$\sem{\phi}: \Lifted{\B} \times \Lifted{\D} \times \Lifted{\D} \to \Lifted{\D}$}} \\

\hspace{0.25in}
\sem{\phi}(cond, t, f)(\i) = 
\begin{cases}
\text{\textbf{if} $\;cond(\i) = \undef$} & 
\text{\textbf{then} $\;\undef$} \\
\text{\textbf{if} $\;cond(\i) = \true$} & 
\text{\textbf{then} $\;t(\i)$} \\
\text{\textbf{if} $\;cond(\i) = \false$} & 
\text{\textbf{then} $\;f(\i)$}
\end{cases} 
\vspace{0.2cm}
\\

\fbox{\raisebox{0pt}[8pt][1pt]{
$\sem{\theta_\ell}: \Lifted{\D} \times \Lifted{\D} \to \Lifted{\D}$}} \\

\hspace{0.25in}
\sem{\theta_\ell}(\base, \loopit)(\i) = 
\begin{cases}
\text{\textbf{if} $\;\i(\ell) = 0$ \textbf{ then} $\;\base(\i)$} \\
\text{\textbf{if} $\;\i(\ell) > 0$ \textbf{ then} $\;\loopit(\i[\ell \mapsto \i(\ell)-1])$}
\end{cases}
\vspace{0.2cm}
\\

\fbox{\raisebox{0pt}[8pt][1pt]{
$\sem{\eval_\ell}: \Lifted{\D} \times \Lifted{\W} \to \Lifted{\D}$}} \\

\hspace{0.25in}
\sem{\eval_\ell}(\loopit, \idx)(\i) = 
\begin{cases}
\text{\textbf{if} $\;\idx(\i) = \undef\;$ 
\textbf{then} $\;\undef$} \\
\text{\textbf{else} $\monotonize_\ell(\loopit)(\i[\ell \mapsto \idx(\i)])$}
\end{cases} 
\vspace{0.2cm}
\\

\fbox{\raisebox{0pt}[8pt][1pt]{
$\sem{\pass_\ell}: \Lifted{\B} \to \Lifted{\W}$}}\\

\hspace{0.25in}
\sem{\pass_\ell}(\cond)(\i) = 
\begin{cases}
\text{\textbf{if} $\;\set{I} = \emptyset$} &
\text{\textbf{then} $\;\undef$} \\
\text{\textbf{if} $\;\set{I} \neq \emptyset$} &
\text{\textbf{then} $\;\min\set{I}$}
\end{cases} 
\vspace{0.2cm}
\\

\hspace{0.25in}
\text{where } \set{I} = \left\{i \in \W \mid
\monotonize_\ell(\cond)(\i[\ell \mapsto i]) = \true \right\} \\
\\

\mbox{where $\sem{\monotonize_\ell}: \Lifted{\D} \to \Lifted{\D}$ is defined as:} \\

\monotonize_\ell(\valueit)(\i) =
\begin{cases}
\text{\textbf{if} $\;\exists \; 0 \leq i < \i(\ell). \; \valueit(\i[\ell \mapsto i]) = \undef$ \textbf{ then} $\;\undef$} \\
\text{\textbf{if} $\;\forall \; 0 \leq i < \i(\ell). \; \valueit(\i[\ell \mapsto i]) \neq \undef$ \textbf{ then} $\;\valueit(\i)$} \\
\end{cases}

\end{array}
$$
\caption{Definition of primitive PEG functions. The important
notation: $\set{L}$ is the set of loop identifiers, $\W$ is the set of
non-negative integers, $\B$ is the set of booleans, $\;\I = \set{L}
\to \W$, $\ \ \U{\D} = \tau \cup \{ \undef \}$, ~~and
$\Lifted{\tau} = \I \to \U{\D}$.
}
\label{fig:prims-defns}
\end{figure}

Figure~\ref{fig:prims-defns} gives the definition of the primitive
functions $\Primitives = \{\phi, \theta_\ell, \eval_\ell,
\pass_\ell\}$. These functions are polymorphic in $\tau$, in that they
can be instantiated for various $\tau$'s, ranging from basic types
like integers and strings to complicated types like the heap summary
nodes that \Peggy uses to represent Java objects. The definitions of
$\eval_\ell$ and $\pass_\ell$ make use of the function $\monotonize_\ell$,
whose definition is given in Figure~\ref{fig:prims-defns}. The
$\monotonize_\ell$ function transforms a sequence so that, once an indexed
value is undefined, all following indexed values are undefined. The
$\monotonize_\ell$ function formalizes the fact that once a value is
undefined at a given loop iteration, the value remains undefined at
subsequent iterations.

The domain semantic functions are defined as $\DomainFunc = \{
\Lifted{op} \mid op \in \DomainOperators \}$, where $\DomainOperators$
is a set of domain operators (like $+$, $*$ and $-$ in the case of
arithmetic), and $\Lifted{op}$ is a $\bot$-lifted, and then
loop-lifted version of $op$. Intuitively, the $\bot$-lifted version of
an operator works like the original operator except that it returns
$\bot$ if any of its inputs are $\bot$, and the loop-lifted version of
an operator applies the original operator for each loop index.

As an example, the semantic function of $+$ in a PEG is $\Lifted{+}$,
and the semantic function of $1$ is $\Lifted{1}$ (since constants like
$1$ are simply nullary operators). 
However, to make the notation less crowded, we omit the tildes
on all domain operators.

\mypara{Node Semantics.} For a \PEG node $n \in \Node$, we denote
its semantic value by $\semvalue{n}$. We assume that
$\semvalue{\cdot}$ is lifted to sequences $\listtype{\Node}$ in the standard
way. The semantic value of $n$ is defined as:
\begin{equation}
\label{eq:semantic-value-defn}
\semvalue{n} = \Label(n)(\semvalue{\Param(n)})
\end{equation}
Equation~\ref{eq:semantic-value-defn} is essentially the evaluation
semantics for expressions. The only complication here is that our
expression graphs are recursive. In this setting, one can think of
Equation~\ref{eq:semantic-value-defn} as a set of recursive equations
to be solved. To guarantee that a unique solution exists, we impose
some well-formedness constraints on \PEGs.

\vspace{12pt}

\begin{defi}[\PEG Well-formedness]
A PEG is well-formed iff:
\begin{enumerate}[(1)]
\item All cycles pass through the second child edge of a
  $\theta$
\item A path from a $\theta_\ell$, $\eval_\ell$, or $\pass_\ell$ to a
$\theta_{\ell'}$ implies $\ell' \leq \ell$ or the path passes through
the first child edge of an $\eval_{\ell'}$ or $\pass_{\ell'}$
\item All cycles containing $\eval_\ell$ or $\pass_\ell$ contain some
$\theta_{\ell'}$ with $\ell' < \ell$
\end{enumerate}
\label{defn:well-formed}
\end{defi}
Condition 1 states that all cyclic paths in the PEG are due to looping
constructs. Condition 2 states that a computation in an outer-loop
cannot reference a value from inside an inner-loop. Condition 3 states
that the final value produced by an inner-loop cannot be expressed in
terms of itself, except if it's referencing the value of the
inner-loop from a \emph{previous} outer-loop iteration. From this
point on, all of our discussion of \PEGs will assume they are
well-formed.

\begin{theorem}
If a \PEG is well-formed, then for each node
$n$ in the \PEG there is a unique semantic value $\semvalue{n}$
satisfying Equation~\ref{eq:semantic-value-defn}.
\label{thm:value-existence}
\end{theorem}

The proof is by induction over the strongly-connected-component DAG of
the \PEG and the loop nesting structure $\leq$.

\mypara{Evaluation Semantics.} The meaning function $\semvalue{\cdot}$
can be evaluated on demand, which provides an executable semantics for
\PEGs. For example, suppose we want to know the result of
$\eval_\ell(x, \pass_\ell(y))$ at some iteration state $\i$. To
determine which case of $\eval_\ell$'s definition we are in, we must
evaluate $\pass_\ell(y)$ on $\i$. From the definition of $\pass_\ell$,
we must compute the minimum $i$ that makes $y$ true. To do this, we
iterate through values of $i$ until we find an appropriate one. The
value of $i$ we've found is the number of times the loop iterates, and
we can use this $i$ back in the $\eval_\ell$ function to extract the
appropriate value out of $x$.  This example shows how an on-demand
evaluation of an $\eval$/$\pass$ sequence essentially leads to an
operational semantics for loops. Though it may seem that this
semantics requires each loop to be evaluated twice (once to determine
the $\pass$ value and once to determine the $\eval$ result), a
practical implementation of \PEGs (such as our
\PEG-to-imperative-code conversion algorithm in
Section~\ref{sec:peg2cfg}) can use a single loop to compute both the
$\pass$ result and the $\eval$ result.

\mypara{Parameter nodes.} Our \PEG definition can easily be extended
to have \emph{parameter nodes}, which are useful for encoding the
input parameters of a function or method. In particular, we assume
that a \PEG $\langle \Node, \Label, \Param\rangle$ has a (possibly
empty) set $\Node_p \subseteq \Node$ of parameter nodes.  A parameter
node $n$ does not have any children, and its label is of the form
$param(x)$ where $x$ is the variable name of the parameter.  To
accommodate for this in the formalism, we extend the type of our
labeling function $\Label$ as $\Label: \Node \rightarrow \Func \cup
\ParamLabels$, where \hbox{$\ParamLabels = \{ param(x) \mid x \mbox{ is a
  variable name} \}$}. There are several ways to give semantics to
\PEGs with parameter nodes. One way is to update the semantic
functions in Figure~\ref{fig:prims-defns} to pass around a value
context $\Sigma$ mapping variables to values. Another way, which we
use here, is to first apply a substitution to the \PEG that replaces
all parameter nodes with constants, and then use the node semantics
$\semvalue{\cdot}$ defined earlier. The node semantics
$\semvalue{\cdot}$ is well defined on a \PEG where all parameters have
been replaced, since $\Label(n)$ in this case would always return a
semantic function from $\Func$, never a parameter label $param(x)$
from $\ParamLabels$.

\begin{figure}
\includegraphics[width=\textwidth]{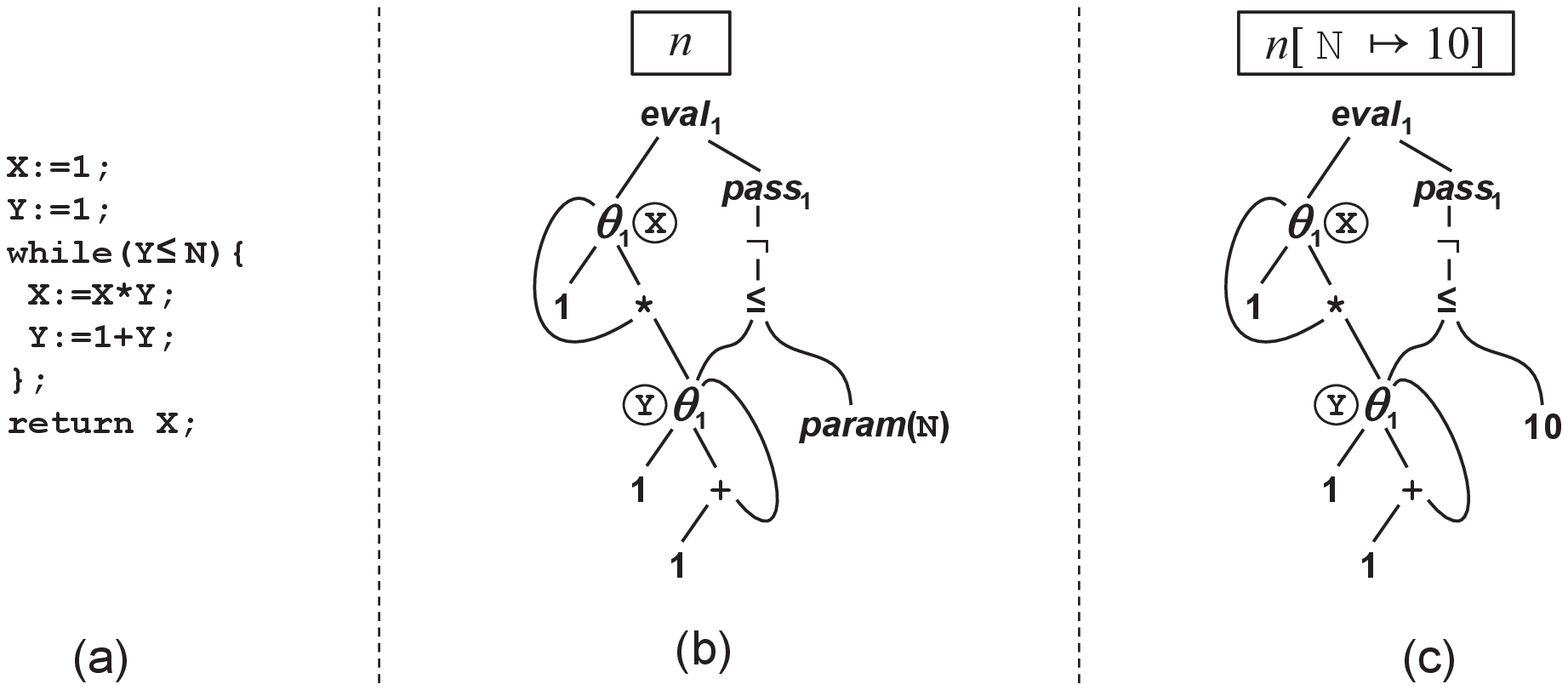}
\caption{Example showing \PEG with parameter nodes. (a) shows code for
  computing the factorial of $\texttt{N}$, where $\texttt{N}$ is a
  parameter (b) shows the \PEG with $n$ being the value returned and
  (c) shows $n[\texttt{N} \mapsto 10]$, which is now a \PEG whose
  semantics is well-defined in our formalism.}
\label{param-example}
\end{figure}

We use the following notation for substitution: given a \PEG node $n$,
a variable name $x$, and a constant $c$ (which is just a nullary
domain operator $\op \in \DomainFunc$), we use $n[x \mapsto c]$ to
denote $n$ with every descendant of $n$ that is labeled with
$param(x)$ replaced with a node labeled with $\Lifted{c}$. We use
$n[x_1 \mapsto c_1, \ldots, x_k \mapsto c_k]$ to denote $n[x_1\mapsto
  c_1] \ldots [x_k \mapsto c_k]$. Figure~\ref{param-example} shows an
example with parameter nodes and an example of the substitution
notation.

\subsection{Formalization of \EPEGs} An \EPEG is a \PEG with a set of equalities
$E$ between nodes. Thus, formally, an \EPEG is a quadruple $\langle
\Node, \Label, \Param, E\rangle$, where $\langle \Node, \Label, \Param
\rangle$ is a \PEG and $E \subseteq \Node \times \Node$ is a set of
pairs of nodes representing equalities. An equality between $n$ and
$n'$ denotes value equality: $\semvalue{n} = \semvalue{n'}$. The set
$E$ forms an equivalence relation $\sim$ (that is, $\sim$ is the
reflexive transitive symmetric closure of $E$), which in turn
partitions the \PEG nodes into equivalence classes. We denote by $[n]$
the equivalence class that $n$ belongs, so that $[n] = \{ n' \in \Node
\mid n' \sim n\}$. We denote by $N/E$ the set of all equivalence
classes. For $n \in \Node$, we denote by $params(n)$ the list of
equivalence classes that are parameters to $n$. In particular, if
$\Param(n) = (n_1, \ldots, n_k)$ then $params(n) = ([n_1], \ldots,
[n_k])$. As mentioned in more detail in
Section~\ref{sec:implementation}, our implementation strategy keeps
track of these equivalence classes, rather than the set $E$.

\subsection{Built-in Axioms.} 
\label{sec:build-in-axioms}

We have developed a set of \PEG built-in axioms that state properties
of the primitive semantic functions. These axioms are used in our
approach as a set of equality analyses that enable reasoning about
primitive \PEG operators.  Some important built-in axioms are given
below, where $\bullet$ denotes ``don't care'':
$$\begin{array}{rcl}
\theta_\ell(A,B) & = & \theta_\ell(\eval_\ell(A,0),B) \\
\eval_\ell(\theta_\ell(A,\bullet),0) & = & \eval_\ell(A,0) \\
\eval_\ell(\eval_\ell(A,B),C) & = & \eval_\ell(A,\eval_\ell(B,C)) \\
\pass_\ell(\true) & = & 0 \\
\pass_\ell(\theta_\ell(\true,\bullet)) & = & 0 \\
\pass_\ell(\theta_\ell(\false,A)) & = & \pass_\ell(A) + 1 \\
\end{array}$$

Furthermore, some axioms make use of an invariance predicate: $\invariant_\ell(n)$ is true if the value of $n$ does not
vary on loop $\ell$. Although we define $\invariant_\ell$ here first,
$\invariant_\ell$ will be used more broadly than for defining
axioms. It will also be used in
Section~\ref{sec:loop-invariant-code-motion} to optimize the
\PEG-to-imperative-code translation, and in
Section~\ref{sec:heuristic} to help us define the cost model for
\PEGs. Invariance can be computed using several syntactic rules, as
shown in the following definition, although there are \PEG nodes which
are semantically invariant but do not satisfy the following syntactic
predicate. Note that we sometimes use ``$n$ is $\invariant_\ell$''
instead of $\invariant_\ell(n)$.

\begin{defi}[Invariance Predicate]
The $\invariant_\ell(n)$ predicate is the largest predicate that
satisfies the following three rules:
\begin{enumerate}[(1)]
\item if $\Label(n)$ is $\theta_\ell$, then $n$ is not $\invariant_\ell$
\item if $\Label(n)$ is $\eval_\ell$, then if the second child of $n$ is not $\invariant_\ell$ then $n$ is also not $\invariant_\ell$
\item otherwise if $\Label(n)$ is not $\pass_\ell$, then if any child of $n$ is not $\invariant_\ell$ then $n$ is also not $\invariant_\ell$
\end{enumerate}

\label{defn:invariant}
\end{defi}

Note that, due to the last rule, nodes without any children, such as
constants and parameter nodes, will always be $\invariant_\ell$ for
all $\ell$. Also, since no rule restricts $\pass_\ell$ nodes, such
nodes will always be $\invariant_\ell$. This syntactic definition of
$\invariant_\ell$ is best computed using an optimistic dataflow
analysis.

Having defined $\invariant_\ell$, the following built-in axioms hold
if $\invariant_\ell(A)$ holds:
$$\begin{array}{rcl}
\eval_\ell(A,\bullet) & = & A \\
x & = & A \quad \mbox{where } x=\theta_\ell(A,x) \\
\peel_\ell(A) & = & A \\
\end{array}$$

One of the benefits of having a well-defined semantics for primitive
\PEG functions is that we can reason formally about these
functions. To demonstrate that this is feasible, we used our semantics
to prove a handful of axioms, in particular, the above axioms, and all
the axioms required to perform the optimizations presented in
Sections~\ref{sec:overview}
through~\ref{sec:local-nonlocal}. Appendix~\ref{sec:axioms} contains
the much longer list of all axioms that we have used in \Peggy.

\subsection{How \PEGs enable our approach}
\label{sec:benefits}

The key feature of \PEGs that makes our equality-saturation approach
effective is that they are referentially transparent, which
intuitively means that the value of an expression depends only on the
values of its constituent expressions~\cite{strachey_ref_trans,
  landin_ref_trans, quine_ref_trans}. In our \PEG representation,
referential transparency can be formalized as follows:
$$
\begin{array}{l}
\forall (n,n') \in \Node^2~.~
\left(
\begin{array}{l}
\Label(n) = \Label(n') \wedge \\
\semvalue{\Param(n)} = \semvalue{\Param(n')}
\end{array}
\right)
\Rightarrow
\semvalue{n} = \semvalue{n'}
\end{array}
$$ 
This property follows from the definition in
Equation~\eqref{eq:semantic-value-defn}, and
the fact that for any $n$, $\Label(n)$ is a pure mathematical
function.

Referential transparency makes equality reasoning effective because it
allows us to show that two expressions are equal by only considering
their constituent expressions, without having to worry about
side-effects. Furthermore, referential transparency has the benefit
that a single node in the \PEG entirely captures the value of a
complex program fragment (including loops) enabling us to record
equivalences between program fragments by using equivalence classes of
nodes. Contrast this to CFGs, where to record equality between complex
program fragments, one would have to record subgraph equality.

Finally, \PEGs allow us to record equalities at the granularity of
individual values, for example the iteration count in a loop, rather
than at the level of the entire program state. Again, contrast this to
CFGs, where the simplest form of equality between program fragments
would record program-state equality.

\section{Representing Imperative Code as \PEGs}
\label{sec:cfg2peg}

In this section we describe how programs written in an imperative
style can be transformed to work within our \PEG-based optimization
system. We first define a minimal imperative language
(Section~\ref{sec:simple}) and then present ML-style functional
pseudocode for converting any program written in this language into a
\PEG (Section~\ref{sec:convert-pseudo-code}). Next, we present a
formal account of the conversion process using type-directed
translation rules in the style of~\cite{Leijen:mlftof}
(Section~\ref{sec:type-directed-translation}). Finally, we outline a
proof of the semantic preservation of the translation
(Section~\ref{sec:preservation-of-semantics}).
Our technical report~\cite{peg2cfg} shows how to extend the techniques in this section for converting \SIMPLE programs to \PEGs in order to convert more complex control flow graphs to \PEGs.

\subsection{The \SIMPLE programming language}
\label{sec:simple}

\begin{figure}[t]
$$\begin{array}{r@{~::=~}l}
p & \texttt{main}(x_1:\tau_1, \ldots, x_n:\tau_n):\tau~\{ s \} \\
s & s_1;s_2 \mid x := e \mid \simpleite{e}{s_1}{s_2} ~|~ \simplewhile{e}{s} \\
e & x \mid \op(e_1, \dots, e_n)
\end{array}$$
\caption{Grammar for \SIMPLE programs}
\label{simple-grammar}
\end{figure}

We present our algorithm for converting to the \PEG representation
using a simplified source language. In particular, we use the \SIMPLE
programming language, the grammar of which is shown in
Figure~\ref{simple-grammar}. A \SIMPLE program contains a single
function \texttt{main}, which declares parameters with their
respective types, a body which uses these variables, and a return
type. There is a special variable $\retvar$, the value of which is
returned by \texttt{main} at the end of execution. \SIMPLE programs
may have an arbitrary set of primitive operations on an arbitrary set
of types; we only require that there is a Boolean type for
conditionals (which makes the translation simpler). Statements in
\SIMPLE programs have four forms: statement sequencing (using
semicolon), variable assignment (the variable implicitly inherits the
type of the expression), if-then-else branches and while
loops. Expressions in \SIMPLE programs are either variables or
primitive operations (such as addition). Constants in \SIMPLE are
nullary primitive operations.

\begin{figure}[t]
{
\begin{tabular}{c}
\multicolumn{1}{l}{
\fbox{\raisebox{0pt}[8pt][1pt]{$\vdash p$ \;\; (programs)}} 
} \\[15pt]
\inference[Type-Prog]{x_1 : \tau_1, \dots, x_n : \tau_n \vdash s : \Gamma
       & \Gamma(\retvar) = \tau
       }{\vdash \texttt{main}(x_1 : \tau_1, \dots, x_n : \tau_n):\tau~\{s\}} \\[20pt]
\multicolumn{1}{l}{
\fbox{\raisebox{0pt}[8pt][1pt]{$\Gamma \vdash s: \Gamma'$ \;\; (statements)}} 
} \\[15pt]
\inference[Type-Seq]{\Gamma \vdash s_1 : \Gamma'
       & \Gamma' \vdash s_2 : \Gamma''
       }{\Gamma \vdash s_1; s_2 : \Gamma''} \;\;
\inference[Type-Asgn]{\Gamma \vdash e : \tau}{\Gamma \vdash x := e : (\Gamma, x : \tau)} \\[20pt]
\inference[Type-If]{\Gamma \vdash e : \texttt{bool}
       & \Gamma \vdash s_1 : \Gamma'
       & \Gamma \vdash s_2 : \Gamma'
       }{\Gamma \vdash \simpleite{e}{s_1}{s_2} : \Gamma'} \;\;
\inference[Type-While]{\Gamma \vdash e : \texttt{bool}
       & \Gamma \vdash s : \Gamma
       }{\Gamma \vdash \simplewhile{e}{s} : \Gamma} \\[20pt]
\inference[Type-Sub]{\Gamma \vdash s : \Gamma'
       & \Gamma'' \subseteq \Gamma'
       }{\Gamma \vdash s : \Gamma''} \\[20pt]
\multicolumn{1}{l}{
\fbox{\raisebox{0pt}[8pt][1pt]{$\Gamma \vdash e: \tau$ \;\; (expressions)}} 
} \\[15pt]
\inference[Type-Var]{\Gamma(x) = \tau
       }{\Gamma \vdash x : \tau} \;\;
\inference[Type-Op]{op : (\tau_1, \dots, \tau_n) \to \tau
       & \Gamma \vdash e_1 : \tau_1 \;\; \dots \;\; \Gamma \vdash e_n : \tau_n
       }{\Gamma \vdash op(e_1, \dots, e_n) : \tau} 
\end{tabular}
}
\caption{Type-checking rules for \SIMPLE programs}
\label{simple-types}
\end{figure}

The type-checking rules for \SIMPLE programs are shown in
Figure~\ref{simple-types}. There are three kinds of judgments: (1)
judgment $\vdash p$ (where $p$ is a program) states that $p$ is
well-typed; (2) judgment $\Gamma \vdash s:\Gamma'$ (where $s$ is a
statement) states that starting with context $\Gamma$, after $s$ the
context will be $\Gamma'$; (3) judgment $\Gamma \vdash e: \tau$ (where
$e$ is an expression) states that in type context $\Gamma$, expression
$e$ has type $\tau$.

For program judgments, there is only one rule, Type-Prog, which
ensures that the statement inside of \texttt{main} is well-typed;
$\Gamma(\retvar) = \tau$ ensures that the return value of
\texttt{main} has type $\tau$. For statement judgments, Type-Seq
simply sequences the typing context through two statements. Type-Asgn
replaces the binding for $x$ with the type of the expression assigned
into $x$ (if $\Gamma$ contains a binding for $x$, then
$(\Gamma,x:\tau)$ is $\Gamma$ with the binding for $x$ replaced with
$\tau$; if $\Gamma$ does not contain a binding for $x$, then
$(\Gamma,x:\tau)$ is $\Gamma$ extended with a binding for $x$). The
rule Type-If requires the context after each branch to be the
same. The rule Type-While requires the context before and after the
body to be the same, specifying the loop-induction variables. The rule
Type-Sub allows the definition of variables to be ``forgotten'',
enabling the use of temporary variables in branches and loops.

\subsection{Translating \SIMPLE Programs to \PEGs}
\label{sec:convert-pseudo-code}

\begin{figure}
\begin{algorithmic}[1]
\declarefunction{\TranslateProg(p: \Prog):\NodeType =} 
\STATE \tab\textbf{let} $m = \InitMap(p.\params,\lambda x.\ParemeterNode(x))$
\STATE \tab\textbf{in} $\TranslateStmt(p.\body, m, 0)(\retvar)$
\finishfunction
\declarefunction{\TranslateStmt(s: \Stmt, \context:\maptype{\Var}{\NodeType}, \loopdepth: \N):\maptype{\Var}{\NodeType}=}
\STATE \tab\textbf{match} $s$ \textbf{with}
\STATE \tab\tab $``s_1;s_2" \Rightarrow \TranslateStmt(s_2, \TranslateStmt(s_1, \context, \loopdepth), \loopdepth)$
\STATE \tab\tab $``x := e" \Rightarrow \context[x \mapsto \TranslateExpr(e,\context)]$
\STATE \tab\tab $``\simpleite{e}{s_1}{s_2}" \Rightarrow \PHI(\TranslateExpr(e,\context), \TranslateStmt(s_1,\context, \loopdepth), \TranslateStmt(s_2,\context,\loopdepth))$
\STATE \tab\tab $``\simplewhile{e}{s}" \Rightarrow$
\STATE \tab\tab \tab \textbf{let} $\vars = \Keys(\context)$
\STATE \tab\tab \tab \textbf{let} $\varnodes = \InitMap(\vars, \lambda v. \TemporaryNode(v))$
\STATE \tab\tab \tab \textbf{let} $\nextnodes = \TranslateStmt(s, \varnodes, \loopdepth + 1)$
\STATE \tab\tab \tab \textbf{let} $\thetanodes = \THETA_{\loopdepth+1}(\context, \nextnodes)$
\STATE \tab\tab \tab \textbf{let} $\thetanodes' = \InitMap(\vars, \lambda v.\FixpointTemporaries(\thetanodes, \thetanodes(v)))$
\STATE \tab\tab \tab \textbf{in} $\EVAL_{\loopdepth + 1}(\thetanodes', \constructor{\pass_{\loopdepth + 1}}(\constructor{\neg}(\TranslateExpr(e,\thetanodes'))))$
\finishfunction
\declarefunction{\TranslateExpr(e: \Expr,\context:\maptype{\Var}{\NodeType}):\NodeType=}
\STATE \tab\textbf{match} $e$ \textbf{with}
\STATE \tab\tab $``x" \Rightarrow \context(x)$
\STATE \tab\tab $``\op(e_1, \ldots, e_k)" \Rightarrow \constructor{\op}(\TranslateExpr(e_1,\context), \ldots, \TranslateExpr(e_k,\context))$ 
\finishfunction
\declarefunction{\PHI(n:\NodeType,\context_1: \maptype{\Var}{\NodeType}, \context_2:\maptype{\Var}{\NodeType}):\maptype{\Var}{\NodeType} =} 
\STATE \tab$\Combine(\context_1,\context_2,\lambda\ t\ f\ .\ \constructor{\phi}(n,t,f))$
\finishfunction
\declaresmallfunction{\THETA_{\loopdepth:\N}(\context_1: \maptype{\Var}{\NodeType}, \context_2:\maptype{\Var}{\NodeType}):\maptype{\Var}{\NodeType} =} 
\STATE \tab$\Combine(\context_1,\context_2,\lambda\ b\ n\ .\ \constructor{\theta_{\loopdepth}}(b,n))$
\finishfunction
\declaresmallfunction{\EVAL_{\loopdepth:\N}(\context: \maptype{\Var}{\NodeType}, n:\NodeType):\maptype{\Var}{\NodeType} =} 
\STATE \tab$\InitMap(\Keys(\context), \lambda v\ . \ \constructor{\eval_{\loopdepth}}(\context(v), n)$
\finishfunction
\declaresmallfunction{\Combine(m_1: \maptype{a}{b}, m_2:\maptype{a}{c},f: b*c\rightarrow d):\maptype{a}{d} =} 
\STATE \tab$\InitMap(\Keys(m_1) \cap \Keys(m_2), \lambda k.f(m_1[k],m_2[k]))$
\finishfunction
\thickstraightline
\end{algorithmic}
\caption{ML-style pseudo-code for converting \SIMPLE programs to \PEGs}
\label{simple-implementation}
\end{figure}

Here we use ML-style functional pseudo-code to describe the
translation from \SIMPLE programs to
\PEGs. Figure~\ref{simple-implementation} shows the entirety of the
algorithm, which uses a variety of simple types and data structures,
which we explain first.Note that if we use \SIMPLE programs to
instantiate our equality saturation approach described in
Figure~\ref{fig:optimize}, then the $\CfgToIr$ function from
Figure~\ref{fig:optimize} is implemented using a call to
$\TranslateProg$.

In the pseudo-code (and in all of Sections~\ref{sec:cfg2peg}
and~\ref{sec:peg2cfg}), we use the notation $\constructor{a}(n_1, \ldots,
n_k)$ to represent a \PEG node with label $a$ and children $n_1$
through $n_k$. Whereas previously the distinction between creating
\PEG nodes and applying functions was clear from context, in a
computational setting like pseudo-code we want to avoid confusion
between for example applying negation in the pseudo-code
$\neg(\ldots)$, vs.\ creating a \PEG node labeled with negation
$\constructor{\neg}(\ldots)$. For parameter nodes, there can't be any
confusion because when we write $param(x)$, $param$ is actually not a
function -- instead $param(x)$ as a whole is label. Still, to be
consistent in the notation, we use $\constructor{param}(x)$ for
constructing parameter nodes.

We introduce the concept of a \emph{node context} $\context$, which is a
set of bindings of the form $x:n$, where $x$ is a \SIMPLE variable,
and $n$ is a \PEG node. A node context states, for each variable $x$,
the PEG node $n$ that represents the current value of $x$. We use
$\context(x) = n$ as shorthand for $(x:n) \in \context$. Aside from using node
contexts here, we will also use them later in our type-directed
translation rules (Section~\ref{sec:type-directed-translation}). For
our pseudo-code, we implement node contexts as an immutable map data
structure that has the following operations defined on it

\begin{enumerate}[$\bullet$]
\item \textit{Map initialization.} The $\InitMap$ function is used to
  create maps. Given a set $K$ of keys and a function $f$ from $K$ to
  $D$, $\InitMap$ creates a map of type $\maptype{K}{D}$ containing,
  for every element $k \in K$, an entry mapping $k$ to $f(k)$.
\item \textit{Keys of a map.} Given a map $m$, $\Keys(m)$ returns the set
  of keys of $m$.
\item \textit{Map read.} Given a map $m$ and a key $k \in \Keys(m)$,
  $m(k)$ returns the value associated with key $k$ in $m$.
\item \textit{Map update.} Given a map $m$, $m[k\mapsto
  d]$ returns a new map in which key $k$ has been updated to map to
  $d$.
\end{enumerate}

The pseudo-code also uses the types $\Prog$, $\Stmt$, $\Expr$ and
$\Var$ to represent \SIMPLE programs, statements, expressions and
variables. Given a program $p$, $p.\params$ is a list of its parameter
variables, and $p.\body$ is its body statement. We use syntax-based
pattern matching to extract information from $\Stmt$ and $\Expr$ types
(as shown on lines 6,7 for statements, and 18, 19 for expressions).

\mypara{Expressions} We explain the pieces of this algorithm
one-by-one, starting with the $\TranslateExpr$ function on line
16. This function takes a \SIMPLE expression $e$ and a node context
$\context$ and returns the \PEG node corresponding to $e$. There are two
cases, based on what type of expression $e$ is.  Line 18 states that
if $e$ is a reference to variable $x$, then we simply ask $\context$ for
its current binding for $x$. Line 19 states that if $e$ is the
evaluation of operator $\op$ on arguments $e_1,\ldots,e_k$, then we
recursively call $\TranslateExpr$ on each $e_i$ to get $n_i$, and then
create a new \PEG node labeled $\op$ that has child nodes
$n_1,\ldots,n_k$.

\mypara{Statements}
Next we explore the $\TranslateStmt$ function on line 4, which takes a
\SIMPLE statement $s$, a node context $\context$, and a loop depth $\ell$
and returns a new node context that represents the changes $s$ made to
$\context$. There are four cases, based on the four statement types. 

Line 6 states that a sequence $s_1;s_2$ is simply the result of
translating $s_2$ using the node context that results from translating
$s_1$. Line 7 states that for an assignment $x:=e$, we simply update
the current binding for $x$ with the \PEG node that corresponds to
$e$, which is computed with a call to $\TranslateExpr$. 

Line 8 handles if-then-else statements by introducing $\phi$ nodes. We
recursively produce updated node contexts $\context_1$ and $\context_2$ for
statements $s_1$ and $s_2$ respectively, and compute the \PEG node
that represents the guard condition, call it $n_c$. We then create
\PEG $\phi$ nodes by calling the PHI function defined on line 20. This
function takes the guard node $n_c$ and the two node contexts $\context_1$
and $\context_2$ and creates a new $\phi$ node in the \PEG for each
variable that is defined in both node contexts. The true child for
each $\phi$ node is taken from $\context_1$ and the false child is taken
from $\context_2$, while all of them share the same guard node $n_c$. Note
that this is slightly inefficient in that it will create $\phi$ nodes
for all variables defined before the if-then-else statement, whether
they are modified by it or not. These can be easily removed, however,
by applying the rewrite $\phi(C,A,A) = A$.

Finally we come to the most complicated case on line 9, which handles
while loops. In line 10 we extract the set of all variables defined up
to this point, in the set $\vars$. We allocate a temporary
\PEG node for each item in $\vars$ on line 11, and bind them
together in the node context $\context_t$. We use $\TemporaryNode(v)$ to
refer to a temporary \PEG node named $v$, which is a new kind of node
that we use only for the conversion process. We then recursively
translate the body of the while loop using the context full of
temporary nodes on line 12. In the resulting context $\context^\prime$,
the temporary nodes act as placeholders for loop-varying values. Note
that here is the first real use of the loop depth parameter $\ell$,
which is incremented by 1 since the body of this loop will be at a
higher loop depth than the code before the loop. For every variable in
$\vars$, we create $\theta_{\ell+1}$ nodes using the THETA
function defined on line 22. This function takes node contexts $\context$
and $\context^\prime$, which have bindings for the values of each variable
before and during the loop, respectively. The binding for each
variable in $\context$ becomes the first child of the $\theta$ (the base
case) and the binding in $\context^\prime$ becomes the second child (the
inductive case). Unfortunately, the $\theta$ expressions we just
created are not yet accurate, because the second child of each
$\theta$ node is defined in terms of temporary nodes. The correct
expression should replace each temporary node with the new $\theta$
node that corresponds to that temporary node's variable, to ``close
the loop'' of each $\theta$ node. That is the purpose of the
$\FixpointTemporaries$ function called on line 14. For each variable
$v\in\vars$, $\FixpointTemporaries$ will rewrite
$\context_\theta(v)$ by replacing any edges to $\TemporaryNode(x)$ with
edges to $\context_\theta(x)$, yielding new node context
$\context_\theta^\prime$. Now that we have created the correct $\theta$
nodes, we merely need to create the $\eval$ and $\pass$ nodes to go
with them. Line 15 does this, first by creating the $\pass_{\ell+1}$
node which takes the break condition expression as its child. The
break condition is computed with a call to $\TranslateExpr$ on $e$,
using node context $\context_\theta^\prime$ since it may reference some of
the newly-created $\theta$ nodes. The last step is to create $\eval$
nodes to represent the values of each variable after the loop has
terminated. This is done by the EVAL function defined on line 24. This
function takes the node context $\context_\theta^\prime$ and the
$\pass_{\ell+1}$ node and creates a new $\eval_{\ell+1}$ node for each
variable in $\vars$. This final node context that maps each
variable to an $\eval$ is the return value of $\TranslateStmt$. Note
that, as in the case of if-then-else statements, we introduce an
inefficiency here by replacing all variables with $\eval's$, not just
the ones that are modified in the loop. For any variable $v$ that was
bound to node $n$ in $\context$ and not modified by the loop, its binding
in the final node context would be
$\eval_{\ell+1}(T,\pass_{\ell+1}(C))$, where $C$ is the guard
condition node and $T=\theta_{\ell+1}(n,T)$ (i.e. the $\theta$ node
has a direct self-loop). We can easily remove the spurious nodes by
applying a rewrite to replace the $\eval$ node with $n$.

\mypara{Programs}
The $\TranslateProg$ function on line 1 is the top-level call to
convert an entire \SIMPLE program to a \PEG. It takes a \SIMPLE
program $p$ and returns the root node of the translated \PEG. It
begins on line 2 by creating the initial node context which contains
bindings for each parameter variable. The nodes that correspond to the
parameters are opaque parameter nodes, that simply name the
parameter variable they represent. Using this node context, we
translate the body of the program starting at loop depth 0 on line
3. This will yield a node context that has \PEG expressions for the
final values of all the variables in the program. Hence, the root of
our translated \PEG will be the node that is bound to the special
return variable \texttt{retvar} in this final node context.

\begin{figure}[t]
\begin{center}
\includegraphics[width=6in]{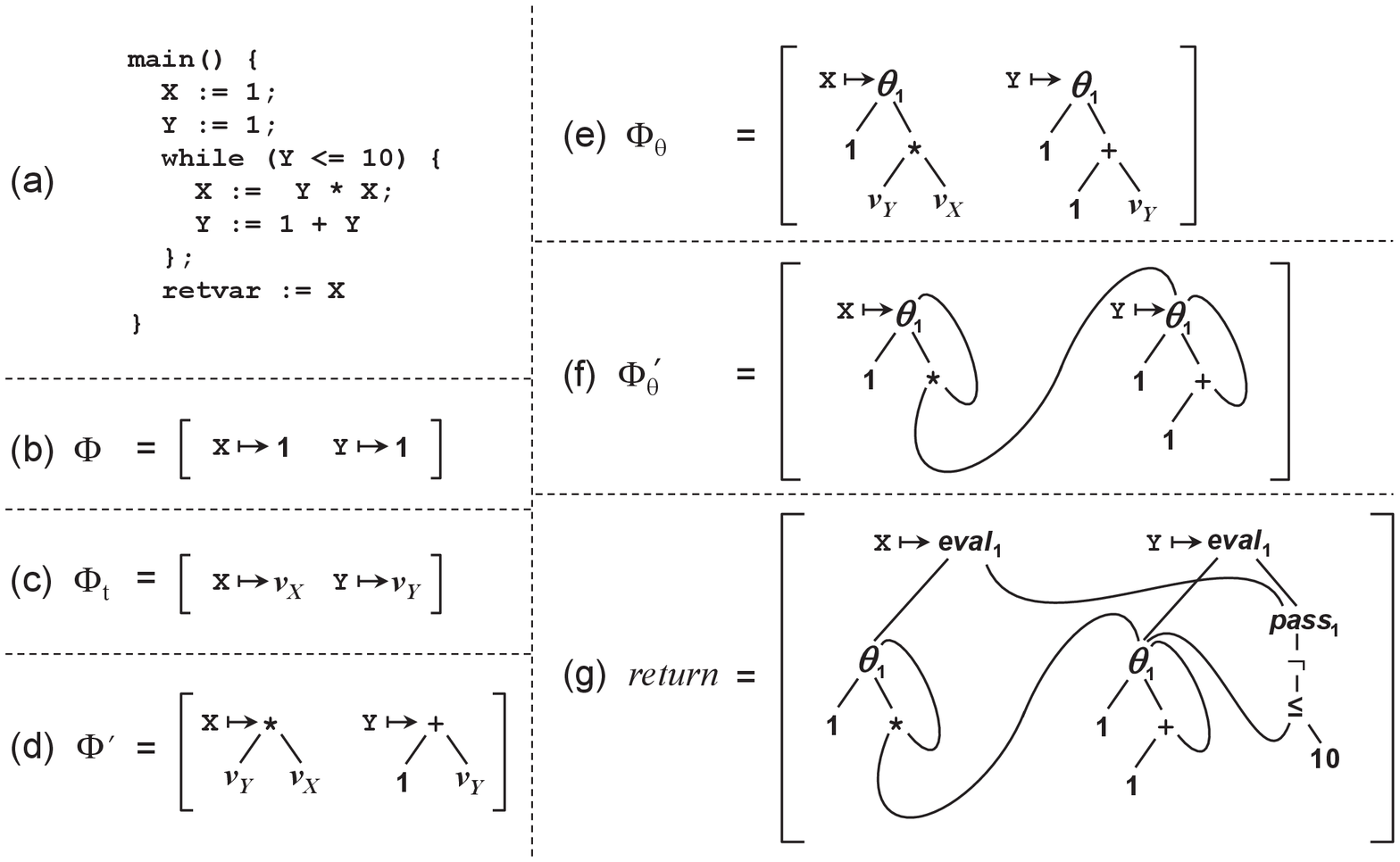}
\end{center}
\caption{Steps of the translation process. (a) shows a \SIMPLE program
  computing the factorial of 10; (b) through (f) show the contents of
  the variables in $\TranslateStmt$ when processing the \texttt{while}
  loop; and (g) shows the return value of $\TranslateStmt$ when
  processing the \texttt{while} loop.}
\label{fig:simple-to-peg-example}
\end{figure}

\mypara{Example} We illustrate how the translation process works on
the \SIMPLE program from Figure~\ref{fig:simple-to-peg-example}(a),
which computes the factorial of $10$. After processing the first two
statements, both \verb-X- and \verb-Y- are bound to the \PEG node
$1$. Then $\TranslateStmt$ is called on the \verb-while- loop, at
which point $\context$ maps both \verb-X- and \verb-Y- to
$1$. Figures~\ref{fig:simple-to-peg-example}(b)
through~\ref{fig:simple-to-peg-example}(g) show the details of
processing the \verb-while- loop. In particular, (b) through (f) show
the contents of the variables in $\TranslateStmt$, and (g) shows the
return value of $\TranslateStmt$. Note that in (g) the node labeled
$i$ corresponds to the $\pass$ node created on line 15 in
$\TranslateStmt$. After the loop is processed, the assignment to
$\retvar$ simply binds $\retvar$ to whatever \verb-X- is bound to in
Figure~\ref{fig:simple-to-peg-example}(g).

\begin{figure}[t]
{
\begin{tabular}{c}
\multicolumn{1}{l}{
\fbox{\raisebox{0pt}[8pt][1pt]{$\vdash p\;\; \vartriangleright \;\; n$ \;\; (programs)}} 
} \\[15pt]
\inference[Trans-Prog]{\typetransstmnt[0]{s}{\{x_1 : \tau_1, \dots, x_k : \tau_k\}}{\Gamma}{\{x_1 : \constructor{param}(x_1), \dots, x_k : \constructor{param}(x_k)\}}{\context} \\
       n = \context(\retvar) \mbox{ (well defined because $\Gamma(\retvar) = \tau$)}
       }{\vdash \texttt{main}(x_1 : \tau_1, \dots, x_k : \tau_k):\tau~\{s\} \;\; \vartriangleright \;\; n} \\[30pt]
\multicolumn{1}{l}{
\fbox{\raisebox{0pt}[8pt][1pt]{$\typetransstmnt{s}{\Gamma}{\Gamma'}{\context}{\context'}$ \;\; (statements)}} 
} \\[15pt]
\inference[Trans-Seq]{\typetransstmnt{s_1}{\Gamma}{\Gamma'}{\context}{\context'}
       & \typetransstmnt{s_2}{\Gamma'}{\Gamma''}{\context'}{\context''}
       }{\typetransstmnt{s_1; s_2}{\Gamma}{\Gamma''}{\context}{\context''}} \\[24pt]
\inference[Trans-Asgn]{\typetransexpr{e}{\Gamma}{\tau}{\context}{n}
       }{\typetransstmnt{x:=e}{\Gamma}{(\Gamma, x : \tau)}{\context}{(\context, x : n)}} \\[24pt]
\inference[Trans-If]{\typetransexpr{e}{\Gamma}{\texttt{bool}}{\context}{n} 
       \\ \typetransstmnt{s_1}{\Gamma}{\Gamma'}{\context}{\context_1}
       & \typetransstmnt{s_2}{\Gamma}{\Gamma'}{\context}{\context_2}
       \\ \{x : n_{(x,1)}\}_{x \in \Gamma'} = \context_1
       & \{x : n_{(x,2)}\}_{x \in \Gamma'} = \context_2
       \\ \context' = \{x : \constructor{\phi}(n, n_{(x,1)}, n_{(x,2)})\}_{x \in \Gamma'}
       }{\typetransstmnt{\simpleite{e}{s_1}{s_2}}{\Gamma}{\Gamma'}{\context}{\context'}} \\[24pt]
\inference[Trans-While]{\typetransexpr{e}{\Gamma}{\texttt{bool}}{\context}{n} 
       \\ \context = \{x : v_x\}_{x \in \Gamma} \mbox{ each $v_x$ fresh} 
       & \ell' = \ell + 1
       & \typetransstmnt[\ell']{s}{\Gamma}{\Gamma}{\context}{\context'}
       \\ \{x : n_x\}_{x \in \Gamma} = \context_0
       & \{x : n'_x\}_{x \in \Gamma} = \context'
       \\ \context_\infty = \{x : \constructor{\eval_{\ell'}}(v_x, \constructor{\pass_{\ell'}}(\constructor{\neg}(n)))\}_{x \in \Gamma} \mbox{ with each $v_x$ unified with $\constructor{\theta_{\ell'}}(n_x, n'_x)$}
       }{\typetransstmnt{\simplewhile{e}{s}}{\Gamma}{\Gamma}{\context_0}{\context_\infty}} \\[24pt]
\inference[Trans-Sub]{\typetransstmnt{s}{\Gamma}{\Gamma'}{\context}{\context'}
       \\ \{x : n_x\}_{x \in \Gamma'} = \context'
       & \context'' = \{x : n_x\}_{x \in \Gamma''} \mbox{ (well defined because $\Gamma'' \subseteq \Gamma'$)}
       }{\typetransstmnt{s}{\Gamma}{\Gamma''}{\context}{\context''}} \\[30pt]
\multicolumn{1}{l}{
\fbox{\raisebox{0pt}[8pt][1pt]{$\typetransexpr{e}{\Gamma}{\tau}{\context}{n}$ \;\; (expressions)}} 
} \\[15pt]
\inference[Trans-Var]{n = \context(x) \mbox{ (well defined because $\Gamma(x) = \tau$)}
       }{\typetransexpr{x}{\Gamma}{\tau}{\context}{n}} \\[20pt]
\inference[Trans-Op]{\typetransexpr{e_1}{\Gamma}{\tau_1}{\context}{n_1}
       & \ldots
       & \typetransexpr{e_k}{\Gamma}{\tau_k}{\context}{n_k}
       }{\typetransexpr{op(e_1, \dots, e_k)}{\Gamma}{\tau}{\context}{\constructor{op}(n_1, \dots, n_k)}}
\end{tabular}
}
\caption{Type-directed translation from \SIMPLE programs to \PEGs}
\label{simple-translate}
\end{figure}

\subsection{Type-directed translation} 
\label{sec:type-directed-translation}
In this section we formalize the translation process described by the
pseudo-code implementation with a type-directed translation from
\SIMPLE programs to \PEGs, in the style of~\cite{Leijen:mlftof}.
The type-directed translation in Figure~\ref{simple-translate} is more
complicated than the implementation in
Figure~\ref{simple-implementation}, but it makes it easier to prove
the correctness of the translation. For example, the implementation
uses maps from variables to \PEG nodes, and at various points queries
these maps (for example, line 18 in Figure~\ref{simple-implementation}
queries the $\context$ map for variable $x$). The fact that these map
operations never fail relies on implicit properties which are tedious
to establish in Figure~\ref{simple-implementation}, as they rely on
the fact that the program being translated is well-typed. In the
type-directed translation, these properties are almost immediate since
the translation operates on an actual proof that the program is
well-typed. 

In fact, the rules in Figure~\ref{simple-translate} are really
representations of each case in a constructive total deterministic
function defined inductively on the proof of well-typedness. Thus when
we use the judgment $\typetransexpr{e}{\Gamma}{\tau}{\context}{n}$ as an
assumption, we are simply binding $n$ to the result of this
constructive function applied to the proof of $\Gamma \vdash e : \tau$
and the PEG context $\context$. Likewise, we use the judgment
$\typetransstmnt{s}{\Gamma}{\Gamma'}{\context}{\context'}$ to bind $\context'$ to
the result of the constructive function applied to the proof of
$\Gamma \vdash s : \Gamma'$ and the PEG context $\context$. Here we
explain how this type-directed translation works.

\mypara{Expressions} The translation process for an expression $e$
takes two inputs: (1) a derivation showing the type correctness of
$e$, and (2) a node context $\context$. The translation process produces
one output, which is the node $n$ that $e$ translates to. We formalize
this with a judgment $\typetransexpr{e}{\Gamma}{\tau}{\context}{n}$, which
states that from a derivation of $\typeexpr{\Gamma}{e}{\tau}$, and a
node context $\context$ (stating what node to use for each variable), the
translation produces node $n$ for expression $e$. For example,
consider the Trans-Op rule, which is used for a primitive operation
expression. The output of the translation process is a new \PEG node
with label $\op$, where the argument nodes $n_1 \ldots n_k$ are
determined by translating the argument expressions $e_1 \ldots e_k$.

The Trans-Var rule returns the \PEG node associated with the variable $x$ in $\context$.
The definition $n = \context(x)$ is well defined because we maintain the invariant that \hbox{$\typetransexpr{e}{\Gamma}{\tau}{\context}{n}$} is only used with contexts $\Gamma$ and $\context$ that are defined on precisely the same set of variables.
Thus, because the Type-Var rule requires $\Gamma(x) = \tau$, $\Gamma$ must be defined on $x$ and so we know $\context$ is also defined on $x$.

Note that a concrete implementation of the translation, like the one
in Figure~\ref{simple-implementation}, would explore a derivation of
\hbox{$\typetransexpr{e}{\Gamma}{\tau}{\context}{n}$} bottom-up: the
translation starts at the bottom of the derivation tree and makes
recursive calls to itself, each recursive call corresponding to a step
up in the derivation tree. Also note that there is a close relation
between the rules in Figure~\ref{simple-translate} and those in
Figure~\ref{simple-types}. In particular, the formulas on the left of
the $\vartriangleright$ correspond directly to the typing rules from
Figure~\ref{simple-types}.

\mypara{Statements} The translation for a statement $s$ takes as input
a derivation of the type-correctness of $s$, a node context capturing
the translation that has been performed up to $s$, and returns the
node context to be used in the translation of statements following
$s$. We formalize this with a judgment
$\typetransstmnt{s}{\Gamma}{\Gamma'}{\context}{\context'}$, which states that
from a derivation of $\typestmnt{\Gamma}{s}{\Gamma'}$, and a node
context $\context$ (stating what node to use for each variable in $s$),
the translation produces an updated node context $\context'$ after
statement $s$ (ignore $\ell$ for now). For example, the rule
Trans-Asgn updates the node context to map variable $x$ to the node
$n$ resulting from translating $e$ (which relies on the fact that $e$
is well typed in type context $\Gamma$).

Again, we maintain the invariant that in all the derivations we
explore, the judgment \hbox{$\typetransexpr{s}{\Gamma}{\Gamma'}{\context}{\context'}$} is
only used with contexts $\Gamma$ and $\context$ that are defined on
precisely the same set of variables, and furthermore the resulting
contexts $\Gamma'$ and $\context'$ will always be defined on the same
set of variables (although potentially different from $\Gamma$ and
$\context$).  It is fairly obvious that the rules preserve this
invariant, although Trans-Sub relies on the fact that $\Gamma''$ must
be a subcontext of $\Gamma'$. The Trans-Seq and Trans-Asgn rules are
self explanatatory, so below we discuss the more complicated rules for
control flow.

The rule Trans-If describes how to translate if-then-else statements
in \SIMPLE programs to $\phi$ nodes in \PEGs.  First, it translates
the Boolean guard expression $e$ to a node $n$ which will later be
used as the condition argument for each $\phi$ node.  Then it
translates the statement $s_1$ for the ``true'' branch, producing a
node context $\context_1$ assigning each live variable after $s_1$ to
a \PEG node representing its value after $s_1$.  Similarly, it
translates $s_2$ for the ``false'' branch, producing a node context
$\context_2$ representing the ``false'' values of each variable. Due
to the invariant we maintain, both $\context_1$ and $\context_2$ will
be defined on the same set of variables as $\Gamma'$. For each $x$
defined in $\Gamma'$, we use the name $n_{(x,1)}$ to represent the
``true'' value of $x$ after the branch (taken from $\context_1$) and
$n_{(x,2)}$ to represent the ``false'' value (taken from
$\context_2$). Finally, the rule constructs a node context $\context'$
which assigns each variable $x$ defined in $\Gamma'$ to the node
$\constructor{\phi}(n, n_{(x,1)}, n_{(x,2)})$, indicating that after
the if-then-else statement the variable $x$ has ``value'' $n_{(x,1)}$
if $n$ evaluates to true and $n_{(x,2)}$ otherwise.  Furthermore, this
process maintains the invariant that the type context $\Gamma'$ and
node context $\context'$ are defined on exactly the same set of
variables.

The last rule, Trans-While, describes how to translate while loops in
\SIMPLE programs to combinations of $\theta$, $\eval$, and $\pass$
nodes in \PEGs.  The rule starts by creating a node context $\context$
which assigns to each variable $x$ defined in $\Gamma$ a fresh
temporary variable node $v_x$.  The clause $\ell' = \ell + 1$ is used
to indicate that the body of the loop is being translated at a higher
loop depth.  In general, the $\ell$ subscript in the notation
$\transstmnt{\context}{s}{\context'}$ indicates the loop depth of the
translation.  Thus, the judgment
$\typetransstmnt[\ell']{s}{\Gamma}{\Gamma'}{\context}{\context'}$ translates
the body of the loop $s$ at a higher loop depth to produce the node
context $\context'$.  The nodes in $\context'$ are in terms of the temporary
nodes $v_x$ in $\Gamma$ and essentially represent how variables change
in each iteration of the loop.  Each variable $x$ defined in $\Gamma$
has a corresponding node $n_x$ in the node context $\context_0$ from
before the loop, again due to the invariants we maintain that $\Gamma$
and $\context$ are always defined on the same set of variables.  This
invariant also guarantees that each variable $x$ defined in $\Gamma$ also has a
corresponding node $n'_x$ in the node context $\context'$.  Thus, for each
such variable, $n_x$ provides the base value and $n'_x$ provides the
iterative value, which can now combined using a $\theta$ node.  To
this end, we unify the temporary variable node $v_x$ with the node
$\constructor{\theta_{\ell'}}(n_x, n'_x)$ to produce a recursive
expression which represents the value of $x$ at each iteration of the
loop. Lastly, the rule constructs the final node context $\context_\infty$
by assigning each variable $x$ defined in $\Gamma$ to the node
$\constructor{\eval_{\ell'}}(v_x,
\constructor{\pass_{\ell'}}(\constructor{\neg}(n)))$ (where $v_x$ has been
unified to produce the recursive expression for $x$).  The node
$\constructor{\neg}(n)$ represents the break condition of the loop; thus
$\constructor{\pass_{\ell'}}(\constructor{\neg}(n))$ represents the number of
times the loop iterates.  Note that the same $\pass$ node is used for
each variable, whereas each variable gets its own $\theta$ node. In
this manner, the rule Trans-While translates while loops to \PEGs, and
furthermore preserves the invariant that the type context $\Gamma$ and
node context $\context_\infty$ are defined on exactly the same set of
variables.

\mypara{Programs} Finally, the rule Trans-Prog shows how to use the
above translation technique in order to translate an entire \SIMPLE
program.  It creates a node context with a \PEG parameter node for
each parameter to \texttt{main}. It then translates the body of the
\texttt{main} at loop depth $0$ to produce a node context $\context$.
Since the return $\retvar$ is guaranteed to be in the final context
$\Gamma$, the invariant that $\Gamma$ and $\context$ are always defined on
the same variables ensure that there is a node $n$ corresponding to
$\retvar$ in the final node context $\context$.  This \PEG node $n$
represents the entire \SIMPLE program.

\mypara{Translation vs.\ pseudo-code}

The pseudo-code in Figure~\ref{simple-implementation} follows the
rules from Figure~\ref{simple-translate} very closely. Indeed, the
code can be seen as using the rules from the type-directed translation
to find a derivation of $\vdash p\;\; \vartriangleright \;\;
n$. The search starts at the end of the derivation tree, and moves up
from there. The entry-point of the pseudo-code is $\TranslateProg$,
which corresponds to rule Trans-Prog, the last rule to be used in a
derivation of $\vdash p\;\; \vartriangleright \;\;
n$. $\TranslateProg$ calls $\TranslateStmt$, which corresponds to
finding a derivation for
$\typetransstmnt[\ell]{s}{\Gamma}{\Gamma'}{\context}{\context'}$. Finally,
$\TranslateStmt$ calls $\TranslateExpr$, which corresponds to finding
a derivation for $\typetransexpr{e}{\Gamma}{\tau}{\context}{n}$.  Each
pattern-matching case in the pseudo-code corresponds to a rule from
the type-directed translation.

The one difference between the pseudo-code and the type-directed
translation is that in the judgments of the type-directed translation,
one of the inputs to the translation is a derivation of the type
correctness of the expression/statement/program being translated,
whereas the pseudo-code does not manipulate any derivations. This can
be explained by a simple erasure optimization in the pseudo-code:
because of the structure of the type-checking rules for \SIMPLE (in
particular there is only one rule per statement kind), the
implementation does not need to inspect the entire derivation -- it
only needs to look at the final expression/statement/program in the
type derivation (which is precisely the expression/statement/program
being translated). It is still useful to have the derivation expressed
formally in the type-directed translation, as it makes the proof of
correctness more direct. Furthermore, there are small changes that can
be made to the \SIMPLE language that prevent the erasure optimization
from being performed. For example, if we add subtyping and implicit
coercions, and we want the \PEG translation process to make coercions
explicit, then the translation process would need to look at the type
derivation to see where the subtyping rules are applied.

Because the type-directed translation in Figure~\ref{simple-translate}
is essentially structural induction on the proof that the \SIMPLE
program is well typed, we can guarantee that its implementation in
Figure~\ref{simple-implementation} terminates.  Additionally, because
of the invariants we maintain in the type-directed translation, we can
guarantee that the implementation always successfully produces a
translation. We discuss the correctness guarantees provided by the
translation below.

\subsection{Preservation of Semantics}
\label{sec:preservation-of-semantics}

While \SIMPLE is a standard representation of programs, \PEGs are far
from standard.  Furthermore, the semantics of \PEGs are even less so,
especially since the node representing the returned value is the first
to be ``evaluated''.  Thus, it is natural to ask whether the
translation above preserves the semantics of the specified programs.
We begin by defining the semantics of \SIMPLE programs, and go on to
examine their relationship to the semantics of \PEGs produced by our
algorithm.

Here we define the evaluation functions $\semvalue{\cdot}$, which
implement the operational semantics of \SIMPLE programs. We don't give
full definitions, since these are standard. These functions are
defined in terms of an \emph{evaluation context $\Sigma$}, which is a
map that for each variable $x$ gives its value $\nu = \Sigma(x)$.

\begin{defn}[Semantics of expressions]
  For a \SIMPLE expression $e$ we define $\semvalue{e}$ to be a
  partial function from evaluation contexts to values. This
  represents the standard operational semantics for \SIMPLE
  expressions. For a given $\Sigma$, $\semvalue{e}(\Sigma)$ returns
  the result of evaluating $e$, using the values of variables
  given in $\Sigma$.
\end{defn}
\begin{defn}[Semantics of statements]
  For a \SIMPLE statement $s$ we define $\semvalue{s}$ to be a
  partial function from evaluation contexts to evaluation contexts.
  This represents the standard operational semantics for \SIMPLE
  statements. For a given $\Sigma$, $\semvalue{s}(\Sigma)$ returns the
  evaluation context that results from executing $s$ in context
  $\Sigma$. If $s$ does not terminate when started in $\Sigma$, then
  $\semvalue{s}(\Sigma)$ is not defined.
\end{defn}
\begin{defn}[Semantics of programs]
  For a \SIMPLE program
  \hbox{$\mathbf{main}(x_1\!:\tau_1,\ldots,x_k\!:\tau_k)\{s\}$} and an evaluation
  context $\Sigma$ that maps each $x_i$ to an appropriately typed
  value, we define $\semvalue{\mathbf{main}}(\Sigma) =
  \semvalue{s}(\Sigma)(\mathtt{retvar})$.
\end{defn}

We will use these functions in our discussion below. For the
translation defined in Section~\ref{sec:type-directed-translation} we
have proven the following theorem (recall the substitution notation
$n[x \mapsto \nu]$ from Section~\ref{sec:formalization}).

\begin{theorem}
If \emph{(1)} $\vdash \mathbf{main}(x_1 : \tau_1, \dots, x_k :
\tau_k):\tau~\{s\} \;\; \vartriangleright \;\; n$, \emph{(2)}
$\Sigma$ is an evaluation context mapping each $x_i$ to an
appropriately typed value $\nu_i$, and \emph{(3)}~\hbox{$\hat{n} =
n[x_1\mapsto\nu_1, \dots, x_k \mapsto \nu_k]$}, then
$\semvalue{\mathbf{main}}(\Sigma) = \nu \implies
\semvalue{\hat{n}}(\lambda \ell. 0) = \nu$.
\end{theorem}

The above theorem only states that our conversion process is
\emph{nearly} se\-man\-tics-pre\-serv\-ing, since it does not perfectly
preserve non-termination. In particular, our translation from \SIMPLE
to \PEG discards any \PEG nodes which are never used to calculate the
return value. Thus, an infinite \SIMPLE loop whose value is never used
will be removed, changing the termination behavior of the program. In
the broader setting beyond \SIMPLE, the only case where we would
change the termination behavior of the program is if there is an
infinite loop that causes no side effects (aside from non-termination)
and does not contribute to the return value of the function. It is
important to keep in mind that these loops have \emph{no} side-effects
(aside from non-termination), and so they cannot modify the heap or
perform I/O. This basically means that these loops are equivalent to a
\verb-while(true) { }- loop. Other modern compilers perform similar
transformations that remove such IO-less infinite loops which do not
contribute to the result~\cite{terminate-blog}. In fact, the newly
planned C++ standard allows the implementation to remove such IO-less
loops, even if termination cannot be
proven~\cite{CPlusPlus0X}. Nonetheless, at the end of this section, we
give a brief overview of how to encode non-termination in \PEGs.

In this theorem, we use both the function $\semvalue{\cdot}$ defined
above for \SIMPLE programs, as well as the function $\semvalue{\cdot}$
defined in Equation~\ref{eq:semantic-value-defn} for \PEG
nodes. Throughout the rest of this section we will mix our uses of the
various $\semvalue{\cdot}$ functions, and the reader can disambiguate
them based on context. The common intuition is that these functions
all implement program semantics, and thus represent executing the
program fragment they are called upon.

We proved the above theorem in full formal detail using the Coq interactive theorem prover~\cite{CoqArt}.
To conduct the proof in Coq we only had to assume the standard axioms for extensional equality of functions and of coinductive types.
The machine-checkable Coq proof is available at: \url{http://cseweb.ucsd.edu/groups/progsys/peg-coq-proof}.
Here we present only the major invariants used in this proof without showing the details of why these invariants are preserved.

Given a program \texttt{main} with parameters
$x_1,\ldots,x_k$, suppose we are given a set of actual values
$v_1,\ldots,v_k$ that correspond to the values passed in for those
parameters. Then given the \PEG $G$ created during the translation of
\texttt{main}, we can construct a new \PEG $\hat{G}$ that is identical
to $G$ except that every parameter node for $x_i$ is replaced with a
constant node for $v_i$. Thus, for every node $n \in G$, there is a
corresponding node $\hat{n}\in \hat{G}$. Furthermore, this
correspondence is natural: $\constructor{\theta}$ nodes correspond to
$\constructor{\theta}$ nodes and so on (except for the parameter nodes
which have been explicitly replaced). Similarly, every \PEG context
$\context$ for $G$ has a naturally corresponding \PEG context
$\hat{\context}$ in terms of nodes in $\hat{G}$.  We can now phrase
our primary invariants in terms of this node correspondence.

In this proof we will rely on the concept of loop-invariance of \PEG
nodes. Earlier in Section~\ref{sec:formalization}, we defined some
simple rules for determining when a node $n$ is invariant with respect
to a given loop-depth $\ell$, which we denote as $\sliwrt{n}{\ell}$.
These rules are based on the syntax of the \PEG rather than the
semantics, so we say that the rules detect \emph{syntactic
  loop-invariance}, rather than semantic loop-invariance. Syntactic
loop-invariance is a useful property since it implies semantic
loop-invariance, which is in general undecidable.  We can generalize
the notion of $\invariant_\ell$ to a \PEG context as follows.

\begin{defn}
  Given a \PEG context $\context$ and a loop-depth $\ell$, we say that
  $\context$ is syntactically loop-invariant with respect to $\ell$ if for
  each binding $(x : n) \in \context$, $n$ is syntactically loop-invariant
  with respect to $\ell$. We denote this by $\sliwrt{\context}{\ell}$.
\end{defn}

With this definition in mind, we can express the first two lemmas that
will help in our proof of semantics preservation.

\begin{lemma} 
\label{lemma:invphi-invn}
If $\typetransexpr{e}{\Gamma}{\tau}{\context}{n}$, then $\forall \ell,
\sliwrt{\hat{\context}}{\ell} \implies \sliwrt{\hat{n}}{\ell}$
\end{lemma}

\begin{proof}
Proved using induction on the proof of $\Gamma \vdash e : \tau$.
\end{proof}

\begin{lemma} 
\label{lemma:invphi-invphiprime}
For all loop-depths $\ell$,  if $\typetransstmnt{s}{\Gamma}{\Gamma'}{\context}{\context'}$, then
\[\forall \ell'>\ell, \sliwrt{\hat{\context}}{\ell'} \implies \forall \ell'>\ell, \sliwrt{\hat{\context}'}{\ell'}\]
\end{lemma}

\begin{proof}
Proved using induction on the proof of $\Gamma \vdash s : \Gamma'$,
with Lemma~\ref{lemma:invphi-invn}.
\end{proof}

Using the above lemmas, and the fact that $\invariant(\cdot)$ implies
semantic loop-invariance, we can proceed to the critical invariant.
First we must introduce the notion of the semantics of a \PEG
context. Given a \PEG context $\context$, there is a unique evaluation
context that is induced by $\context$ for a given loop vector
$\mathbf{i}$. Namely, it is the evaluation context that maps every
variable $x$ to the value $\semvalue{\context(x)}(\mathbf{i})$. This
provides a useful relationship between the semantics of \PEGs and the
semantics of \SIMPLE programs.

\begin{defn}
  Given \PEG context $\context$, we define $\semvalue{\context}$ to be a
  partial function from loop vectors to evaluation contexts defined by
  \[ \forall \mathbf{i}~.~ \semvalue{\context}(\mathbf{i}) = \{(x:v)~|~v=\semvalue{\context(x)}\} \]
\end{defn}

\begin{lemma} 
\label{lemma:evale-evaln}
If $\typetransexpr{e}{\Gamma}{\tau}{\context}{n}$ and $\corresponds{\hat{\context}}{\Sigma}{\mathbf{i}}$, then
\[ \forall \nu~.~ \semvalue{e}(\Sigma)=\nu ~\implies~ \hat{n}(\mathbf{i})=\nu \]
\end{lemma}

\begin{proof}
Proved using induction on the proof of $\Gamma \vdash e : \tau$, with
Lemma~\ref{lemma:invphi-invn}.
\end{proof}

\begin{lemma}
\label{lemma:evals-evalphi}
For any loop-depth $\ell$, if \emph{(1)}
$\typetransstmnt{s}{\Gamma}{\Gamma'}{\context}{\context'}$, and \emph{(2)} for each
$\ell'>\ell, \sliwrt{\hat{\context}}{\ell'}$ holds, and \emph{(3)}
$\corresponds{\hat{\context}}{\Sigma}{\mathbf{i}}$, then
\[ \forall \Sigma'~.~ \semvalue{s}(\Sigma)=\Sigma' ~\implies~ \corresponds{\hat{\context}'}{\Sigma'}{\mathbf{i}} \]
\end{lemma}

\begin{proof}
Proved using induction on the proof of $\Gamma \vdash s : \Gamma'$,
with Lemmas~\ref{lemma:invphi-invphiprime} and
\ref{lemma:evale-evaln}. For the \textbf{while} case, the proof relies
on the fact that syntactic invariance of $\hat{\context}$ implies semantic
invariance of $\hat{\context}$, which implies that $\hat{\context}$
corresponds to $\Sigma$ at all loop vectors $\mathbf{i}'$ which only
differ from $\mathbf{i}$ at loop-depth $\ell + 1$.
\end{proof}

Our semantics-preservation theorem is a direct corollary of the above
lemma, and so we have shown that the evaluation of a \PEG is
equivalent to the evaluation of its corresponding \SIMPLE program,
modulo termination.

\mypara{Preserving Non-Termination}

There is a simple change to \PEGs that would allow them to preserve
non-termination, even for loops that don't contribute to the result.
In particular, we can use the concept of an effect token. For our task
of preserving non-termination, the effect token will encode the
non-termination effect, although one can use a similar strategy for
other effects (and in fact we use a similar strategy for encoding heap
modifications and exceptions using a $\sigma$ heap summary node in
Section~\ref{sec:implementation}). Any effectful operation must take
an effect token. If the operation might also change the state of the
effect, then it must produce an effect token as output. In \SIMPLE,
the $\div$ operation could modify our non-termination effect token, if
we choose to encode division-by-zero as non-termination (since \SIMPLE
does not contain exceptions). In we added functions to \SIMPLE, then
function calls would also consume and produce a non-termination effect
token, since the function call could possibly not terminate.

For every loop, there is one $\pass$ node (although there may be many
$\eval$ nodes), and evaluation of a $\pass$ node fails to terminate if
the condition is never true. As a result, since a $\pass$ node may
fail to terminate, it therefore must take and produce a
non-termination effect token. In particular, we would modify the
$\pass$ node to take two inputs: a node representing the break
condition of the loop and a node representing how the state of the
effect changes inside the loop. The $\pass$ node would also be
modified to have an additional output, which is the state of the
effect at the point when the loop terminates. Then, the translation
process is modified to thread this effect token through all the
effectful operations (which includes $\pass$ nodes), and finally
produce a node representing the effect of the entire function. This
final effect node is added as an output of the function, along with
the node representing the return value.

Whereas before an infinite loop would be elided if it does not
contribute to the final return value of the program, now the $\pass$
node of the loop contributes to the result because the effect token is
threaded through it and returned. This causes the $\pass$ node to
be evaluated, which causes the loop's break condition to be evaluated,
which will lead to non-termination since the condition is never
true.

\begin{figure}
\begin{center}
\includegraphics[width=\textwidth]{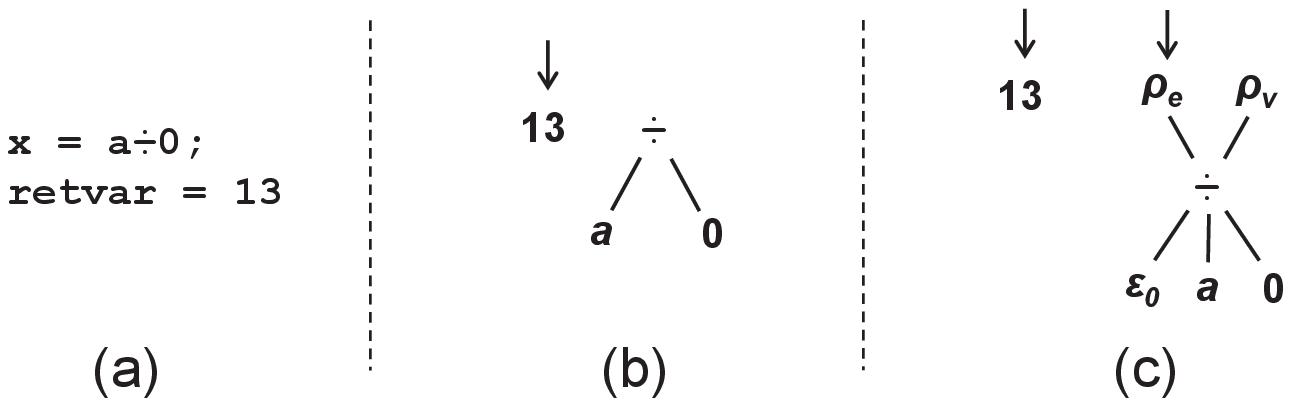}
\end{center}
\caption{Representation of a division by zero: (a) the original source
  code, (b) the \PEG produced without effect tokens, (c) the \PEG
  produced with effect tokens. The arrows indicate the return values.}
\label{fig:divide-by-zero}
\end{figure}

We present an example in Figure~\ref{fig:divide-by-zero} to
demonstrate the use of effect tokens. The \SIMPLE code in part (a)
shows an attempt to divide by zero, followed by a return value of
$13$. Let's assume that $\div$ does not terminate when dividing by
0. A similar encoding works if $\div$ throws an exception instead of
running forever -- we describe our encoding of exceptions in
Section~\ref{sec:implementation}, but in short the example would look
essentially the same, except that we would use a $\sigma$ heap summary
node instead of an effect token. Part (b) shows the corresponding \PEG
without effect tokens. The arrow indicates the return value, which is
$13$. Even though this \PEG has the nodes for the division by zero,
they are not reachable from the return value, and hence the \PEG would
be optimized to $13$, which would remove the divide-by-zero (thus
changing the termination behavior of the code).

Using a non-termination effect token can fix this problem, by
producing the \PEG in part (c). The division operator now returns a
tuple of $(\mathit{effect},\mathit{value})$ and the components are
fetched using $\rho_e$ and $\rho_v$ respectively. As previously, we
have $13$ as a return value, but we now have an additional return
value: the effect token produced by the \PEG. Since the division is
now reachable from the return values, it is not removed anymore, even
if the value of the division (the $\rho_v$ node) is never used.

\section{Reverting a \PEG to Imperative Code}
\label{sec:peg2cfg}

In this section we present the complement to the previous section: a
procedure for converting a \PEG back to a \SIMPLE program. Whereas the
translation from \SIMPLE programs to \PEGs was fairly simple, the
translation back (which we call reversion) is much more
complicated. Since the order of \SIMPLE execution is specified
explicitly, \SIMPLE programs have more structure than \PEGs, and so it
is not surprising that the translation from \PEGs back to \SIMPLE is
more complex than the translation in the forward direction. Because of
the complexity involved in reversion, we start by presenting a
simplified version of the process in
Sections~\ref{sec:well-typed-pegs}
through~\ref{sec:reversion-seq}. This simple version is correct but
produces \SIMPLE programs that are inefficient because they contain a
lot of duplicated code. We then present in Sections~\ref{sec:loop-fusion}
through~\ref{sec:loop-invariant-code-motion} several optimizations on
top of the simple process to improve the quality of the generated code
by removing the code duplication. These optimizations include branch
fusion, loop fusion, hoisting code that is common to both sides of a
branch, and hoisting code out of loops. In the setting of \SIMPLE,
these optimizations are optional -- they improve the performance of
the generated code, but are not required for correctness. However, if
we add side-effecting operations like heap reads/writes (as we do in
Section~\ref{sec:implementation}), these optimizations are not
optional anymore: they are needed to make sure that we don't
incorrectly duplicate side-effecting operations.
In our technical report~\cite{peg2cfg},] we present more advanced techniques for reverting \PEGs to CFGs rather than \SIMPLE programs, taking advantage of the flexibility of CFGs in order to produce even more efficient translations of \PEGs. This is particularly important when reverting \PEGs translated from programs in a language with more advanced control flow such as \texttt{break}s and \texttt{continue}s.

\subsection{CFG-like PEGs}
\label{sec:well-typed-pegs}

Before we can proceed, we need to define the precondition of our
reversion process. In particular, our reversion process assumes that
the \PEGs we are processing are \emph{CFG-like}, as formalized by the
following definition.

\begin{defn}[CFG-like \PEG context]
We say that a \PEG context $\context$ is CFG-like if $\Gamma \vdash \context
: \Gamma'$ using the rules in Figure~\ref{peg-types}.
\end{defn}

\begin{figure}[t]
{
\begin{tabular}{c}
\multicolumn{1}{l}{
\fbox{\raisebox{0pt}[8pt][1pt]{$\Gamma \vdash \context : \Gamma'$}} 
} \\[15pt]
\inference[Type-PEG-Context]{\forall (x : \tau) \in \Gamma'. \;\; \context(x) = n
  \Rightarrow \Gamma \vdash n : \tau}{\Gamma \vdash \context : \Gamma'}\\[20pt]

\multicolumn{1}{l}{
\fbox{\raisebox{0pt}[8pt][1pt]{$\Gamma \vdash n : \tau$}} 
} \\[15pt]
\inference[Type-PEG]{\Gamma, 0, \varnothing \vdash n : \tau
       }{\Gamma \vdash n : \tau} \\[20pt]
\multicolumn{1}{l}{
\fbox{\raisebox{0pt}[8pt][1pt]{$\Gamma, \ell, \Theta \vdash n : \tau$}} 
} \\[15pt]
\inference[Type-Param]{\Gamma(x) = \tau
       }{\Gamma, \ell, \Theta \vdash \overline{param}(x) : \tau} \\[20pt]
\inference[Type-Op]{op : (\tau_1, \dots, \tau_n) \to \tau
       & \Gamma, \ell, \Theta \vdash n_1 : \tau_1 \;\; \dots \;\; \Gamma, \ell, \Theta \vdash n_n : \tau_n
       }{\Gamma, \ell, \Theta \vdash \overline{op}(n_1, \dots, n_n) : \tau} \\[20pt]
\inference[Type-Phi]{\Gamma, \ell, \Theta \vdash c : \texttt{bool}
       & \Gamma, \ell, \Theta \vdash t : \tau
       & \Gamma, \ell, \Theta \vdash f : \tau
       }{\Gamma, \ell, \Theta \vdash \overline{\phi}(c, t, f) : \tau} \\[20pt]
\inference[Type-Theta]{\ell' = \ell - 1
       \\ \Gamma, \ell', \Theta \vdash b : \tau
       & \Gamma, \ell, (\Theta, (\ell \vdash \overline{\theta_\ell}(b, n) : \tau)) \vdash n : \tau
       }{\Gamma, \ell, \Theta \vdash \overline{\theta_\ell}(b, n) : \tau} \\[20pt]
\inference[Type-Eval-Pass]{\ell = \ell' + 1
       & \forall \ell_1, n, \tau'.[(\ell_1 \vdash n : \tau') \in \Theta \Rightarrow \ell_1 < \ell']
       \\ \Gamma, \ell', \Theta \vdash v : \tau
       & \Gamma, \ell', \Theta \vdash c : \texttt{bool}
       }{\Gamma, \ell, \Theta \vdash \overline{\eval_{\ell'}}(v, \overline{\pass_{\ell'}}(c)) : \tau} \\[20pt]
\inference[Type-Reduce]{\ell \geq \ell'
       & \Theta \supseteq \Theta'
       & \Gamma, \ell', \Theta' \vdash n : \tau
       }{\Gamma, \ell, \Theta \vdash n : \tau} \\[20pt]
\inference[Type-Assume]{(\ell \vdash n : \tau) \in \Theta
       }{\Gamma, \ell, \Theta \vdash n : \tau}
\end{tabular}
}
\caption{Rules for defining CFG-like PEGs}
\label{peg-types}
\end{figure}

The rules in Figure~\ref{peg-types} impose various restrictions on the
structure of the \PEG which makes our reversion process simpler. For
example, these rules guarantee that the second child of an $\eval$
node is a $\pass$ node, and that by removing the second outgoing edge
of each $\theta$ node, the \PEG becomes acyclic. If a \PEG context is
CFG-like, then it is well-formed (Definition~\ref{defn:well-formed}
from Section~\ref{sec:formalization}). Furthermore, all \PEGs produced
by our \SIMPLE-to-\PEG translation process from
Section~\ref{sec:cfg2peg} are CFG-like.  However, not all
well-formed \PEGs are CFG-like, and in fact it is useful for equality
saturation to consider \PEGs that are not CFG-like as intermediate
steps. To guarantee that the reversion process will work during
optimization, the Pseudo-boolean formulation described in
Section~\ref{sec:heuristic} ensures that the \PEG selected for
reversion is CFG-like.

In Figure~\ref{peg-types}, $\Gamma$ is a type context assigning
parameter variables to types.  $\ell$ is the largest loop-depth with
respect to which the \PEG node $n$ is allowed to be loop-variant;
initializing $\ell$ to $0$ requires $n$ to be loop-invariant with
respect to all loops.  $\Theta$ is an assumption context used to
type-check recursive expressions. Each assumption in $\Theta$ has the
form $\ell \vdash n:\tau$, where $n$ is a $\theta_\ell$ node;
$\ell \vdash n:\tau$ states that $n$ has type $\tau$ at loop depth
$\ell$. Assumptions in $\Theta$ are introduced in the Type-Theta rule,
and used in the Type-Assume rule. The assumptions in $\Theta$ prevent
``unsolvable'' recursive PEGs such as $x = 1 + x$ or ``ambiguous''
recursive PEGs such as $x = 0 * x$.  The requirement in Type-Eval-Pass
regarding $\Theta$ prevents situations such as $x =
\overline{\theta_2}(\eval_1(\eval_2(x, \dots), \dots), \dots)$, in
which essentially the initializer for the nested loop is the final
result of the outer loop. 

\begin{figure}
\includegraphics[width=\textwidth]{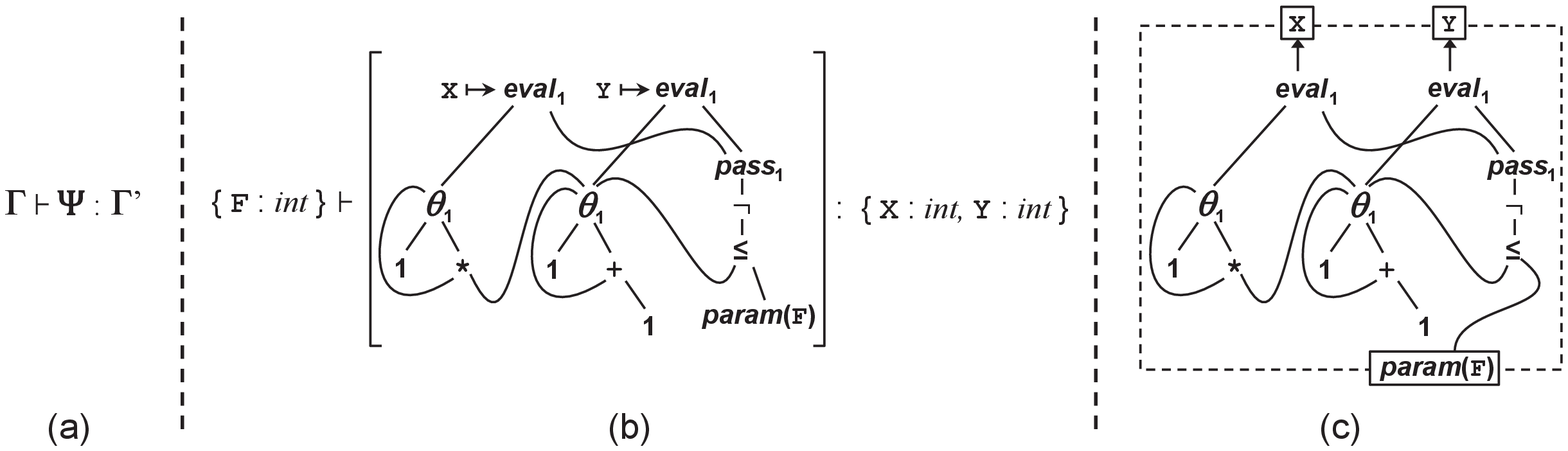}
\caption{Visual representation of the judgment for CFG-like \PEG
  contexts. (a) shows the syntactic form of the judgment; (b) shows an
  example of the syntactic form; and (c) shows the same example in
  visual form.}
\label{notation-peg}
\end{figure}

Although $\Gamma \vdash \context : \Gamma'$ is the syntactic form which we
will use in the body of the text, our diagrams will use the visual
representation of $\Gamma \vdash \context : \Gamma'$ shown in
Figure~\ref{notation-peg}. Part (a) of the figure shows the syntactic
form of the judgment (used in the text); (b) shows an example of the
syntactic form; and finally (c) shows the same example in the visual
form used in our diagrams.

\subsection{Overview}
\label{sec:reversion-overview}

We can draw a parallel between CFG-like PEG contexts and well-typed
statements: both can be seen as taking inputs $\Gamma$ and producing
outputs $\Gamma'$. With this parallel in mind, our basic strategy for
reverting \PEGs to \SIMPLE programs is to recursively translate
CFG-like \PEG contexts $\Gamma \vdash \context : \Gamma'$ to well-typed
\SIMPLE statements $\Gamma \vdash s : \Gamma'$. Therefore, the 
precondition for our reversion algorithm is that the \PEG context we
revert must be CFG-like according to the rules in
Figure~\ref{peg-types}.

For the reversion process, we make two small changes to the type
checking rules for \SIMPLE programs. First, we want to allow the
reversion process to introduce temporary variables without having to
add them to $\Gamma'$ (in $\Gamma \vdash s : \Gamma'$). To this end,
we allow intermediate variables to be dropped from $\Gamma'$, so that
$\Gamma \vdash s : \Gamma'$ and $\Gamma'' \subseteq \Gamma'$ implies
$\Gamma \vdash s : \Gamma''$.

Second, to more clearly highlight which parts of the generated \SIMPLE
code modify which variables, we introduce a notion $\Gamma_0; \Gamma
\vdash s : \Gamma'$ of \SIMPLE statements $s$ which use but do not
modify variables in $\Gamma_0$ (where $\Gamma$ and $\Gamma_0$ are
disjoint). We call $\Gamma_0$ the \emph{immutable context}. The rules
in Figure~\ref{simple-types} can be updated appropriately to disallow
variables from $\Gamma_0$ to be modified. Similarly, we add $\Gamma_0$
to the notion of CFG-like \PEG contexts: $\Gamma_0; \Gamma \vdash
\context : \Gamma'$. Since \PEG contexts cannot modify variables anyway,
the semantics of bindings in $\Gamma_0$ is exactly the same as
bindings in $\Gamma$ (and so we do not need to update the rules from
Figure~\ref{peg-types}). Still, we keep an immutable context
$\Gamma_0$ around for \PEG contexts because $\Gamma_0$ from the \PEG
context gets reflected into the generated \SIMPLE statements where it
has a meaning. Thus, our reversion algorithm will recursively
translate CFG-like \PEG contexts $\Gamma_0;\Gamma \vdash \context :
\Gamma'$ to well-typed \SIMPLE statements $\Gamma_0;\Gamma \vdash s :
\Gamma'$.

\newcommand{\initwtpeg}[1]{\overline{#1}}

Throughout the rest of section~\ref{sec:peg2cfg}, we assume that there
is a \PEG context
$\initwtpeg{\Gamma_0};\initwtpeg{\Gamma}\vdash\initwtpeg{\context}:\initwtpeg{\Gamma'}$
that we are currently reverting. Furthermore, we define $\Gamma_0$ to
be:
\begin{equation}
\Gamma_0 = \initwtpeg{\Gamma_0} \cup \initwtpeg{\Gamma}
\label{gamma0}
\end{equation}
As will become clear in Section~\ref{sec:reversion-loops}, the
$\Gamma_0$ from Equation~\eqref{gamma0} will primarily be used as the
immutable context in recursive calls to the reversion algorithm. The
above definition of $\Gamma_0$ states that when making a recursive
invocation to the reversion process, the immutable context for the
recursive invocation is the entire context from the current
invocation.

\begin{figure}
\includegraphics[width=\textwidth]{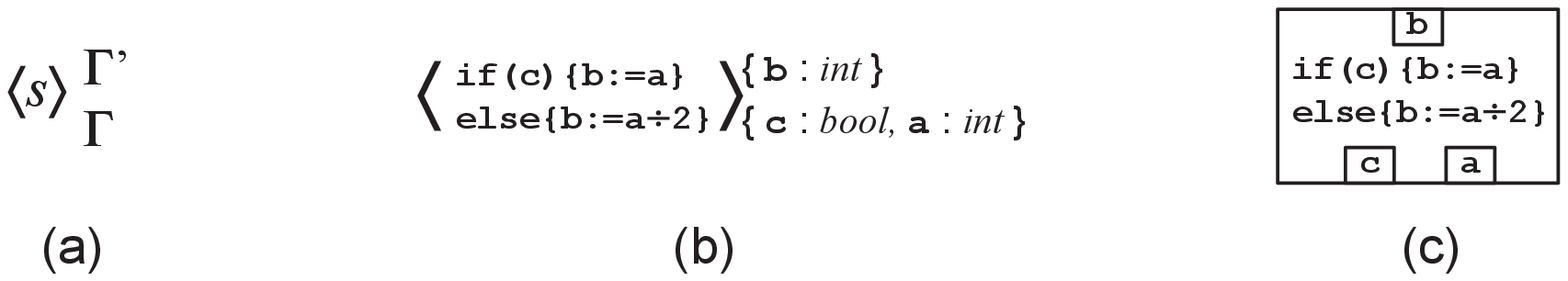}
\caption{Diagrammatic representation of a statement node}
\label{notation-statement}
\end{figure}

\mypara{Statement nodes} The major challenge in reverting a CFG-like
\PEG context lies in handling the primitive operators for encoding
control flow: $\phi$ for branches and $\eval$, $\pass$ and $\theta$
for loops. To handle such primitive nodes, our general approach is to
repeatedly replace a subset of the \PEG nodes with a new kind of \PEG
node called a \emph{statement node}. A statement node is a \PEG node
$\langle s \rangle_\Gamma^{\Gamma'}$ where $s$ is a \SIMPLE statement
satisfying $\Gamma_0; \Gamma \vdash s : \Gamma'$ (recall that $\Gamma_0$
comes from Equation~\eqref{gamma0}). The node has many inputs, one for
each variable in the domain $\Gamma$, and unlike any other node we've
seen so far, it also has many outputs, one for each variable of
$\Gamma'$. A statement node can be perceived as a primitive operator
which, given an appropriately typed list of input values, executes the
statement $s$ with those values in order to produce an appropriately
typed list of output values. Although $\langle s
\rangle_\Gamma^{\Gamma'}$ is the syntactic form which we will use in
the body of the text, our diagrams will use the visual representation
of $\langle s \rangle_\Gamma^{\Gamma'}$ shown in
Figure~\ref{notation-statement}. Part (a) of the figure shows the
syntactic form of a statement node (used in the text); (b) shows an
example of the syntactic form; and finally (c) shows the same example
in the visual form used in our diagrams. Note that the visual
representation in part (c) uses the same visual convention that we
have used throughout for all \PEG nodes: the inputs flow into the
bottom side of the node, and the outputs flow out from the top side of
the node.

The general approach we take is that $\eval$ nodes will be replaced
with while-loop statement nodes (which are statement nodes $\langle s
\rangle_\Gamma^{\Gamma'}$ where $s$ is a while statement) and $\phi$
nodes will be replaced with if-then-else statement nodes (which are
statement nodes $\langle s \rangle_\Gamma^{\Gamma'}$ where $s$ is an
if-then-else statement). To this end, our most simplistic reversion
algorithm, which we present first, converts \PEG contexts to
statements in three steps:
\begin{enumerate}
\item We replace all $\eval$, $\pass$, and $\theta$ nodes with
  while-loop statement nodes. This results in an acyclic \PEG.
\item We replace all $\phi$ nodes with if-then-else statement nodes.
  This results in a PEG with only statement nodes and domain operators
  such as $+$ and $*$ (that is to say, there are no more $\eval$,
  $\pass$, $\theta$ or $\phi$ nodes).
\item We sequence the statement nodes and domain operators into
  successive assignment statements. For a statement node $\langle s
  \rangle_\Gamma^{\Gamma'}$, we simply inline the statement $s$ into
  the generated code.
\end{enumerate}
We present the above three steps in
Sections~\ref{sec:reversion-loops}, \ref{sec:reversion-branches}
and~\ref{sec:reversion-seq}, respectively. Finally, since the process
described above is simplistic and results in large amounts of code
duplication, we present in sections~\ref{sec:loop-fusion}
through~\ref{sec:loop-invariant-code-motion} several optimizations
that improve the quality of the generated \SIMPLE code.

\subsection{Translating Loops}
\label{sec:reversion-loops}

In our first pass, we repeatedly convert each loop-invariant $\eval$
node, along with the appropriate $\pass$ and $\theta$ nodes, into a
while-loop statement node.  The nested loop-variant $\eval$ nodes will
be taken care of when we recursively revert the ``body'' of the
loop-invariant $\eval$ nodes to statements. For each loop-invariant
$\eval_\ell$ node, we apply the process described below.

First, we identify the set of $\theta_\ell$ nodes reachable from the
current $\eval_\ell$ node or its $\pass_\ell$ node without passing
through other loop-invariant nodes (in particular, without passing
through other $\eval_\ell$ nodes). Let us call this set of nodes
$S$. As an illustrative example, consider the left-most \PEG context
in Figure~\ref{reversion-eval-example}, which computes the factorial
of $10$. When processing the single $\eval_1$ node in this \PEG, the
set $S$ will contain both $\theta_1$ nodes. The intuition is that each
$\theta$ node in $S$ will be a loop variable in the \SIMPLE code we
generate. Thus, our next step is to assign a fresh variable $x$
for each $\theta_\ell$ node in $S$; let $b_x$ refer to the first child
of the $\theta_\ell$ node (i.e. the base case), and $i_x$ refer to the
second child (i.e. the iterative case); also, given a node $n \in S$,
we use $\var(n)$ for the fresh variable we just created for $n$. In
the example from Figure~\ref{reversion-eval-example}, we created two
fresh variables $\texttt{x}$ and $\texttt{y}$ for the two $\theta_1$
nodes in $S$. After assigning fresh variables to all nodes in $S$, we
then create a type context $\Gamma$ as follows: for each $n \in S$,
$\Gamma$ maps $\var(n)$ to the type of node $n$ in the \PEG, as given
by the type-checking rules in Figure~\ref{peg-types}. For example, in
Figure~\ref{reversion-eval-example} the resulting type context
$\Gamma$ is $\{\texttt{x}: \texttt{int}, \texttt{y}: \texttt{int}\}$.

\begin{figure}
\includegraphics[width=\textwidth]{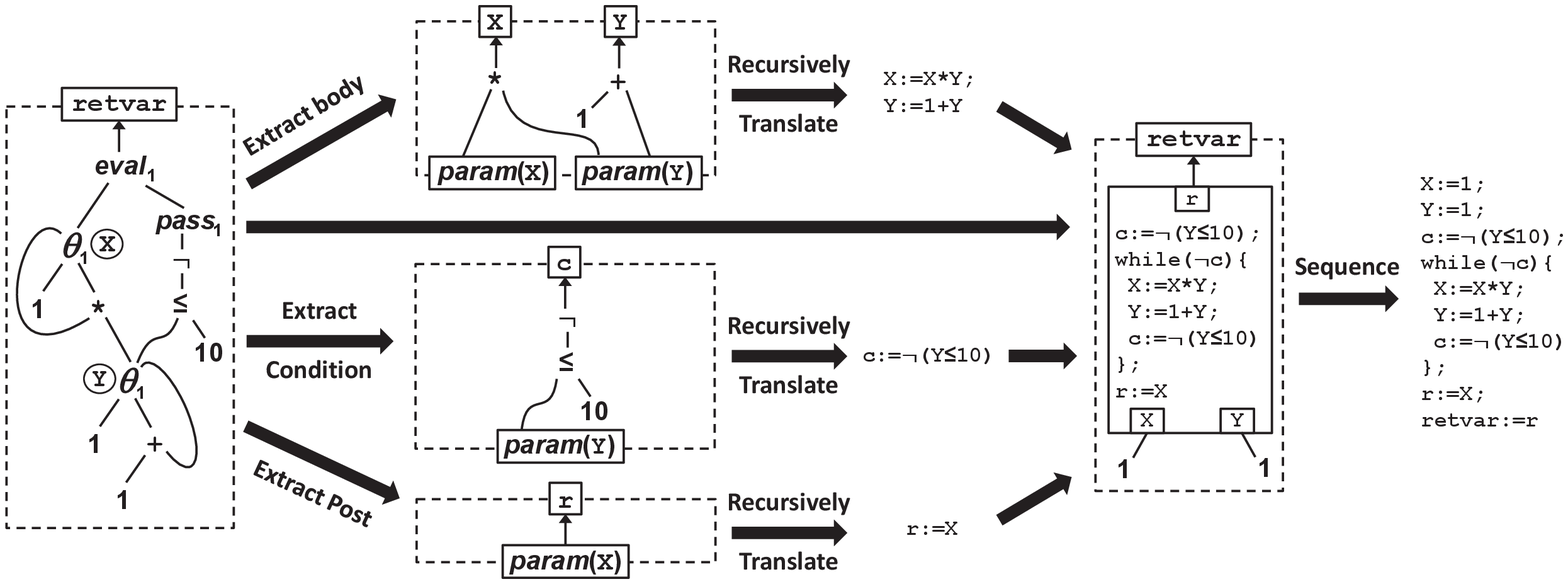}
\caption{Example of converting $\eval$ nodes to while-loop statement nodes}
\label{reversion-eval-example}
\end{figure}

Second, we construct a new \PEG context $\context_i$ that represents the
body of the loop. This \PEG context will state in \PEG terms how the
loop variables are changed in one iteration of the loop. For each
variable $x$ in the domain of $\Gamma$, we add an entry to $\context_i$
mapping $x$ to a copy of $i_x$ (recall that $i_x$ is the second child
of the $\theta$ node which was assigned variable $x$). The copy of
$i_x$ is a fully recursive copy, in which descendants have also been
copied, but with one important modification: while performing the
copy, when we reach a node $n \in S$, we don't copy $n$; instead we
use a parameter node referring to $\var(n)$. This has the effect of
creating a copy of $i_x$ with any occurrence of $n \in S$ replaced by
a parameter node referring to $\var(n)$, which in turn has the effect
of expressing the next value of variable $x$ in terms of the current
values of all loop variables. From the way it is constructed, $\context_i$
will satisfy $\Gamma_0;\Gamma \vdash \context_i : \Gamma$, essentially
specifying, in terms of \PEGs, how the loop variables in $\Gamma$ are
changed as the loop iterates (recall that $\Gamma_0$ comes from
Equation~\eqref{gamma0}). Next, we recursively revert $\context_i$
satisfying $\Gamma_0; \Gamma \vdash \context_i : \Gamma$ to a \SIMPLE
statement $s_i$ satisfying $\Gamma_0; \Gamma \vdash s_i : \Gamma$.
The top-center \PEG context in Figure~\ref{reversion-eval-example}
shows $\context_i$ for our running example. In this case $\context_i$ states
that the body of the loop modifies the loop variables as follows: the
new value of $\texttt{x}$ is $\texttt{x*y}$, and the new value of
$\texttt{y}$ is $\texttt{1+y}$. Figure~\ref{reversion-eval-example}
also shows the \SIMPLE statement resulting from the recursive
invocation of the reversion process.

Third, we take the second child of the $\eval$ node that we are
processing. From the way the \PEG type rules are setup in
Figure~\ref{peg-types}, this second child must be the $\pass$ node of
the $\eval$. Next, we take the first child of this $\pass$ node, and
make a copy of this first child with any occurrence of $n \in S$
replaced by a parameter node referring to $\var(n)$. Let $c$ be the
\PEG node produced by this operation, and let $\context_c$ be the
singleton \PEG context $\{x_c : c\}$, where $x_c$ is fresh. $\context_c$
represents the computation of the break condition of the loop in terms
of the loop variables. From the way it is constructed, $\context_c$ will
satisfy $\Gamma_0;\Gamma \vdash \context_c :\{x_c : \texttt{bool}\}$. We
then recursively revert $\context_c$ satisfying $\Gamma_0;\Gamma \vdash
\context_c :\{x_c : \texttt{bool}\}$ to a \SIMPLE statement $s_c$
satisfying $\Gamma_0; \Gamma \vdash s_c : \{x_c : \texttt{bool}\}$.
$s_c$ simply assigns the break condition of the loop to the variable
$x_c$.  The middle row of Figure~\ref{reversion-eval-example} shows
the PEG context for the break condition and the corresponding \SIMPLE
statement evaluating the break condition.

Fourth, we take the first child of the $\eval$ node and make a copy of
this first child with any occurrence of $n \in S$ replaced with a
parameter node referring to $\var(n)$. Let $r$ be the \PEG node
produced by this operation, and let $\context_r$ be the singleton \PEG
context $\{x_r : r\}$, where $x_r$ is fresh. $\context_r$ represents the
value desired after the loop in terms of the loop variables. From the
way it is constructed, $\context_r$ will satisfy $\Gamma_0;\Gamma \vdash
\context_r : \{x_r : \tau\}$, where $\tau$ is the type of the first child
of the $\eval$ node in the original \PEG. We then recursively revert
$\context_r$ satisfying $\Gamma_0; \Gamma \vdash \context_r : \{x_r : \tau\}$
to a \SIMPLE statement $s_r$ satisfying $\Gamma_0; \Gamma \vdash s_r :
\{x_r : \tau\}$.  $s_r$ simply assigns the value desired after the
loop into variable $x_r$. The bottom row of
Figure~\ref{reversion-eval-example} shows the PEG context for the
value desired after the loop and the corresponding \SIMPLE statement
evaluating the desired value.  Often, but not always, the first child
of the $\eval_\ell$ node will be a $\theta_\ell$ node, in which case
the statement will simply copy the variable as in this example.

Finally, we replace the $\eval_\ell$ node being processed with the
while-loop statement node $\langle s_c; \simplewhile{\neg x_c}{s_i;
  s_c}; s_r \rangle_\Gamma^{\{x_r : \tau\}}$.
Figure~\ref{reversion-eval-example} shows the while-loop statement
node resulting from translating the $\eval$ node in the original \PEG
context. Note that the while-loop statement node has one input for
each loop variable, namely one input for each variable in the domain
of $\Gamma$. For each such variable $x$, we connect $b_x$ to the
corresponding $x$ input of the while-loop statement node (recall that
$b_x$ is the first child of the $\theta$ node which was assigned
variable $x$). Figure~\ref{reversion-eval-example} shows how in our
running example, this amounts to connecting $1$ to both inputs of the
while-loop statement node. In general, our way of connecting the
inputs of the while-loop statement node makes sure that each loop
variable $x$ is initialized with the its corresponding base value
$b_x$. After this input initialization, $s_c$ assigns the status of
the break condition to $x_c$.  While the break condition fails, the
statement $s_i$ updates the values of the loop variables, then $s_c$
assigns the new status of the break condition to $x_c$.  Once the
break condition passes, $s_r$ computes the desired value in terms of
the final values of the loop variables and assigns it to $x_r$.  Note
that it would be more ``faithful'' to place $s_r$ inside the loop,
doing the final calculations in each iteration, but we place it after
the loop as an optimization since $s_r$ does not affect the loop
variables. The step labeled ``Sequence'' in
Figure~\ref{reversion-eval-example} shows the result of sequencing the
\PEG context that contains the while-loop statement node. This
sequencing process will be covered in detail in
Section~\ref{sec:reversion-seq}.

Note that in the computation of the break condition in
Figure~\ref{reversion-eval-example}, there is a double negation, in that we
have \texttt{c := $\neg \ldots$ ;} and \texttt{while($\neg$c)}. Our
more advanced reversion algorithm, described in the accompanying
technical report~\cite{peg2cfg}, takes advantage of more advanced control
structures present in CFGs but not in \SIMPLE, and does not introduce these
double negations.

\subsection{Translating Branches}
\label{sec:reversion-branches}

\begin{figure}
\includegraphics[width=\textwidth]{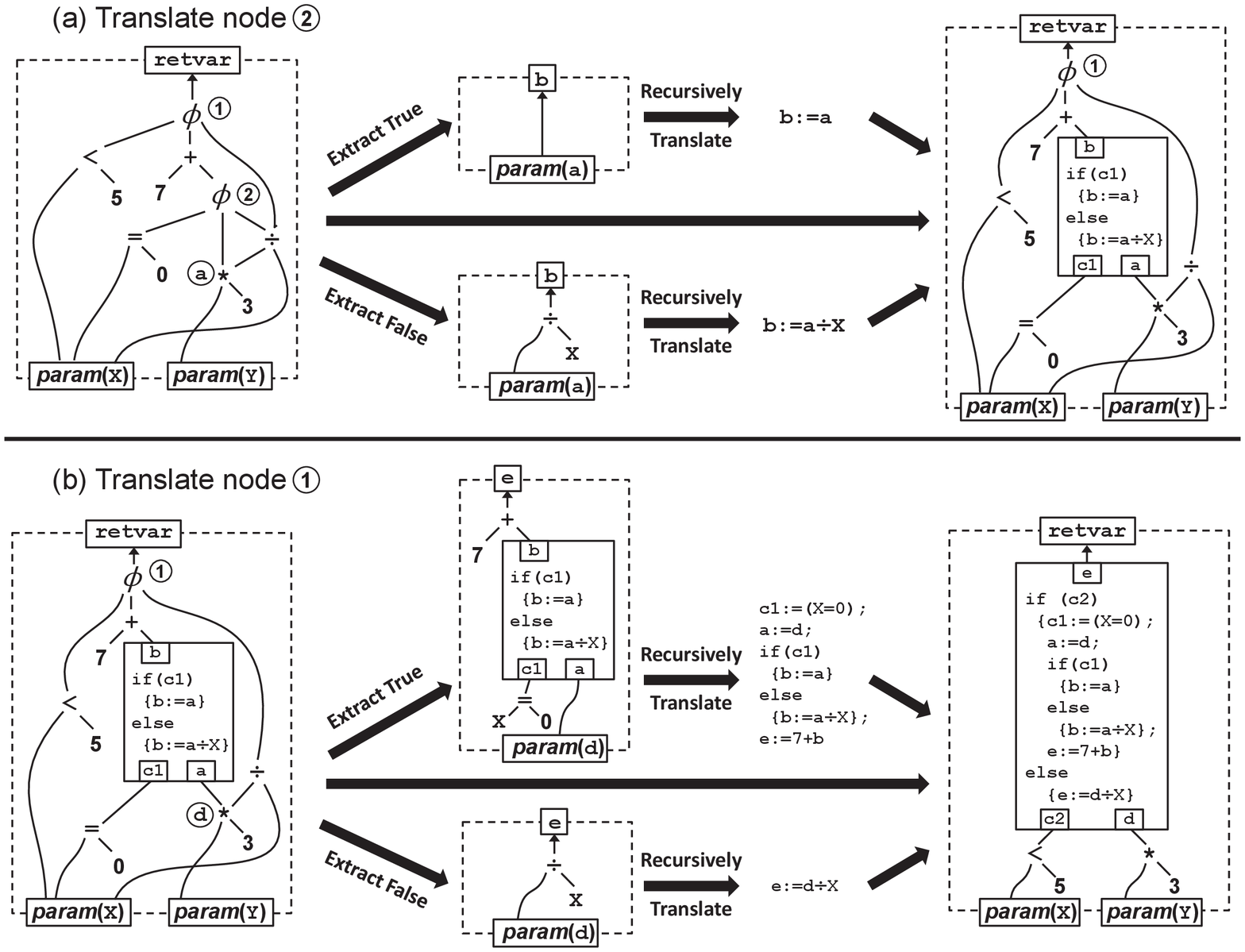}
\caption{Example of converting $\phi$ nodes to if-then-else statement nodes}
\label{reversion-phi-example}
\end{figure}

After our first pass, we have replaced $\eval$, $\pass$ and $\theta$
nodes with while-loop statement nodes. Thus, we are left with an
acyclic $\PEG$ context that contains $\phi$ nodes, while-loop
statement nodes, and domain operators like $+$ and $*$. In our second
pass, we repeatedly translate each $\phi$ node into an if-then-else
statement node.  In order to convert a $\phi$ node to an if-then-else
statement node, we must first determine the set of nodes which will
always need to be evaluated regardless of whether the guard condition
is true or false.  This set can be hard to determine when there is
another $\phi$ node nested inside the $\phi$ node. To see why this
would be the case, consider the example in
Figure~\ref{reversion-phi-example}, which we will use as a running
example to demonstrate branch translation. Looking at part (a), one
might think at first glance that the \texttt{$\div$} node in the first
diagram is always evaluated by the $\phi$ node labeled \textcircled{1}
since it is used in the \PEGs for both the second and third children.
However, upon further examination one realizes that actually the
\texttt{$\div$} node is evaluated in the second child only when
\texttt{x$\neq$0} due to the $\phi$ node labeled \textcircled{2}.  To
avoid these complications, it is simplest if we first convert $\phi$
nodes that do not have any other $\phi$ nodes as descendants. After
this replacement, there will be more $\phi$ nodes that do not have
$\phi$ descendants, so we can repeat this until no $\phi$ nodes are
remaining. In the example from Figure~\ref{reversion-phi-example}, we
would convert node \textcircled{2} first, resulting in (b). After this
conversion, node \textcircled{1} no longer has any $\phi$ descendants
and so it can be converted next. Thus, we replace $\phi$ nodes in a
bottom-up order. For each $\phi$ node, we use the following process.

First, we determine the set $S$ of nodes that are descendants of both
the second and third child of the current $\phi$ node (i.e. the true
and false branches). These are the nodes that will get evaluated
regardless of which way the $\phi$ goes. We assign a fresh variable to
each node in this set, and as in the case of loops we use $\var(n)$ to
denote the variable we've assigned to $n$. In
Figure~\ref{reversion-phi-example}(a), the \texttt{*} node is a
descendant of both the second and third child of node \textcircled{2},
so we assign it the fresh variable \texttt{a}.  Note that the
\texttt{3} node should also be assigned a variable, but we do not show
this in the figure since the variable is never used.  Next, we take
the second child of the $\phi$ node and make a copy of this second
child in which any occurrence of $n \in S$ has been replaced with a
parameter node referring to $\var(n)$. Let $t$ be the \PEG node
produced by this operation. Then we do the same for the third child of
the $\phi$ node to produce another \PEG node $f$. The $t$ and $f$
nodes represent the true and false computations in terms of the \PEG
nodes that get evaluated regardless of the direction the $\phi$ goes.
In the example from Figure~\ref{reversion-phi-example}(a), $t$ is
$\overline{param}(\texttt{a})$ and $f$ is
$\overline{\div}(\overline{param}(\texttt{a}),
\overline{param}(\texttt{x}))$. Examining $t$ and $f$, we produce a
context $\Gamma$ of the newly created fresh variables used by either
$t$ or $f$.  In the example, $\Gamma$ would be simply $\{\texttt{a}:
\texttt{int}\}$.  The domain of $\Gamma$ does not contain \texttt{x}
since \texttt{x} is not a new variable (i.e. \texttt{x} is in the
domain of $\Gamma_0$, where $\Gamma_0$ comes from
Equation~\eqref{gamma0}). Thus, $t$ and $f$ are PEGs representing the
true and false cases in terms of variables $\Gamma$ representing
values that would be calculated regardless.

Second, we invoke the reversion process recursively to translate $t$
and $f$ to statements. In particular, we create two singleton contexts
$\context_t = \{x_\phi : t\}$ and $\context_f = \{x_\phi : f\}$ where $x_\phi$
is a fresh variable, making sure to use the same fresh variable in the
two contexts. From the way it is constructed, $\context_t$ satisfies
$\Gamma_0; \Gamma \vdash \context_t : \{x_\phi : \tau\}$, where $\tau$ is
the type of the $\phi$ node. Thus, we recursively revert $\context_t$
satisfying $\Gamma_0; \Gamma \vdash \context_t : \{x_\phi : \tau\}$ to a
\SIMPLE statement $s_t$ satisfying $\Gamma_0; \Gamma \vdash s_t :
\{x_\phi : \tau\}$. Similarly, we revert $\context_f$ satisfying
$\Gamma_0; \Gamma \vdash \context_f : \{x_\phi : \tau\}$ to a statement
$s_f$ satisfying $\Gamma_0; \Gamma \vdash s_f : \{x_\phi : \tau\}$.
The steps labeled ``Extract True'' and ``Extract False'' in
Figure~\ref{reversion-phi-example} show the process of producing
$\context_t$ and $\context_f$ in our running example, where the fresh variable
$x_\phi$ is \texttt{b} in part (a) and \texttt{e} in part (b). Note
that $\context_t$ and $\context_f$ may themselves contain statement nodes, as
in Figure~\ref{reversion-phi-example}(b), but this poses no problems
for our recursive algorithm. Finally, there is an important notational
convention to note in Figure~\ref{reversion-phi-example}(a). Recall
that $\context_f$ satisfies $\Gamma_0; \Gamma \vdash \context_f : \{x_\phi :
\tau\}$, and that in Figure~\ref{reversion-phi-example}(a) $\Gamma_0 =
\{ \texttt{x}:\texttt{int}\}$ and $\Gamma =
\{\texttt{a}:\texttt{int}\}$. In the graphical representation
of $\Gamma_0; \Gamma \vdash \context_f : \{x_\phi : \tau\}$ in
Figure~\ref{reversion-phi-example}(a), we display variable \texttt{a}
as a real boxed input (since it is part of $\Gamma$), whereas because
\texttt{x} is in $\Gamma_0$, we display \texttt{x} without a box, and
using the shorthand of omitting the $\param$ (even though in reality
it is there).

Finally, we replace the $\phi$ node we are processing with the
if-then-else statement node $\langle \simpleite{x_c}{s_t}{s_f}
\rangle_{(\Gamma, x_c : \texttt{bool})}^{\{x_\phi : \tau\}}$ (where
$x_c$ is fresh). This statement node has one input for each entry in
$\Gamma$, and it has one additional input $x_c$ for the guard
value. We connect the $x_c$ input to the first child of the $\phi$
node we are currently processing. For each variable $x$ in the domain
of $\Gamma$, we connect the $x$ input of the statement node to the
``always-evaluated'' node $n$ for which $\var(n) =
x$. Figure~\ref{reversion-phi-example} shows the newly created
if-then-else statement nodes and how they are connected when
processing $\phi$ node \textcircled{1} and \textcircled{2}. In
general, our way of connecting the inputs of the if-then-else
statement node makes sure that each always-evaluated node is assigned
to the appropriate variable, and the guard condition is assigned to
$x_c$. After this initialization, the statement checks $x_c$, the
guard condition, to determine whether to take the true or false
branch.  In either case, the chosen statement computes the desired
value of the branch and assigns it to $x_\phi$.

\subsection{Sequencing Operations and Statements}
\label{sec:reversion-seq}

When we reach our final pass, we have already eliminated all primitive
operators ($\eval$, $\pass$, $\theta$, and $\phi$) and replaced them
with statement nodes, resulting in an acyclic \PEG context containing
only statement nodes and domain operators like $+$ and $*$. At this
point, we need to sequence these statement nodes and operator
nodes. We start off by initializing a statement variable $S$ as the
empty statement.  We will process each \PEG node one by one,
postpending lines to $S$ and then replacing the \PEG node with a
parameter node. It is simplest to process the \PEG nodes from the
bottom up, processing a node once all of its inputs are parameter
nodes. Figure~\ref{reversion-seq-example} shows every stage of
converting a \PEG context to a statement. At each step, the current
statement $S$ is shown below the \PEG context.

\begin{figure}
\includegraphics[width=\textwidth]{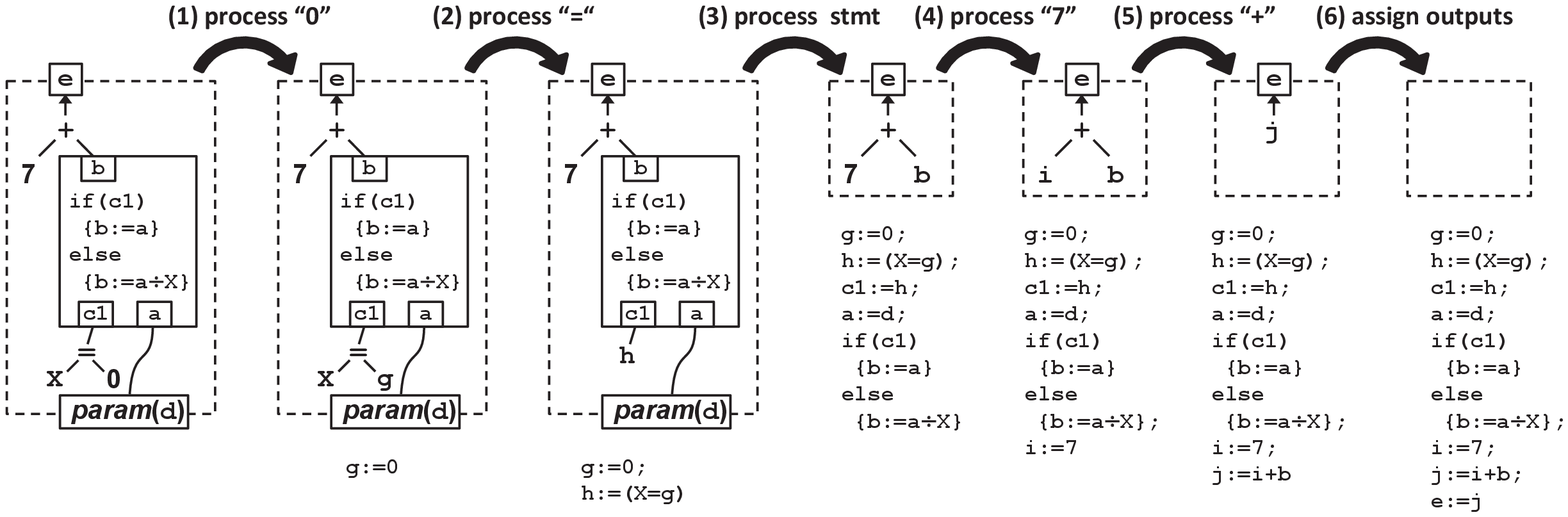}
\caption{Example of sequencing a PEG context (with statement nodes) into a statement}
\label{reversion-seq-example}
\end{figure}

If the node being processed is a domain operator, we do the following.
Because we only process nodes where all of its inputs are parameter
nodes, the domain operator node we are processing must be of the
following form: $\overline{op}(\overline{param}(x_1), \dots,
\overline{param}(x_k))$.  We first designate a fresh variable $x$.
Then, we postpend the line $x := op(x_1, \dots, x_k)$ to $S$.
Finally, we replace the current node with the node
$\overline{param}(x)$. This process is applied in the first, second,
fourth, and fifth steps of Figure~\ref{reversion-seq-example}.  Note
that in the first and fourth steps, the constants $0$ and $7$ are a
null-ary domain operators.

If on the other hand the node being processed is a statement node, we
do the following. This node must be of the following form: $\langle s
\rangle_\Gamma^{\Gamma'}$. For each input variable $x$ in the domain
of $\Gamma$, we find the $\overline{param}(x_0)$ that is connected to
input $x$, and postpend the line $x := x_0$ to $S$. In this way, we
are initializing all the variables in the domain of $\Gamma$. Next, we
postpend $s$ to $S$. Finally, for each output variable $x'$ in the
domain of $\Gamma'$, we replace any links in the $\PEG$ to the $x'$
output of the statement node with a link to $\overline{param}(x')$.
This process is applied in the third step of
Figure~\ref{reversion-seq-example}.

Finally, after processing all domain operators and statement nodes, we
will have each variable $x$ in the domain of the \PEG context being
mapped to a parameter node $\overline{param}(x')$.  So, for each such
variable, we postpend the line $x' := x$ to $S$. All these assignments
should intuitively run in parallel. This causes problems if
there is a naming conflict, for example $x$ gets $y$ and $y$ gets
$x$. In such cases, we simply introduce intermediate fresh copies of
all the variables being read, and then we perform all the assignments
by reading from the fresh copies. In the case of $x$ and $y$, we would
create copies $x'$ and $y'$ of $x$ and $y$, and then assign $x'$ to
$y$, and $y'$ to $x$. This process is applied in the sixth and last
step of Figure~\ref{reversion-seq-example} (without any naming
conflicts). The value of $S$ is the final result of the reversion,
although in practice we apply copy propagation to this statement since
the sequencing process produces a lot of intermediate variables.

\subsection{Loop Fusion}
\label{sec:loop-fusion}

\begin{figure}
\includegraphics[width=\textwidth]{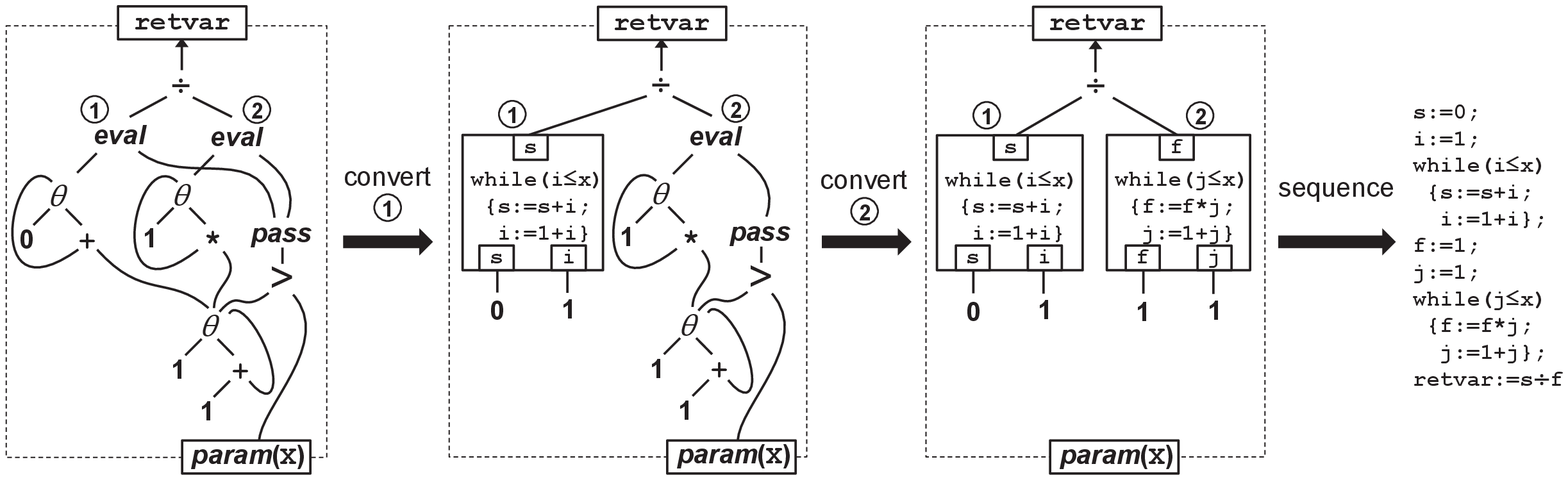}
\caption{Reversion of a PEG without applying loop fusion}
\label{reversion-loop-fusion-bad}
\end{figure}

Although the process described above will successfully revert a PEG to
a \SIMPLE program, it does so by duplicating a lot of code.  Consider
the reversion process in Figure~\ref{reversion-loop-fusion-bad}.  The
original \SIMPLE program for this \PEG was as follows:

\vspace{6pt}
\begin{tabular}{l}
        \qquad\qquad\verb-s:=0; f:=1; i:=1;-\\
        \qquad\qquad\verb-while(i<=x) { s:=s+i; f:=f*i; i:=1+i };-\\
        \qquad\qquad\texttt{retvar:=s$\div$f}
\end{tabular}
\vspace{6pt}

The conversion of the above code to \PEG results in two $\eval$ nodes,
one for each variable that is used after the loop.  The reversion process
described so far converts each $\eval$ node separately, resulting in
two separate loops in the final \SIMPLE program.  Here we present a
simple optimization that prevents this code duplication by fusing
loops during reversion.  In fact, this added loop-fusion step can even
fuse loops that were distinct in the original program. Thus, loop
fusion can be performed simply by converting to a \PEG and immediately
reverting back to a \SIMPLE program, without even having to do any
intermediate transformations on the \PEG.

\begin{figure}
\includegraphics[width=\textwidth]{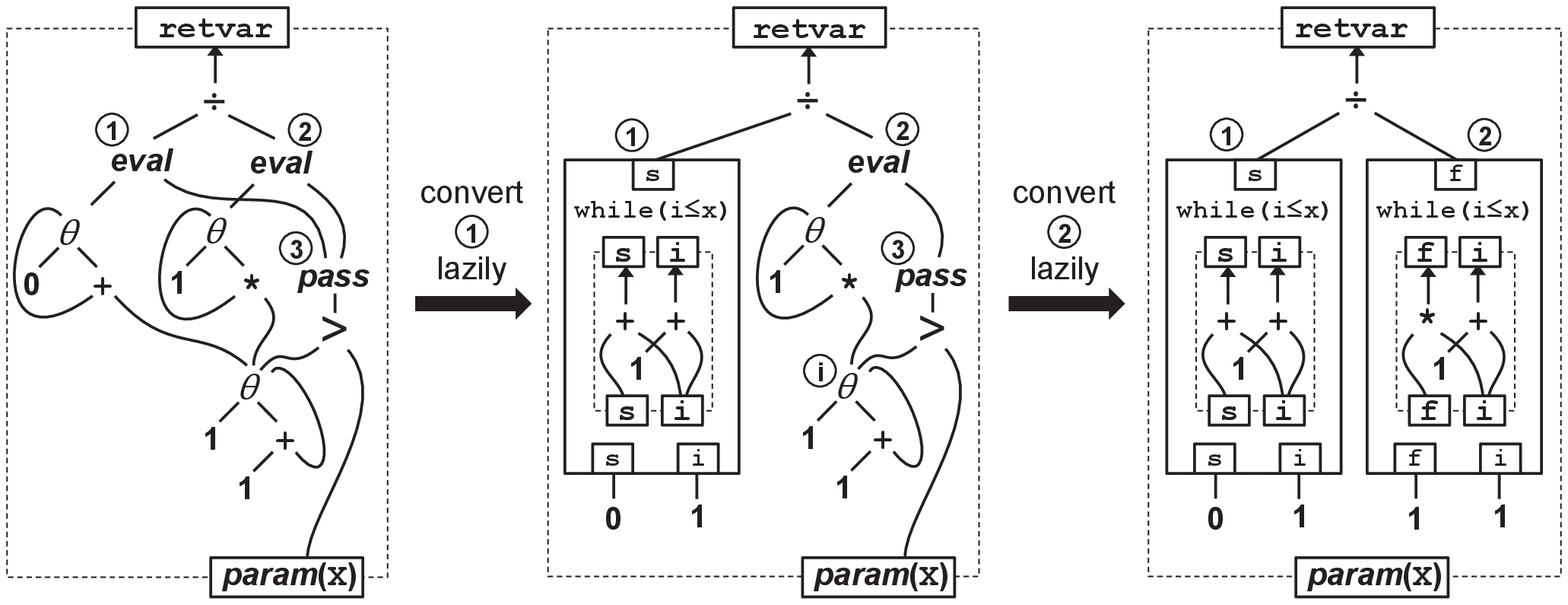}
\caption{Conversion of $\eval$ nodes to revised loop nodes}
\label{reversion-loop-fusion-convert}
\end{figure}

We update our reversion process to perform loop-fusion by making
several changes. First, we modify the process for converting $\eval$
nodes to while-loop statement nodes in three ways; the revised
conversion process is shown in
Figure~\ref{reversion-loop-fusion-convert} using the same PEG as
before.  The first modification is that we tag each converted $\theta$
node with the fresh variable we designate for it.  For example, the
conversion process for the first $\eval$ node in
Figure~\ref{reversion-loop-fusion-convert} generates a fresh variable
\texttt{i} for one of the $\theta$ nodes, and so we tag this $\theta$
node with \texttt{i}.  If we ever convert a $\theta$ node that has
already been tagged from an earlier $\eval$ conversion, we reuse that
variable.  For example, when converting the second $\eval$ node in
Figure~\ref{reversion-loop-fusion-convert}, we reuse the variable
\texttt{i} unlike in Figure~\ref{reversion-loop-fusion-bad} where we
introduced a fresh variable \texttt{j}.  This way all the $\eval$
nodes are using the same naming convention.
The second modification is that when processing an $\eval$ node, we do
not immediately revert the \PEG context for the loop body into a
statement, but rather we remember it for later.  This is why the
bodies of the while-loop statement nodes in
Figure~\ref{reversion-loop-fusion-convert} are still \PEG contexts
rather than statements.  Thus, we have to introduce a new kind of
node, which we call a \emph{loop node}, which is like a while-loop
statement node, except that it stores the body (and only the body) of
the loop as a PEG context rather than a statement -- the remaining
parts of the loop are still converted to statements (in particular,
the condition and the post-loop computation are still converted to
statements, as was previously shown in
Figure~\ref{reversion-eval-example}). As an example, nodes
\textcircled{1} and \textcircled{2} in the right most part of
Figure~\ref{reversion-loop-fusion-convert} are loop
nodes. Furthermore, because we are leaving \PEGs inside the loop nodes
to be converted for later, we use the term ``convert lazily'' in
Figure~\ref{reversion-loop-fusion-convert}.
The third modification is that the newly introduced loop nodes store
an additional piece of information when compared to while-loop
statement nodes. In particular, when we replace an $\eval$ node with a
loop node, the new loop node will store a link back to the $\pass$
node of the $\eval$ node being replaced. We store these additional
links so that we can later identify fusable loops: we will consider
two loop nodes fusable only if they share the same $\pass$ node. We do
not show these additional links explicitly in
Figure~\ref{reversion-loop-fusion-convert}, but all the loop nodes in
that Figure implicitly store a link back to the same $\pass$ node,
namely node \textcircled{3}.

Second, after converting the $\phi$ nodes but before sequencing, we
search for loop nodes which can be fused.  Two loop nodes can be fused
if they share the same $\pass$ node and neither one is a descendant of
the other. For example, the two loop nodes in
Figure~\ref{reversion-loop-fusion-convert} can be fused.  If one loop
node is a descendant of the other, then the result of finishing the
descendant loop is required as input to the other loop, and so they
cannot be executed simultaneously.  To fuse the loops, we simply union
their body \PEG contexts, as well as their inputs and their outputs.
The step labeled ``fuse \textcircled{1} \& \textcircled{2}'' in
Figure~\ref{reversion-loop-fusion-good} demonstrates this process on
the result of Figure~\ref{reversion-loop-fusion-convert}.  This
technique produces correct results because we used the same naming
convention across $\eval$ nodes and we used fresh variables for all
$\theta$ nodes, so no two distinct $\theta$ nodes are assigned the
same variable. We repeat this process until there are no fusable loop
nodes.

\begin{figure}
\includegraphics[width=\textwidth]{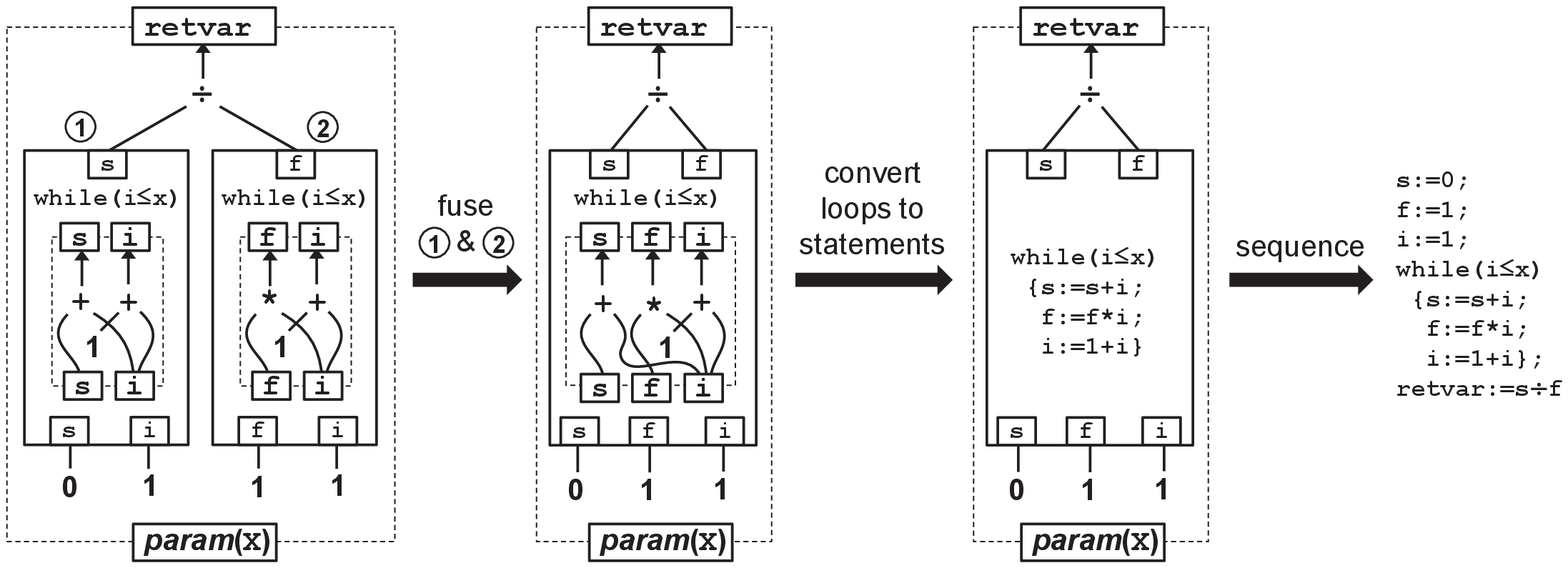}
\caption{Fusion of loop nodes}
\label{reversion-loop-fusion-good}
\end{figure}

Finally, the sequencing process is changed to first convert all loop
nodes to while-loop statement nodes, which involves recursively
translating the body \PEG context inside the loop node to a
statement. This additional step is labeled ``convert loops to
statements'' in Figure~\ref{reversion-loop-fusion-good}. The final
\SIMPLE program has only one while loop which simultaneously
calculates both of the desired results of the loop, as one would
expect.

To summarize, the process described so far is to (1) translate $\eval$
nodes into loop nodes, (2) translate $\phi$ nodes into if-then-else
statements, (3) perform fusion of loop nodes and (4) sequencing step.
It is important to perform the fusion of loop nodes \emph{after}
converting $\phi$ nodes, rather than after converting $\eval$
nodes. Consider for example two loops with the same break condition,
neither of which depend on the other, but where one is always executed
and the other is only executed when some branch guard is true (that is
to say, its result is used only on one side of a $\phi$ node). If we
perform fusion of loop nodes before converting $\phi$ nodes, then
these two loop nodes appear to be fusable, but fusing them would cause
both loops to always be evaluated, which is not semantics
preserving. We avoid this problem by processing all $\phi$ nodes
first, after which point we know that all the remaining nodes in the
\PEG context must be executed (although some of these nodes may be
branch nodes). In the example just mentioned with two loops, the loop
which is under the guard (and thus under a $\phi$ node) will be
extracted and recursively processed when the $\phi$ node is
transformed into a branch node. In this recursive reversion, only one
loop will be present, the one under the guard, and so no loop fusion
is performed.  After the $\phi$ node is processed, there will be only
one remaining loop node, the one which is executed
unconditionally. Again, since there is only one loop node, no loop
fusion is performed, and so the semantics is preserved.

\subsection{Branch Fusion}
\label{branch-fusion}

\begin{figure}
\includegraphics[width=\textwidth]{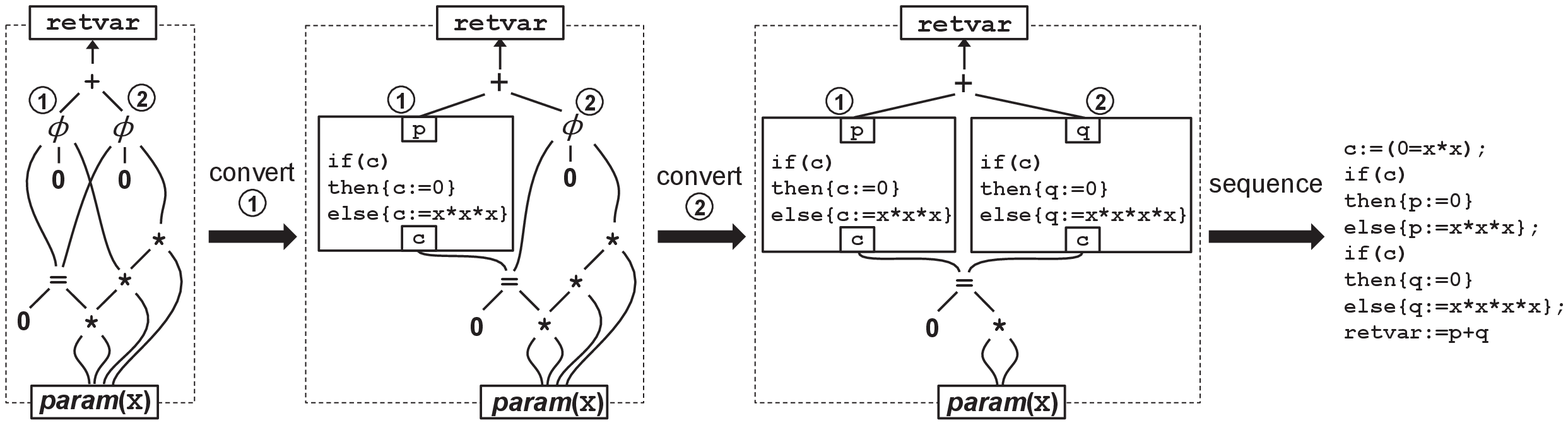}
\caption{Reversion of a PEG without applying branch fusion}
\label{reversion-branch-fusion-bad}
\end{figure}

In the same way that our previously described reversion process
duplicated loops, so does it duplicate branches, as demonstrated in
Figure~\ref{reversion-branch-fusion-bad}. Similarly to loop fusion,
our reversion process can be updated to perform branch fusion.

First, we modify the processing of $\phi$ nodes to make the reversion
of recursive \PEG contexts lazy: rather than immediately processing
the extracted true and false \PEG contexts, as was done in
Figure~\ref{reversion-phi-example}, we instead create a new kind of
node called a \emph{branch node} and store the true and false \PEG
contexts in that node. A branch node is like an if-then-else statement
node, except that instead of having \SIMPLE code for the true and
false sides of the statement, the branch node contains \PEG contexts
to be processed later. As with if-then-else statement nodes, a branch
node has a guard input which is the first child of the $\phi$ node
being replaced (that is to say, the value of the branch
condition). For example, Figure~\ref{reversion-branch-fusion-convert}
shows this lazy conversion of $\phi$ nodes on the same example as
Figure~\ref{reversion-branch-fusion-bad}. The nodes labeled
\textcircled{1} and \textcircled{2} in the right most part of
Figure~\ref{reversion-branch-fusion-convert} are branch nodes.

\begin{figure}
\includegraphics[width=\textwidth]{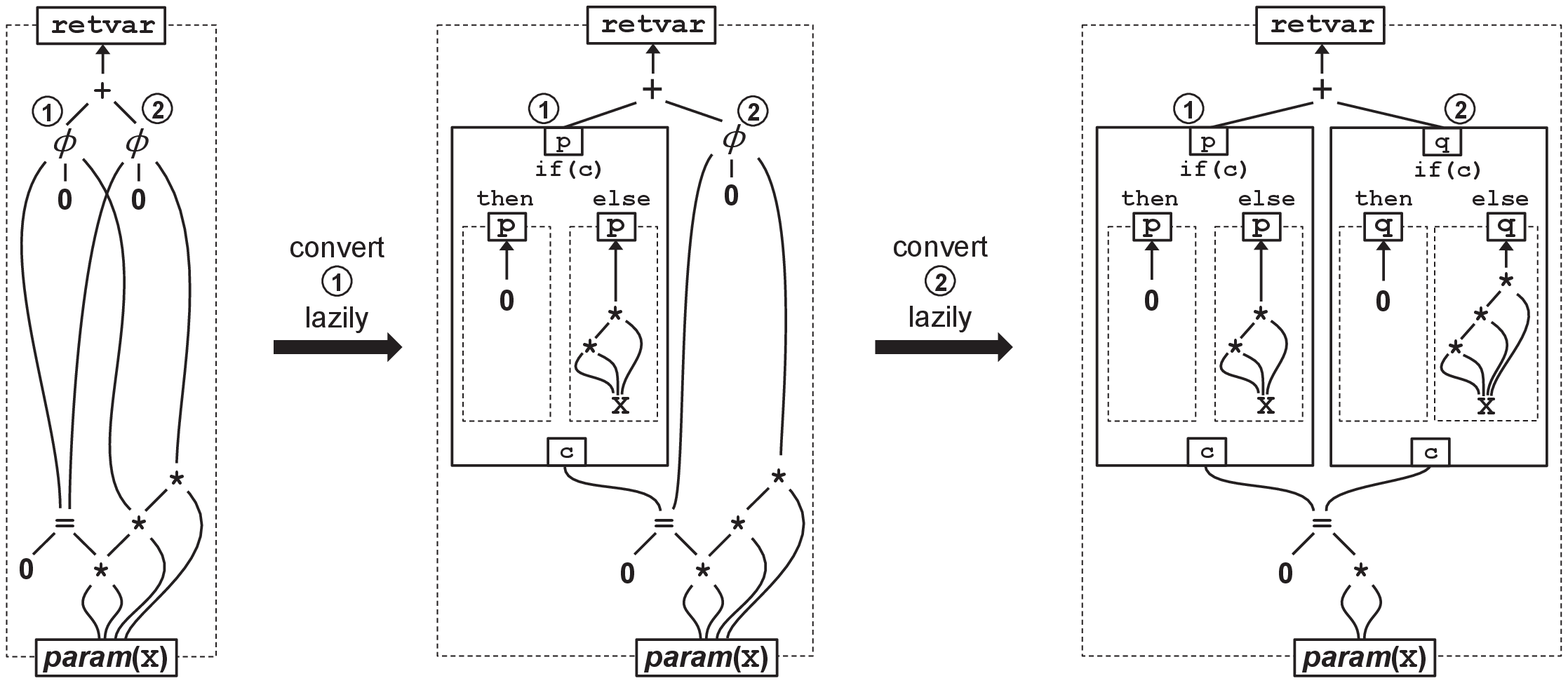}
\caption{Conversion of $\phi$ nodes to revised branch nodes}
\label{reversion-branch-fusion-convert}
\end{figure}

Second, after all $\phi$ nodes have been converted to branch nodes, we
search for branch nodes that can be fused. If two branch nodes share
the same guard condition input, and neither one is a descendant of the
other, then they can be fused.  Their true \PEG contexts, false \PEG
contexts, inputs, and outputs are all unioned together respectively.
This process is much like the one for loop fusion, and is demonstrated
in Figure~\ref{reversion-branch-fusion-good} in the step labeled
``fuse \textcircled{1} \& \textcircled{2}''. Notice that when we union
two \PEGs, if there is a node in each of the two \PEGs representing
the exact same expression, the resulting union will only contain one
copy of this node. This leads to an occurrence of common
sub-expression elimination in
Figure~\ref{reversion-branch-fusion-good}: when the false and
true \PEGs are combined during fusion, the resulting \PEG only has
one \verb-x*x*x-, which allows the computation for \verb-q- in the
final generated code to be \verb-c*x-, rather than \verb-x*x*x*x-.

\begin{figure}
\includegraphics[width=\textwidth]{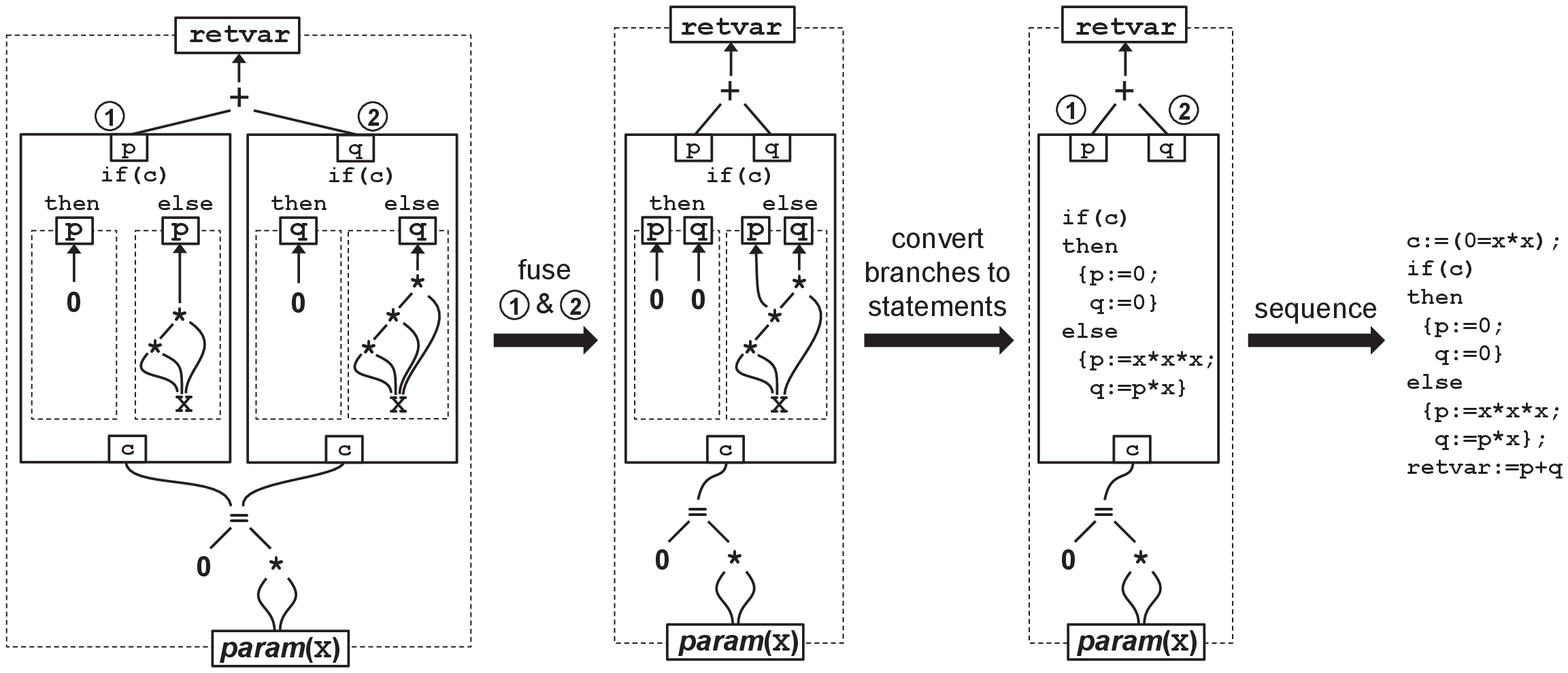}
\caption{Fusion of branch nodes}
\label{reversion-branch-fusion-good}
\end{figure}

Finally, the sequencing process is changed to first convert all
branch nodes into statement nodes, which involves recursively
translating the true and false \PEG contexts inside the branch node to
convert them to statements. This additional step is labeled ``convert
branches to statements'' in
Figure~\ref{reversion-branch-fusion-good}. The final \SIMPLE program
has only one if-then-else which simultaneously calculates both of the
desired results of the branches, as one would expect.

To summarize, the process described so far is to (1) translate $\eval$
nodes into loop nodes, (2) translate $\phi$ nodes into branch nodes,
(3) perform fusion of loop nodes (4) perform fusion of branch nodes
(5 sequencing step. As with loop fusion, it is important to perform fusion
of branch nodes after each and every $\phi$ node has been converted to
branch nodes. Otherwise, one may end up fusing two branch nodes where
one branch node is used under some $\phi$ and the other is used
unconditionally. 

\begin{figure}
\includegraphics[width=\textwidth]{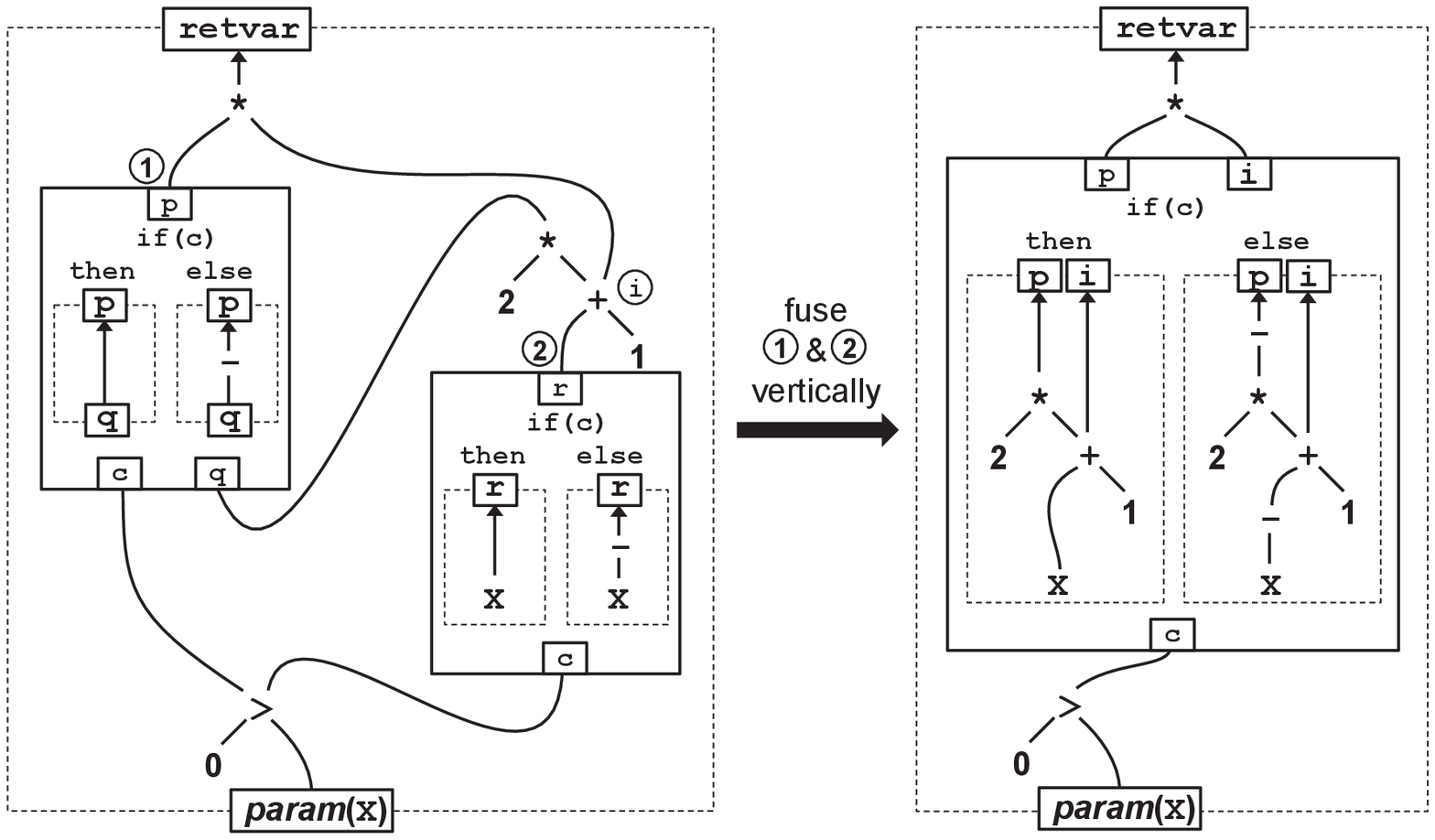}
\caption{Vertical fusion of branch nodes}
\label{reversion-loop-fusion-vertical}
\end{figure}

If two branch nodes share the same guard condition input, but one is a
descendant of the other, we can even elect to fuse them vertically, as
shown in Figure~\ref{reversion-loop-fusion-vertical}. In particular,
we sequence the true \PEG context of one branch node with the true
\PEG context of the other, and do the same with the false \PEG
contexts. Note how, because the node labeled \textcircled{\texttt{i}}
is used elsewhere than just as an input to branch node
\textcircled{1}, we added it as an output of the fused branch node.

\subsection{Hoisting Redundancies from Branches}
\label{sec:hosting-from-branches}

Looking back at the branch fusion example from
Figures~\ref{reversion-branch-fusion-convert}
and~\ref{reversion-branch-fusion-good}, there is still one
inefficiency in the generated code. In particular, \verb-x*x- is
computed in the false side of the branch, even though \verb-x*x- has
already been computed before the branch.

In our original description for converting $\phi$ nodes in
Section~\ref{sec:reversion-branches}, we tried to avoid this kind of
redundant computation by looking at the set of nodes that are
reachable from both the true and false children (second and third
children) of a $\phi$ node. This set was meant to capture the nodes
that, for a given $\phi$, are need to be computed regardless of which
side the $\phi$ node goes -- we say that such nodes execute
unconditionally with respect to the given $\phi$ node. These nodes
were kept outside of the branch node (or the if-then-else statement
node if using such nodes). As an example, the node labeled
\textcircled{\texttt{a}} in Figure~\ref{reversion-phi-example} was
identified as belonging to this set when translating $\phi$ node
\textcircled{2}, and this is why the generated if-then-else statement
node does not contain node \textcircled{\texttt{a}}, instead taking it
as an input (in addition to the \texttt{c1} input which is the branch
condition).

It is important to determine as completely as possible the set of
nodes that execute unconditionally with respect to a
$\phi$. Otherwise, code that intuitively one would think of executing
unconditionally outside of the branch (either before the branch or
after it) would get duplicated in one or both sides of branch. This is
precisely what happened in
Figure~\ref{reversion-branch-fusion-convert}: our approach of
computing the nodes that execute unconditionally (by looking at nodes
reachable from the true and false children) returned the empty set,
even though \verb-x*x- actually executes unconditionally. This is what
lead to \verb-x*x- being duplicated, rather than being kept outside of
the branch (in the way that \texttt{a} was in
Figure~\ref{reversion-phi-example}). A more precise analysis would be
to say that a node executes unconditionally with respect to a $\phi$
node if it is reachable from the true and false children of the $\phi$
(second and third children), \emph{or} from the branch condition
(first child). This would identify \verb-x*x- as being executed
unconditionally in Figure~\ref{reversion-phi-example}. However, even
this more precise analysis has limitations. Suppose for example that
some node is used only on the true side of the $\phi$ node, and never
used by the condition, so that the more precise analysis would
not identify this node as always executing. However, this node could
be used unconditionally higher up in the \PEG, or alternatively it
could be the case that the condition of the $\phi$ node is actually
equivalent to true. In fact, this last possibility points to the fact
that computing exactly what nodes execute unconditionally with respect
to a $\phi$ node is undecidable (since it reduces to deciding if a
branch is taken in a Turing-complete computation). However, even
though the problem is undecidable, more precision leads to less code
duplication in branches.

\mypara{$\MustEval$ Analysis} 
To modularize the part of the system that deals with identifying nodes
that must be evaluated unconditionally, we define a $\MustEval$
analysis. This analysis returns a set of nodes that are known to
evaluate unconditionally in the current \PEG context. An implementation
has a lot of flexibility in how to define the $\MustEval$
analysis. More precision in this analysis leads to less code
duplication in branches.

\begin{figure}
\includegraphics[width=\textwidth]{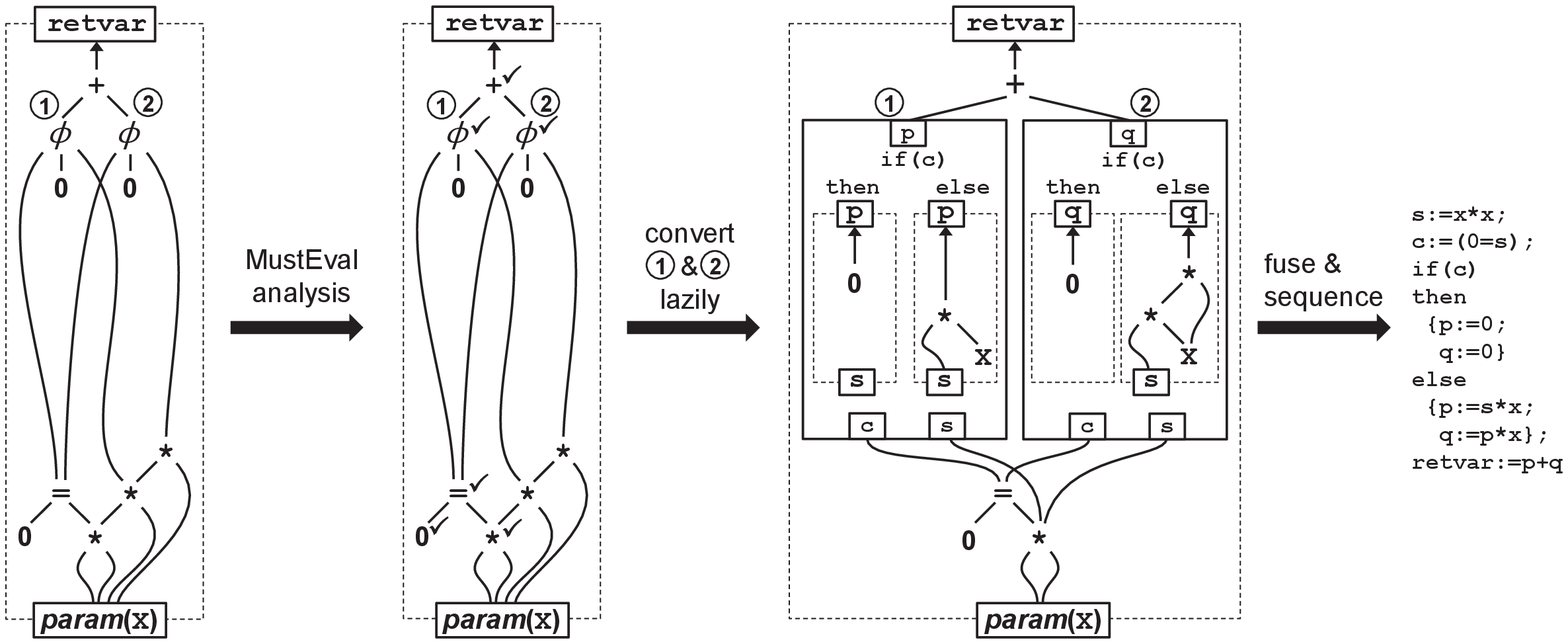}
\caption{Same example as in
  Figure~\ref{reversion-branch-fusion-convert}, but this time
  hoisting redundancies from branches}
\label{reversion-branch-fusion-precise}
\end{figure}

Figure~\ref{reversion-branch-fusion-precise} shows the example from
Figure~\ref{reversion-branch-fusion-convert} again, but this time using
a refined process that uses the $\MustEval$ analysis. The nodes
which our $\MustEval$ analysis identifies as always evaluated have
been marked, including \texttt{x*x}.
After running the $\MustEval$ analysis, we convert all $\phi$ nodes
which are marked as always evaluated. All remaining $\phi$ nodes will
be pulled into the resulting branch nodes and so handled by recursive
calls. Note that this is a change from
Section~\ref{sec:reversion-branches}, where $\phi$ nodes were
processed starting from the lower ones (such as node \textcircled{2}
in Figure~\ref{reversion-phi-example}) to the higher ones (such as
node \textcircled{1} in Figure~\ref{reversion-phi-example}). Our
updated process, when running on the example from
Figure~\ref{reversion-phi-example}, would process node \textcircled{1}
first, and in doing so would place node \textcircled{2} in the true
\PEG context of a branch node. Thus, node \textcircled{2} would get
processed later in a recursive reversion.

Going back to Figure~\ref{reversion-branch-fusion-precise}, both
$\phi$ nodes are marked as always evaluated, and so we process both of
them. To have names for any values that are always computed, we assign
a fresh variable to each node that is marked by the $\MustEval$
analysis, reusing the same fresh variables for each $\phi$ node we
convert.  For example, in
Figure~\ref{reversion-branch-fusion-precise} we use the fresh
variable \textcircled{s} for the node \texttt{x*x}. We then produce a
$t$ and $f$ node for each $\phi$ node as before, replacing nodes which
will always be evaluated with parameter nodes of appropriate
variables.  Figure~\ref{reversion-branch-fusion-convert} shows that
$t$ for the first $\phi$ node is simply \texttt{0} whereas $f$ is
\texttt{s*x}, using \texttt{s} in place of \texttt{x*x}.
After all the $\phi$ nodes have been converted, we perform branch
fusion and sequencing as
before. Figure~\ref{reversion-branch-fusion-convert} shows this
updated process. The resulting code now performs the \verb-x*x-
computation before the branch.

One subtlety is that the $\MustEval$ analysis must satisfy some
minimal precision requirements. To see why this is needed, recall that
after the $\MustEval$ analysis is run, we now only process the $\phi$
nodes that are marked as always evaluated, leaving the remaining
$\phi$ nodes to recursive invocations. Thus, if $\MustEval$ doesn't
mark any nodes as being always evaluated, then we would not process
any $\phi$ nodes, which is a problem before after the
$\phi$-processing stage, we require there to be no more $\phi$ nodes
in the \PEG. As a result, we require the $\MustEval$ analysis to be
\emph{minimally precise}, as formalized in the following definition.

\begin{defn} We say that a $\MustEval$ analysis is minimally precise
if for any \PEG context $\context$, \hbox{$S = \MustEval(\context)$} implies the following properties:
$$
\begin{array}{rll}
(x,n) \in \context & \Rightarrow & n \in S \\
\constructor{\op}(n_1 \ldots n_k) \in S & \Rightarrow & n_1 \in S \wedge \ldots \wedge n_k \in S \\
\constructor{\phi}(c,a,b) \in S  & \Rightarrow & c \in S \\
\constructor{\langle s \rangle}(n_1 \ldots n_k) \in S & \Rightarrow & n_1 \in S \wedge \ldots \wedge n_k \in S
\end{array}
$$
\end{defn}
In essence the above simply states that $\MustEval$ must at least
return those nodes which can trivially be identified as always
evaluated. To see why this is sufficient to guarantee that we make
progress on $\phi$ nodes, consider the worst case, which is when
$\MustEval$ returns nothing more than the above trivially identified
nodes. Suppose we have a $\phi$ node that is not identified as always
evaluated. This node will be left to recursive invocations of
reversion, and at some point in the recursive invocation chain, as we
translate more and more $\phi$ nodes into branch nodes, our original
$\phi$ node will become a top-level node that is always evaluated (in
the \PEG context being processed in the recursive invocation). At that
point we will process it into a branch node.

\subsection{Loop-Invariant Code Motion}
\label{sec:loop-invariant-code-motion}

\begin{figure}
\includegraphics[width=\textwidth]{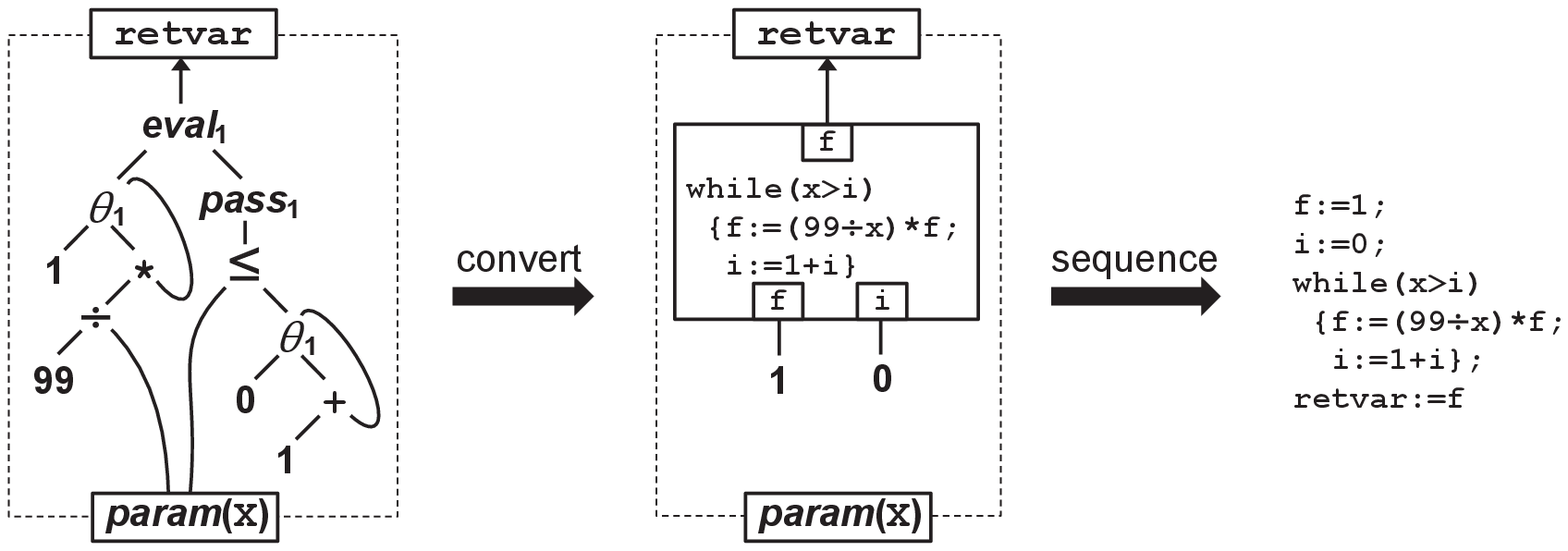}
\caption{Reversion of a PEG without applying loop-invariant code motion}
\label{reversion-loop-invariant-bad}
\end{figure}

The process described so far for converting $\eval$ nodes in
Section~\ref{sec:reversion-loops} duplicates code unnecessarily. In
particular, every node used by the loop body is copied into the loop
body, including loop-invariant nodes.
Figure~\ref{reversion-loop-invariant-bad} shows how placing the
loop-invariant operation \texttt{99$\div$x} into the loop body results
in the operation being evaluated in every iteration of the loop. In
the same way that in Section~\ref{branch-fusion} we updated our
processing of $\phi$ nodes to hoist computations that are common to
both the true and false sides, we can also update our processing of
$\eval$ nodes to hoist computations that are common across all loop
iterations, namely loop-invariant computations.

Recall that in Section~\ref{sec:build-in-axioms} we defined a
predicate $\invariant_\ell(n)$ which is true if the value of $n$ does
not vary in loop $\ell$, meaning that $n$ is invariant with respect to
loop $\ell$. The general approach will therefore be as follows: when
converting $\eval_\ell$, we simply keep any node $n$ that is invariant
with respect to $\ell$ outside of the loop body. Unfortunately, there
are some subtleties with making this approach work correctly. For
example, consider the \PEG from
Figure~\ref{reversion-loop-invariant-bad}. In this \PEG the $\div$
node is invariant with respect to loop 1 (\texttt{99$\div$x} produces
the same value no matter what iteration the execution is at). However,
if we were to evaluate the loop-invariant operation \texttt{99$\div$x}
before the loop, we would change the semantics of the program. In
particular, if \texttt{x} were $0$, the $\div$ operation would fail,
whereas the original program would simply terminate and return $1$
(because the original program only evaluates \texttt{99$\div$x} if
\texttt{x} is strictly greater than $0$). Thus, by pulling the
loop-invariant operation out of the loop, we have changed the
semantics of the program.

Even traditional formulations of loop-invariant code motion must deal
with this problem. The standard solution is to make sure that pulling
loop-invariant code outside of a loop does not cause it to execute in
cases where it would not have originally. In our \PEG setting, there
is a simple but very conservative way to guarantee this: when
processing an $\eval_\ell$ node, if we find a node $n$ that is
invariant with respect to $\ell$, we pull $n$ outside of the loop only
if there are no $\theta$ or $\phi$ nodes between the $\eval_\ell$ node
and $n$. The intuition behind disallowing $\phi$ and $\theta$ is that
both of these nodes can bypass evaluation of some of their children:
the $\phi$ chooses between its second and third child, bypassing the
other, and the $\theta$ node can bypass its second child if the loop
performs no iterations. This requirement is more restrictive than it
needs to be, since a $\phi$ node always evaluates its first child, and
so we could even allow $\phi$ nodes, as long as the loop invariant
node was used in the first child, not the second or third. In general,
it is possible to modularize the decision as to whether some code
executes more frequently than another in an \emph{evaluation-condition
  analysis}, or $\EvalCond$ for short. An $\EvalCond$ analysis would
compute for every node in the \PEG context an abstract evaluation
condition capturing under which cases the \PEG node is
evaluated. $\EvalCond$ is a generalization of the $\MustEval$ analysis
from Section~\ref{sec:hosting-from-branches}, and as with $\MustEval$,
an implementation has a lot of flexibility in defining $\EvalCond$. In
the more general setting of using an $\EvalCond$ analysis, we would
only pull a loop-invariant node if its evaluation condition is implied
by the evaluation condition of the $eval$ node being processed.

\begin{figure}
\includegraphics[width=\textwidth]{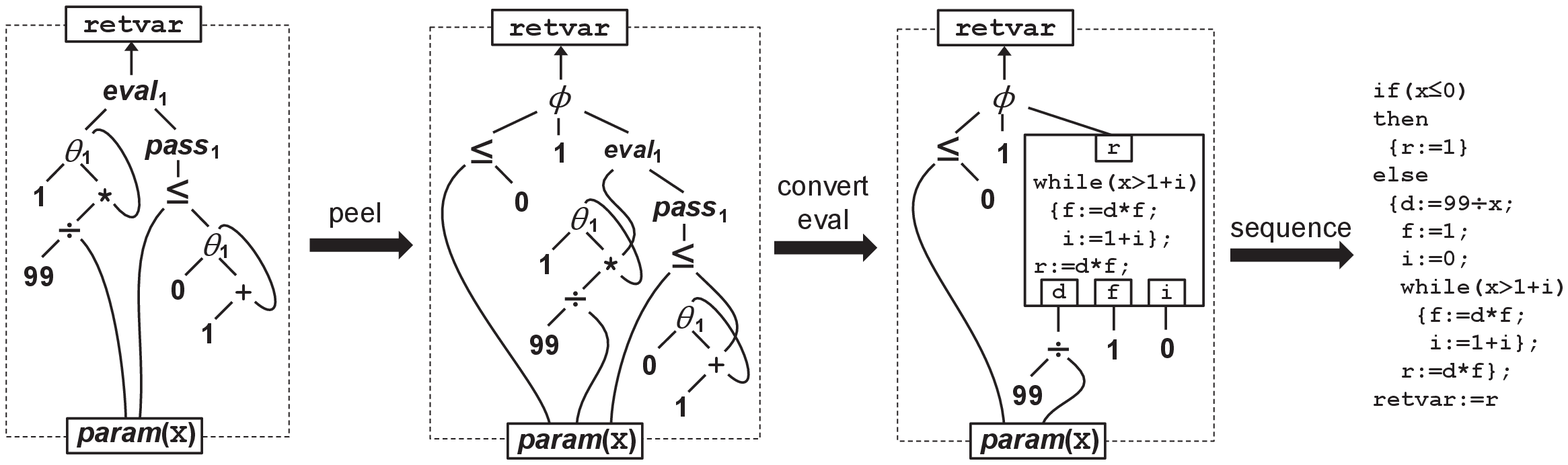}
\caption{Reversion of a PEG after peeling the loop once}
\label{reversion-loop-invariant-good}
\end{figure}

Since we now prevent loop-invariant code from being hoisted if it
would execute more often after being hoisted, we correctly avoid
pulling \texttt{99$\div$x} out of the loop. However, as is well known
in the compiler literature~\cite{AppelBook}, even in such cases it is
still possible to pull the loop-invariant code out of the loop by
performing loop peeling first. For this reason, we perform loop
peeling in the reversion process in cases where we find a loop
invariant-node that (1) cannot directly be pulled out because doing so
would make the node evaluate more often after hosting and (2) is
always evaluated provided the loop iterates a few times. Loop peeling
in the reversion process works very similarly to loop peeling as
performed in the engine (see Section~\ref{sec:loop-peeling}), except
that instead of using equality analyses, it is performed destructively
on the \PEG representation. Using the same starting example as before,
Figure~\ref{reversion-loop-invariant-good} shows the result of this
peeling process (step labeled ``peel''). After peeling, the $\phi$
node checks the entry condition of the original loop and evaluates the
peeled loop if this condition fails. Notice that the $\eval$ and
$\leq$ nodes in the new \PEG loop refer to the second child of the
$\theta$ nodes rather than $\theta$ nodes themselves, effectively
using the value of the loop variables after one iteration. An easy way
to read such nodes is to simply follow the edges of the \PEG; for
example, the ``$+$'' node can be read as ``$1+\theta_1(\ldots)$''.

In general, we repeat the peeling process until the desired
loop-invariant nodes used by the body of the loop are also used before
the body of the loop. In our example from
Figure~\ref{reversion-loop-invariant-good}, only one run of peeling is
needed. Notice that, after peeling, the $\div$ node is still
loop-invariant, but now there are no $\phi$ or $\theta$ nodes between
the $\eval$ node and the $\div$ node. Thus, although
\texttt{99$\div$x} is not always evaluated (such as when
\texttt{x$\leq$0} is true), it is always evaluated whenever the
$\eval$ node is evaluated, so it is safe to keep it out of the loop
body. As a result, when we convert the $\eval$ nodes to loop nodes, we
no longer need to keep the $\div$ node in the body of the loop, as
shown in
Figure~\ref{reversion-loop-invariant-good}. Figure~\ref{reversion-loop-invariant-good}
also shows the final generated \SIMPLE program for the peeled loop.
Note that the final code still has some code duplication: \texttt{1+i}
is evaluated multiple times in the same iteration, and \texttt{d*f} is
evaluated both when the while-loop guard succeeds and the guard
fails. These redundancies are difficult to remove without using more
advanced control structures that are not present in \SIMPLE. Our
implementation can take advantage of more advanced control structures
to remove these remaining redundancies. We do not show the details
here -- instead we refer the interested reader to our
technical report~\cite{peg2cfg}.

We should also note that, since the $\EvalCond$ analysis can handle
loop operators and subsumes the $\MustEval$ analysis, it is possible
to convert $\phi$ nodes before converting $\eval$ nodes, although both
still need to happen after the loop peeling stage. This rearrangement
enables more advanced redundancy elimination optimizations.

\section{The Peggy Instantiation}
\label{sec:implementation}

In this section we discuss details of our concrete implementation of
equality saturation as the core of an optimizer for Java bytecode
programs. We call our system \Peggy, named after our \PEG intermediate
representation. As opposed to the previous discussion of the \SIMPLE
language, \Peggy operates on the entire Java bytecode instruction set,
complete with side effects, method calls, heaps, and exceptions.
Recall from Figure~\ref{fig:optimize} that an instantiation of our
approach consists of three components: (1) an IR where equality
reasoning is effective, along with the translation functions
$\CfgToIr$ and $\IrToCfg$, (2) a saturation engine $\Saturate$, and
(3) a global profitability heuristic $\SelectBest$. We now describe
how each of these three components work in \Peggy.

\subsection{Intermediate Representation}

\Peggy uses the \PEG and \EPEG representations which, as explained in
Section~\ref{sec:semantics}, are well suited for our approach. Because
\Peggy is a Java bytecode optimizer, an additional challenge is to encode
Java-specific concepts like the heap and exceptions in \PEGs.

\mypara{Heap.} We model the heap using heap summaries which
we call $\sigma$ nodes. Any operation that can read and/or write some
object state may have to take and/or return additional $\sigma$
values.  Because Java stack variables cannot be modified except by
direct assignments, operations on stack variables are precise in our
\PEGs and do not involve $\sigma$ nodes.  None of these decisions of
how to represent the heap are built into the \PEG representation.  As
with any heap summarization strategy, one can have different levels of
abstraction, and we have simply chosen one where all objects are put
into a single summarization object $\sigma$.

\begin{figure}
\begin{center}
\includegraphics[width=3.5in]{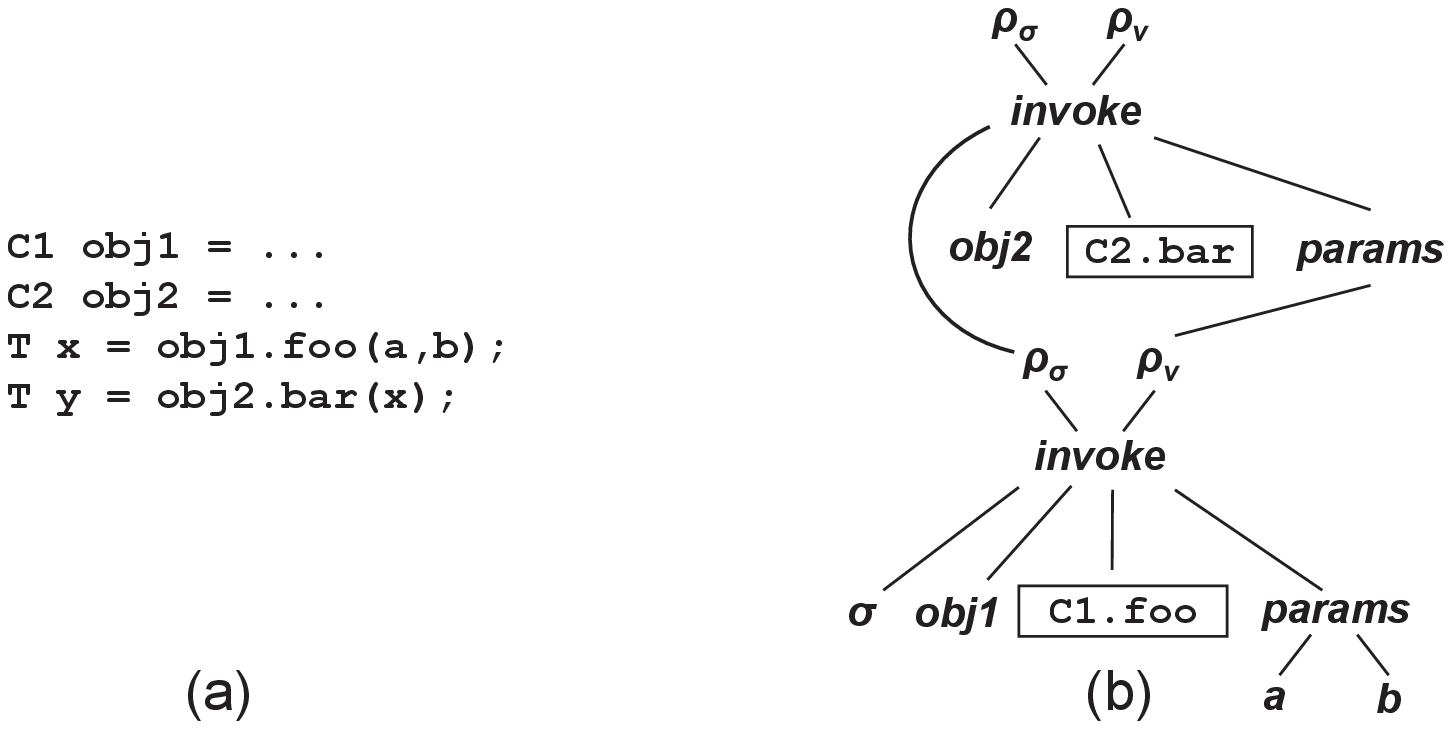}
\end{center}
\caption{Representation of Java method calls in a \PEG; (a) the
  original Java source code, (b) the corresponding \PEG.}
\label{fig:method-call}
\end{figure}

\mypara{Method calls.}  Figure~\ref{fig:method-call} shows an example
of how we encode two sequential method calls in a \PEG. Each
non-static method call operator has four parameters: the input
$\sigma$ heap summary, the receiver object of the method call, a
method identifier, and a list of actual parameters. A static method
call simply elides the receiver object. Logically, our {\tt invoke}
nodes return a tuple $(\sigma,v)$, where $\sigma$ is the resulting
heap summary and the $v$ is the return value of the method. The
operator $\rho_\sigma$ is used to project out the $\sigma$ value from
this tuple, and $\rho_v$ is used to project out the return value. From
this figure we can see that the call to {\tt bar} uses the output
$\sigma$ value from {\tt foo} as its input $\sigma$ value. This is how
we encode the control dependency between the two {\tt invoke}
operators; by representing it as a data dependency on the heap. Many
other Java operators have side effects that can be modeled as
modifications to our $\sigma$ heap summary, including array accesses,
field accesses, object creation, and synchronization. They similarly
take a $\sigma$ value as input and/or return one as output, which
encodes the data dependencies between the different heap operations.

\mypara{Exceptions.} In order to maintain the program state at points
where exceptions are thrown, we bundle the exception state into our
abstraction of the heap, namely the $\sigma$ summary nodes. As a
result, operations like division which may throw an exception, but do
not otherwise modify the heap, now take and return a $\sigma$ node (in
addition to their regular parameters and return values). 

The control flow of exceptions in a \PEG is computed from a CFG where
exceptional control flow has been made explicit. To build such an
exception-aware CFG, we introduce a new boolean function called
$\mathit{isException}$, which returns true if the current state of the
program is in an exceptional state. After every statement in the CFG
that could possibly cause an exception, we insert a conditional branch
on $\mathit{isException}$, which jumps to the exception handler if one
exists in the current method, or to the end of the CFG otherwise. Once
exceptions have been made explicit in the CFG, we simply use our
conversion algorithm to create a \PEG from the exception-aware
CFG. The $\mathit{isException}$ tester in the CFG gets translated into
a new \PEG operator $\mathit{isException}$, which reads the output
$\sigma$ node of an operation, and returns a boolean indicating
whether an exception occurred. The normal translation from CFG to \PEG
introduces the appropriate $\phi$ nodes to produce the correct values
in the exception vs. non-exception cases.

By bundling the exceptional state inside the heap summary node
$\sigma$, and using an explicit tester function
$\mathit{isException}$, the above encoding makes exceptions
explicit. Since our optimization process preserves heap behavior, it
therefore forces the observable state at the point where an exception
is thrown to be preserved.

\mypara{Reverting Uses of the Heap.}

One issue that affects our Peggy instantiation is that the $\sigma$
nodes representing the heap cannot be duplicated or copied when
running the program. More technically, heap values must be used
linearly in the CFG. Linear values complicate the reversion presented
in Section~\ref{sec:peg2cfg}, which assumes any value can be
duplicated freely. In order to adapt the algorithm from
Section~\ref{sec:peg2cfg} to handle linear values, one has to
``linearize'' heap uses, which consists of finding a valid order for
executing operations so that the heap does not need to be
duplicated. Since \PEGs come from actual programs which use the heap
linearly, we decided to linearize the uses of the heap within each
branch nodes and loop nodes generated in the reversion algorithm. This
approach, unfortunately, is not complete. There are cases (which we
can detect while running the conversion algorithm) where partitioning
into branch nodes and loop nodes, and \emph{then} trying to solve
linearity constraints, leads to unsolvable constraints, even though
the constraints are solvable when taking a global view of the
\PEG. Experimentally, this incompleteness occurs in less than 3\% of
the Java methods we compiled (in which case we simply do not optimize
the method). We briefly present the challenges behind solving this
problem and some potential solutions.

\begin{figure}[t]
\begin{center}
\includegraphics[width=5.0in]{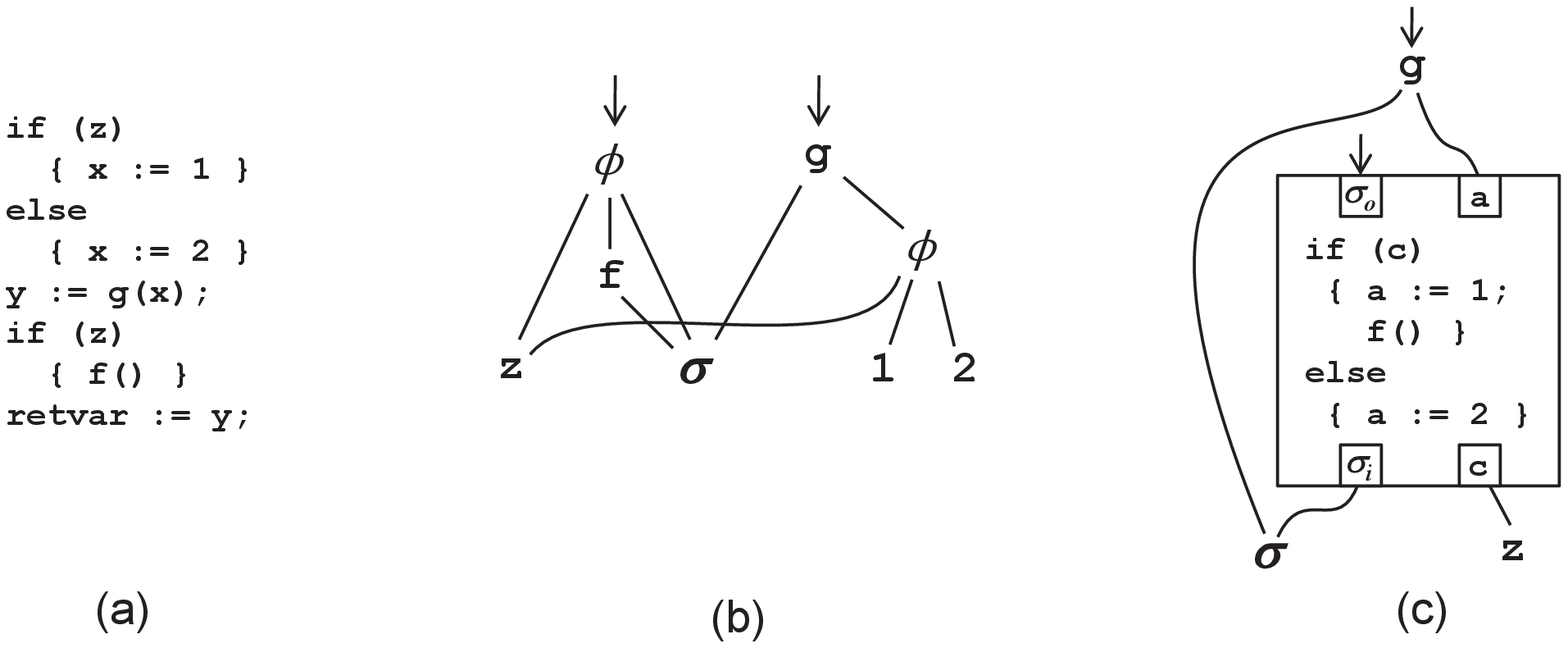}
\end{center}
\caption{Example demonstrating challenges of linearizing heap values.}
\label{fig:linear}
\end{figure}

Consider the \SIMPLE style code shown in Figure~\ref{fig:linear}(a),
and its \PEG representation in Figure~\ref{fig:linear}(b). Suppose
that \verb-g- is a side-effect free operation of one parameter, which
reads the heap and returns a value (but does not make any heap
modifications). Thus, in the \PEG representation, \verb-g- takes two
inputs: a heap value $\sigma$ (which is the original heap at the
entry), and its regular input parameter; \verb-g- also returns a
single value -- the computed return value of \verb-g- -- but it does
\emph{not} return a new heap value since it does not modify the
heap. We also assume that \verb-f- is an operation that reads and
writes to the heap, but returns \verb-void-. For example, \verb-f-
could just increment a global on the heap. Thus, in the \PEG
representation, \verb-f- takes an input heap parameter, and produces
an output heap (but does not produce a real return value since it
returns \verb-void-). We see from the code in part (a) that the return
value is produced by \verb-g-. However, since \verb-f- may have
produced a new heap, we must also encode the new heap as a return
value. Thus, the code in part (a) returns two values: its regular
return value, and a new heap, as shown with the two arrows in part
(b).

Looking at the \PEG in Figure~\ref{fig:linear}(b), we already see
something that may be cause for concern: the $\sigma$ node, which
represents a linear value in imperative code, is used in three
different places. However, there is absolutely nothing wrong with
doing this from the perspective of \PEGs, since \PEGs treat all values
(including heap values) functionally. The only problem is that to
convert this \PEG back to imperative code, we must find a valid
ordering of instructions so that we can run all the instructions in
the \PEG without having to duplicate the heap. The ordering in this
case is obvious: since \verb-f- modifies the heap, and \verb-g- does
not, run \verb-g- first, followed by \verb-f-.

The problem with our current approach is that our reversion first
creates and fuses branch nodes, and then it tries to linearize heap
values in each branch nodes. For example, Figure~\ref{fig:linear}(c)
shows the \PEG after creating a branch node for each of the two $\phi$
nodes in part (a), and fusing these two branch nodes
together. Unfortunately, once we have decided to place \verb-f- inside
of the branch block, the linearization constraints are not solvable
anymore: \verb-f- can no longer be executed after \verb-g- since
\verb-g- relies on the result of the branch block which executes
\verb-f-. This example shows the source of incompleteness in our
current reversion algorithm for linear values, and points to the fact
that one needs to solve the linearization constraints while taking a
global view of the \PEG.

However, devising a complete linearization algorithm for heap nodes
that takes a global view of the \PEG is non-trivial. Heap nodes can be
used by many functions, some of which may commute with others. Other
functions may occur in different control paths, such as one on the
true path of a branch and the other on the false path. In a \PEG,
identifying which situations these functions fall in is not a simple
task, especially considering that \PEGs can represent complex loops
and branches. This gets more challenging when using the saturation
engine. The saturation engine can determine that a write is accessing
a different array-index or field than a read, and therefore the write
commutes with the read. This information isn't available to the
reversion algorithm, though, so it cannot determine that the write can
be safely moved after the read.

\mypara{Equivalences.}

The \EPEG data structure contains a large number of \PEG{s}, and
stores the equivalences between their nodes. The number of
equivalences discovered during saturation can be exponential in the
number of axioms applied. So far in this paper, we have depicted
equivalences as dotted lines between \EPEG nodes. In reality, storing
a distinct object for each equivalence discovered would require a
large amount of memory. Instead, we represent equivalences between
nodes by partitioning the nodes into equivalence classes, then simply
storing the members of each class. Before saturation, every node is in
its own equivalence class. Saturation proceeds by merging pairs of
equivalence classes $A$ and $B$ together whenever an equality is
discovered between a node in $A$ and a node in $B$. The merging is
performed using Tarjan's union-find
algorithm~\cite{tarjan_union_find}. This approach makes the memory
overhead of an \EPEG proportional to the number of nodes it contains,
and requires no additional memory to represent the equivalences
between nodes. 

\subsection{Saturation Engine}
\label{sec:saturation-engine}

The saturation engine's purpose is to repeatedly dispatch equality
analyses. In our implementation an equality analysis is a pair
$(\Trigger,\Search)$ where $\Trigger$ is a trigger, which is an \EPEG
pattern with free variables, and $\Search$ is a callback function that
should be run when the pattern $\Trigger$ is found in the \EPEG. While
running, the engine continuously monitors the \EPEG for the presence
of the pattern $\Trigger$, and when it is discovered, the engine
constructs a \emph{matching substitution}, which is a map from each
node in the pattern to the corresponding \EPEG node. At this point,
the engine invokes $\Search$ with this matching substitution as a
parameter, and $\Search$ returns a set of equalities that the engine
adds to the \EPEG.  In this way, an equality analysis will be invoked
only when events of interest to it are discovered. Furthermore, the
analysis does not need to search the entire \EPEG because it is
provided with the matching substitution.

\begin{figure}
\begin{algorithmic}[1]
\declarefunction{\Saturate(\peg : \PEGType, A: 
\mbox{set of analyses}) : \EPEGType}
  \STATE {\bf let } $\epeg = \CreateInitialEPEG(\peg)$
  \WHILE{$\exists (\Trigger,\Search) \in A, \subst \in \SubstType  \ . \ \subst = \Match(\Trigger, \epeg )$}
    \STATE $\epeg := \AddEdges(\epeg, \Search(\subst, \epeg))$
  \ENDWHILE
  \STATE {\bf return } $\epeg$
\thickstraightline
\finishfunction
\end{algorithmic}
\caption{\Peggy's Saturation Engine.
  We use $\SubstType$ to denote the set of all
  substitutions from pattern nodes to \EPEG nodes.}
\label{fig:saturate}
\end{figure}

Figure~\ref{fig:saturate} shows the pseudo-code for \Peggy's
saturation engine. On line 2, the call to $\CreateInitialEPEG$ 
takes the input \PEG and generates an \EPEG that initially contains
no equality information.  The $\Match$ function invoked in the loop
condition performs pattern matching: if an analysis trigger occurs
inside an \EPEG, then $\Match$ returns the matching substitution. Once
a match occurs, the saturation engine uses $\AddEdges$ to add the
equalities computed by the analysis into the \EPEG.

A remaining concern in Figure~\ref{fig:saturate} is how to efficiently
implement the existential check on line 3.  The main challenge in
applying axioms lies in the fact that one axiom application may
trigger others. A naive implementation would repeatedly check all
axioms once an equality has been added, which leads to a lot of
redundant work since many of the axioms will not be triggered by the
new equality. Our original attempt at an implementation used this
approach, and it was unusably slow. To make our engine efficient, we
use well-known techniques from the AI community. In particular, our
problem of applying axioms is very similar to that of applying rules
to infer facts in rule-based systems, expert systems, or planning
systems. These systems make use of an efficient pattern matching
algorithm called the Rete
algorithm~\cite{ExpertSystemsBook}. Intuitively, the Rete algorithm
stores the state of every partially completed match as a finite state
machine. When new information is added to the system, rather than
reapplying every pattern to every object, it simply steps the state of
the relevant machines. When a machine reaches its accept state, the
corresponding pattern has made a complete match. We have adapted this
pattern matching algorithm to the \EPEG domain. The patterns of our
Rete network are the preconditions of our axioms. These generally look
for the existence of particular sub-\PEG{s} within the \EPEG, but can
also check for other properties such as loop-invariance. When a
pattern is complete it triggers the response part of the axiom, which
can build new nodes and establish new equalities within the \EPEG. The
creation of new nodes and equalities can cause other state machines to
progress, and hence earlier axioms applications may enable later ones.

In general, equality saturation may not terminate. Termination is also
a concern in traditional compilers where, for example, inlining
recursive functions can lead to unbounded expansion. By using triggers
to control when equality edges are added (a technique also used in
automated theorem provers), we can often avoid infinite expansion. The
trigger for an equality axiom typically looks for the left-hand-side
of the equality, and then makes it equal to the right-hand-side. On
occasion, though, we use more restrictive triggers to avoid expansions
that are likely to be useless. For example, the trigger for the axiom
equating a constant with a loop expression used to add edge D in
Figure~\ref{fig:inter-loop-strength-reduction} also checks that there
is an appropriate ``pass'' expression. In this way, it does not add a
loop to the \EPEG if there wasn't some kind of loop to begin
with. Using our current axioms and triggers, our engine completely
saturates \completionrate\% of the methods in our benchmarks.

In the remaining cases, we impose a limit on the number of expressions
that the engine fully processes (where fully processing an expression
includes adding all the equalities that the expression triggers). To
prevent the search from running astray and exploring a single
infinitely deep branch of the search space, we currently use a
breadth-first order for processing new nodes in the \EPEG. This
traversal strategy, however, can be customized in the implementation
of the Rete algorithm to better control the searching strategy in
those cases where an exhaustive search would not terminate.

\subsection{Global Profitability Heuristic}
\label{sec:heuristic}

\Peggy's $\SelectBest$ function uses a Pseudo-Boolean solver called
Pueblo~\cite{pueblo} to select which nodes from an \EPEG to include in
the optimized program. A Pseudo-Boolean problem is an integer linear
programming problem where all the variables have been restricted to 0
or 1. By using these 0-1 variables to represent whether or not nodes
have been selected, we can encode the constraints that must hold for
the selected nodes to be a CFG-like \PEG. In particular, for each node
or equivalence class $x$, we define a pseudo-boolean variable that
takes on the value 1 (true) if we choose to evaluate $x$, and 0
(false) otherwise. The constraints then enforce that the resulting
\PEG is CFG-like. The nodes assigned 1 in the solution that Pueblo
returns are selected to form the \PEG that $\SelectBest$ returns.

Recall that an \EPEG is a quadruple $\langle \Node, \Label, \Param, E
\rangle$, where $\langle \Node, \Label, \Param \rangle$ is a \PEG and
$E$ is a set of equalities which induces a set of equivalence classes
$\Node/E$. Also recall that for $n \in \Node$, $params(n)$ is the list
equivalence classes that are parameters to $n$. We use $q \in
params(n)$ to denote that equivalence class $q$ is in the list. 
For each node $n \in \Node$, we define a boolean variable $B_n$ that
takes on the value true if we choose to evaluate node $n$, and false
otherwise. For equivalence class $q \in (\Node/E)$, we define a
boolean variable $B_q$ that takes on the value true if we choose to
evaluate some node in the equivalence class, and false otherwise.  We
use $r$ to denote the equivalence class of the return value.

\Peggy generates the boolean constraints for a given \EPEG $\langle
\Node, \Label, \Param, E \rangle$ using the following
$\mathsf{Constraints}$ function (to simplify exposition, we describe
the constraints here as boolean constraints, but these can easily be
converted into the standard ILP constraints that Pueblo expects):
$$
\begin{array}{lll}
\mathsf{Constraints}( \langle \Node,\Label, \Param, E\rangle) 
 & \equiv & B_r \wedge
  \bigwedge_{n \in N} \mathsf{F}(n) \wedge
  \bigwedge_{q \in (N/E)} \mathsf{G}(q)\\[2pt]
\mathsf{F}(n) & \equiv & B_n \Rightarrow \bigwedge_{q\in params(n)}
  B_q \\[2pt]
\mathsf{G}(q) & \equiv & B_q \Rightarrow
  \bigvee_{n\in q} B_n
\end{array}
$$ Intuitively, these constraints state that: (1) we must compute the
return value of the function (2) for each node that is selected, we
must select all of its parameters (3) for each equivalence class
that is selected, we must compute at least one of its nodes.

Once the constraints are computed, \Peggy sends the following
minimization problem to Pueblo: 
$$\qquad\mathbf{min}\ \sum_{n \in N} B_n \cdot C_n\ \mathbf{s. t. }\
\mathsf{Constraints}( \langle \Node,\Label, \Param, E\rangle)$$ where
$C_n$ is the constant cost of evaluating $n$ according to our cost
model.  The nodes which are set to true in the solution that Pueblo
returns are selected to form a \PEG.

The cost model that we use assigns a constant cost $C_n$ to each node
$n$. In particular, $C_n = \bcost(n) \cdot k^{\maxvariance(n)}$, where
$\bcost(n)$ accounts for how expensive $n$ is by itself, and
$k^{\maxvariance(n)}$ accounts for how often $n$ is executed.  $k$ is
simply a constant, which we have chosen as 20. We use
$\maxvariance(n)$ to denote the loop nesting depth of $n$, computed as
follows (recalling definition~\ref{defn:invariant} of
$\invariant_\ell$ from Section~\ref{sec:build-in-axioms}):
$\maxvariance(n) = \max \{ \ell \mid \neg \invariant_\ell(n)\}$. Using
this cost model, Peggy asks Pueblo to minimize the objective function
subject to the constraints described above. Hence, the \PEG that
Pueblo returns is CFG-like and has minimal cost, according to our cost
model.

The above cost model is very simple, taking into account only the cost
of operators and how deeply nested they are in loops. Despite being
crude, and despite the fact that \PEGs pass through a reversion
process that performs branch fusion, loop fusion and loop-invariant
code motion, our cost model is a good predictor of \emph{relative}
performance. A smaller cost usually means that, after reversion, the
code will use cheaper operators, or will have certain operators moved
outside of loops, which leads to more efficient code. One of the main
contributors to the accuracy of our cost model is that
$\maxvariance(n)$ is defined in terms of $\invariant_\ell$, and
$\invariant_\ell$ is what the reversion process uses for pulling code
outside of loops (see
Section~\ref{sec:loop-invariant-code-motion}). As a result, the cost
model can accurately predict at what loop depth the reversion
algorithm will place a certain node, which makes the cost model
relatively accurate, even in the face of reversion.

There is an additional subtlety in the above encoding. Unless we are
careful, the Pseudo-Boolean solver can return a \PEG that contains
cycles in which none of the nodes are $\theta$ nodes. Such \PEGs are
not CFG-like. For example, consider the expression $x + 0$. After
axiom application, this expression (namely, the $+$ node) will become
equivalent to the $x$ node. Since $+$ and $x$ are in the same
equivalence class, the above encoding allows the Pseudo-Boolean solver
to select $+$ with $+$ and $0$ as its parameters. To forbid such
invalid \PEGs, we explicitly encode that all cycles must have a
$\theta$ node in them. In particular, for each pair of nodes $i$ and
$j$, we define a boolean variable $B_{i\leadsto j}$ that represents
whether or not $i$ reaches $j$ without going through any $\theta$
nodes in the selected solution. We then state rules for how these
variables are constrained. In particular, if a non-$\theta$ node $i$
is selected ($B_i$) then $i$ reaches its immediate children (for each
child $j$ of $i$, $B_{i\leadsto j}$). Also, if $i$ reaches a
non-$\theta$ node $j$ in the current solution ($B_{i\leadsto j}$), and
$j$ is selected ($B_j$), then $i$ reaches $j$'s immediate children
(for each child $k$ of $j$, $B_{i\leadsto k}$). Finally, we add the
constraint that for each non-$\theta$ node $n$, $B_{n\leadsto n}$ must
be false.

It is worth noting that the particular choice of Pseudo-Boolean solver
is independent of the correctness of this encoding. We have chosen to
use Pueblo because we have found that it runs efficiently on the types
of problems that Peggy generates, but it is a completely pluggable
component of the overall system. This modularity is beneficial because
it makes it easy to take advantage of advances in the field of
Pseudo-Boolean solvers. In fact, we have tested two other solvers
within our framework: Minisat~\cite{minisat} and
SAT4J~\cite{sat4j}. We have found that occasionally Minisat performs
better than Pueblo, and that SAT4J uniformly performs worse than the
other two. These kinds of comparisons are very simple to do with our
framework since we can easily swap one solver for another.

\subsection{Eval and Pass}
\label{evalpass}

One might wonder why we have $\eval$ and $\pass$ as separate operators
rather than combining them into a single operator, say $\mu$.  At this
point we can reflect upon this design decision and argue why we
maintain the separation. One simple reason why we maintain this
separation is that there are useful operators other than $\pass$ that
can act as the second child to an $\eval$. The loop peeling example
from Section~\ref{sec:loop-peeling} gives three such examples, namely
$S$, $Z$, and $\phi$.  It is also convenient to have each loop
represented by a single node, namely the $\pass$ node.  This does not
happen when using $\mu$ nodes, since there would be many $\mu$ nodes
for each loop.  These $\mu$ nodes would all share the same break
condition, but we illustrate below why that does not suffice.

Suppose that during equality saturation, some expensive analysis
decides the engine should explore peeling a loop.  Using $\eval$ and
$\pass$, this expensive analysis could initiate the peeling process by
simply replacing the $\pass$ node of that loop with an appropriate
$\phi$ node. Afterward, simple axioms would apply to each $\eval$
node independently in order to propagate the peeling process.  Using
$\mu$ nodes on the other hand, the expensive analysis would have to
explicitly replace every $\mu$ node with its peeled version.  Thus,
using $\eval$ and $\pass$ allows the advanced analysis to initiate
peeling only once, whereas using $\mu$ nodes requires the advanced
analysis to process each $\mu$ node separately.

Next we consider our global profitability heuristic in this situation
after loop peeling has been performed. Now for every $\eval$ or $\mu$
node there are two versions: the peeled version and the original
version.  Ideally we would select either only peeled versions or only
original versions.  If we mix them up, this forces us to have two
different versions of the loop in the final result.  In our
Pseudo-Boolean heuristic with $\eval$ and $\pass$ nodes, we encourage
the use of only one loop by making $\pass$ nodes expensive; thus the
solver would favor PEGs with only one $\pass$ node (i.e. one loop)
over two $\pass$ nodes.  However, there is no way to encourage this
behavior using $\mu$ nodes as there is no single node which represents
a loop. The heuristic would select the peeled versions for those
$\mu$ nodes where peeling was beneficial and the original versions for
those $\mu$ nodes where peeling was detrimental, in fact encouraging
an undesirable mix of peeled and original versions.

For similar reasons, the separation of $\eval$ and $\pass$ nodes is
beneficial to the process of reverting PEGs to CFGs.  A loop is peeled
by rewriting the $\pass$ node, and then all $\eval$ nodes using that
$\pass$ node are automatically peeled simultaneously.  Thus, when an
optimization such as loop-invariant code motion determines that a loop
needs to be peeled, the optimization needs to make only one change and
the automatic rewriting mechanisms will take care of the rest.

To summarize, the separation of $\eval$ and $\pass$ nodes makes it
easy to ensure that any restructuring of a loop is applied
consistently: the change is just made to the $\pass$ node and the rest
follows suit. This allows restructuring analyses to apply once and be
done with. The separation also enables us to encourage the global
profitability heuristic to select \PEGs with fewer loops.

\section{Evaluation}
\label{sec:eval}

In this section we use our \Peggy implementation to validate three
hypotheses about our approach for structuring optimizers: our approach
is practical both in terms of space and time
(Section~\ref{sec:overhead}), it is effective at discovering both
simple and intricate optimization opportunities
(Section~\ref{sec:opts}), and it is effective at performing
translation validation (Section~\ref{sec:tv}).

\subsection{Time and space overhead}
\label{sec:overhead}

To evaluate the running time of the various \Peggy components, we
compiled SpecJVM, which comprises \compilemethods methods. For
1\% of these methods, Pueblo exceeded a one minute timeout we
imposed on it, in which case we just ran the conversion to \PEG and
back. We imposed this timeout because in some rare cases, Pueblo 
runs too long to be practical.

\clearlineno

\newcommand{\tabindent}{}
\newcommand{\TableSection}[1]{
\hline\hline \textbf{#1} & \textbf{Description} \\ \hline}
\newcommand{\FirstTableSection}[1]{
\hline \textbf{#1} & \textbf{Description} \\ \hline}
\newcommand{\EndTableSection}{}

\newcommand{\TableEntry}[2]{\tabindent\lineno~~#1 & #2 \\}

\begin{figure}
\begin{center}
\begin{tabular}{|@{}l|@{\hspace{2pt}}l@{\hspace{2pt}}|}
\FirstTableSection{(a) EQ Analyses}
\TableEntry
{Built-in \EPEG ops}
{Axioms about primitive \PEG nodes ($\phi$, $\theta$, $\eval$, $\pass$) }
\TableEntry
{Basic Arithmetic}
{Axioms about arithmetic operators like $+$, $-$, $*$,
  $/$, $\text{\tt <<}$, $\text{\tt >>}$}
\TableEntry
{Constant Folding}
{Equates a constant expression with its constant value}
\TableEntry
{Java-specific}
{Axioms about Java operators like field/array accesses}
\TableEntry
{TRE}
{Replaces the body of a tail-recursive procedure with a loop}
\TableEntry
{Method Inlining}
{Inlining based on intraprocedural class analysis}
\TableEntry
{Domain-specific}
{User-provided axioms about application domains}
\EndTableSection

\hline\TableSection{(b) Optimizations}
\TableEntry
{Constant Prop/Fold}
{Traditional Constant Propagation and Folding}
\TableEntry
{Simplify Algebraic}
{Various forms of traditional algebraic simplifications}
\TableEntry
{Peephole SR}
{Various forms of traditional peephole optimizations}
\TableEntry
{Array Copy Prop}
{Replace read of array element by last expression written}
\TableEntry
{CSE for Arrays}
{Remove redundant array accesses}
\TableEntry
{Loop Peeling}
{Pulls the first iteration of a loop outside of the loop}
\TableEntry
{LIVSR}
{Loop-induction-variable Strength Reduction}
\startrange{unanticipated}\TableEntry
{Interloop SR}
{Optimization described in Section~\ref{sec:loops}}
\TableEntry
{Entire-loop SR}
{Entire loop becomes one op, \eg~$n$ incrs becomes ``plus $n$''}
\TableEntry
{Loop-op Factoring}
{Factor op out of a loop, \eg~multiplication}
\TableEntry
{Loop-op Distributing}
{Distribute op into loop, where it cancels out with another}
\TableEntry
{Partial Inlining}
{Inlines part of method in caller, but keeps the call}
\TableEntry
{Polynomial Factoring}
{Evaluates a polynomial in a more efficient manner
 \finishrange{unanticipated}}
\EndTableSection

\hline\TableSection{(c) DS Opts}
\TableEntry
{DS LIVSR}
{LIVSR on domain ops like matrix addition and multiply}
\TableEntry
{DS Code Hoisting}
{Code hoisting based on domain-specific invariance axioms}
\TableEntry
{DS Remove Redundant\linelabel{redundancy-removal}}
{Removes redundant computations based on domain axioms}
\TableEntry
{Temp. Object Removal\linelabel{temp-obj-removal}}
{Remove temp objects made by calls to, \eg, matrix libraries}
\TableEntry
{Math Lib Specializing}
{Specialize matrix algs based on, \eg, the size of the matrix}
\TableEntry
{Design-pattern Opts}
{Remove overhead of common design patterns}
\TableEntry
{Method Outlining}
{Replace code by method call performing same computation}
\TableEntry
{Specialized Redirect}
{Replace call with more efficient call based on calling context}
\EndTableSection
\hline

\end{tabular}
\end{center}
\caption{Optimizations performed by \Peggy. 
Throughout this table we use the following abbreviations:
EQ means ``equality'', 
DS means ``domain-specific'',
TRE means ``tail-recursion elimination'',
SR means ``strength reduction''}

\label{fig:opts}
\end{figure}

The following table shows the average time in milliseconds taken per
method for the 4 main \Peggy phases (for Pueblo, a timeout counts as
60 seconds).

\vspace{3pt}
\[
\begin{tabular}{|c|c|c|c|c|}
\hline
 & CFG to \PEG & Saturation & Pueblo & \PEG to CFG \\
\hline
Time & 13.9 ms & 87.4 ms &  1,499 ms & 52.8 ms \\
\hline
\end{tabular}
\]
\vspace{3pt}

All phases combined take slightly over 1.5 seconds. An end-to-end run
of Peggy is on average 6 times slower than Soot with all of its
intraprocedural optimizations turned on. Nearly all of our time is
spent in the Pseudo-Boolean solver. We have not focused our efforts on
compile-time, and we conjecture there is significant room for
improvement, such as better pseudo-boolean encodings, or other kinds
of profitability heuristics that run faster.

Since \Peggy is implemented in Java, to evaluate memory footprint, we
limited the JVM to a heap size of \heapsize MB, and observed that \Peggy
was able to compile all the benchmarks without running out of 
memory.

In \completionrate\% of compiled methods, the engine ran to complete
saturation, without imposing bounds. For the remaining cases, the
engine limit of \enginebound was reached, meaning that the engine 
ran until fully processing \enginebound expressions in the \EPEG, along
with all the equalities they triggered. In these cases, we cannot
provide a completeness guarantee, but we can give an estimate of the
size of the explored state space. In particular, using just \heapsize
MB of heap, our \EPEGs represented more than \lowerbound versions of
the input program (using geometric average).

\subsection{Implementing optimizations} 
\label{sec:opts}

The main goal of our evaluation is to demonstrate that common, as well
as unanticipated, optimizations result in a natural way from our
approach. To achieve this, we implemented a set of basic equality
analyses, listed in Figure~\ref{fig:opts}(a). We then manually browsed
through the code that \Peggy generates on a variety of benchmarks
(including SpecJVM) and made a list of the optimizations that we
observed. Figure~\ref{fig:opts}(b) shows the optimizations that we
observed fall out from our approach using equality analyses 1 through
6, and Figure~\ref{fig:opts}(c) shows optimizations that we observed
fall out from our approach using equality analyses 1 through 7. Based
on the optimizations we observed, we designed some micro-benchmarks
that exemplify these optimizations. We then ran \Peggy on each of these
micro-benchmarks to show how much these optimizations improve the code
when isolated from the rest of the program.

Figure~\ref{fig:benchmarks-runtime-opts} shows our experimental
results for the runtimes of the micro-benchmarks listed in
Figure~\ref{fig:opts}(b) and (c). The y-axis shows run-time normalized
to the runtime of the unoptimized code. Each number along the x-axis
is a micro-benchmark exemplifying the optimization from the
corresponding row number in Figure~\ref{fig:opts}. The ``rt'' and
``sp'' columns correspond to our larger raytracer benchmark and
SpecJVM, respectively. The value reported for SpecJVM is the average
ratio over all benchmarks within SpecJVM. Our experiments with Soot
involve running it with all intra-procedural optimizations turned on,
which include: common sub-expression elimination, lazy code motion,
copy propagation, constant propagation, constant folding, conditional
branch folding, dead assignment elimination, and unreachable code
elimination.  Soot can also perform interprocedural optimizations,
such as class-hierarchy-analysis, pointer-analysis, and
method-specialization. We did not enable these optimizations when
performing our comparison against Soot, because we have not yet
attempted to express any interprocedural optimizations in \Peggy. In
terms of runtime improvement, \Peggy performed very well on the
micro-benchmarks, optimizing all of them by at least 10\%, and in many
cases much more. Conversely, Soot gives almost no runtime
improvements, and in some cases makes the program run slower. For the
larger raytracer benchmark, \Peggy is able to achieve a 7\% speedup,
while Soot does not improve performance. On the SpecJVM benchmarks
both \Peggy and Soot had no positive effect, and \Peggy on average
made the code run slightly slower. This leads us to believe that
traditional intraprocedural optimizations on Java bytecode generally
produce only small gains, and in this case there were few or no
opportunities for improvement.

\begin{figure}[t]
  \input{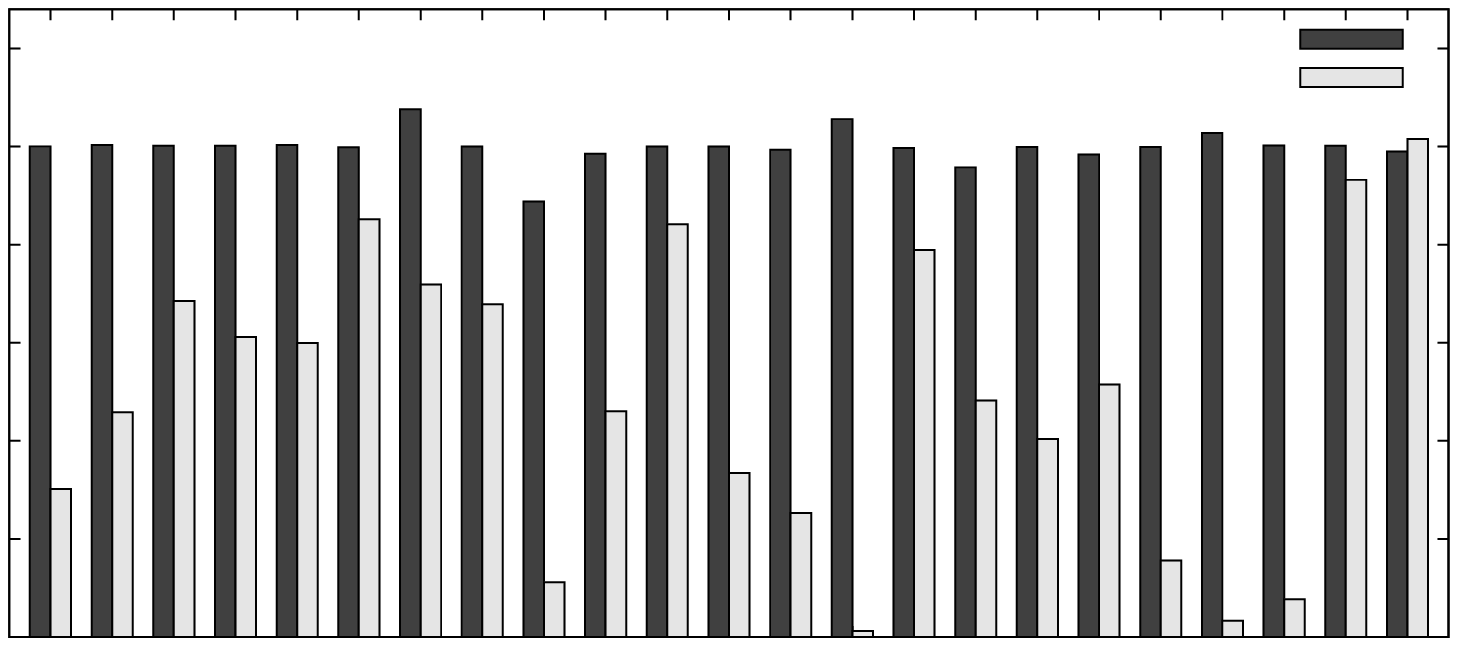}
  \caption{Runtimes of generated code from Soot and \Peggy, normalized
    to the runtime of the unoptimized code. The x-axis denotes the
    optimization number from Figure~\ref{fig:opts}, where ``rt'' is
    our raytracer benchmark and ``sp'' is the average over the SpecJVM
    benchmarks.}
  \label{fig:benchmarks-runtime-opts}
\end{figure}

With effort similar to what would be required for a compiler writer to
implement the optimizations from part (a), our approach enables the
more advanced optimizations from parts (b) and (c). 
\Peggy performs some optimizations (for example
\refrange{unanticipated}) that are quite complex given the simplicity
of its equality analyses. To implement such optimizations in a
traditional compiler, the compiler writer would have to explicitly
design a pattern that is specific to those optimizations. In contrast,
with our approach these optimizations fall out from the interaction of
basic equality analyses without any additional developer effort, and
without specifying an order in which to run them. Essentially, \Peggy
finds the right sequence of equality analyses to apply for producing
the effect of these complex optimizations.

With the addition of domain-specific axioms, our approach enables even
more optimizations, as shown in part (c). To give a flavor for these
domain-specific optimizations, we describe two examples.

The first is a ray tracer application (5 KLOCs) that one of the
authors had previously developed. To make the implementation clean and
easy to understand, the author used immutable vector objects in a
functional programming style. This approach however introduces many
intermediate objects. With a few simple vector axioms, \Peggy is able
to remove the overhead of these temporary objects, thus performing a
kind of deforestation optimization. This makes the application 7\%
faster, and reduces the number of allocated objects by 40\%. Soot is
not able to recover any of the overhead, even with interprocedural
optimizations turned on. This is an instance of a more general
technique where user-defined axioms allow \Peggy to remove temporary
objects (optimization~\reflineno{temp-obj-removal} in
Figure~\ref{fig:opts}).

Our second example targets a common programming idiom involving {\tt
List}s, which consists of checking that a {\tt List} contains an
element $e$, and if it does, fetching and using the index of the
element. If written cleanly, this pattern would be implemented with a
branch whose guard is $\mathtt{contains}(e)$ and a call to
$\mathtt{indexOf}(e)$ on the true side of the branch. Unfortunately,
$\mathtt{contains}$ and $\mathtt{indexOf}$ would perform the same
linear search, which makes this clean way of writing the code
inefficient. Using the equality axiom \hbox{$l.\mathtt{contains}(e) =
(l.\mathtt{indexOf}(e) \neq -1)$}, \Peggy can convert the clean code
into the hand-optimized code that programmers typically write, which
stores $\mathtt{indexOf}(e)$ into a temporary, and then branches if
the temporary is not $-1$. An extensible rewrite system would not be
able to provide the same easy solution: although a rewrite of
$l.\mathtt{contains}(e)$ to $(l.\mathtt{indexOf}(e) \neq -1)$ would
remove the redundancy mentioned above, it could also degrade
performance in the case where the list implements an efficient
hash-based $\mathtt{contains}$. In our approach, the equality simply
adds information to the \EPEG, and the profitability heuristic can
decide after saturation which option is best, taking the entire
context into account. In this way our approach transforms
$\mathtt{contains}$ to $\mathtt{indexOf}$, but only if
$\mathtt{indexOf}$ would have been called anyway.

These two examples illustrate the benefits of user-defined axioms. In
particular, the clean, readable, and maintainable way of writing code
can sometimes incur performance overheads. User-defined axioms allow
the programmer to reduce these overheads while keeping the code base
clean of performance-related hacks. Our approach makes domain-specific
axioms easier to add for the end-user programmer, because the
programmer does not need to worry about what order the user-defined
axioms should be run in, or how they will interact with the compiler's
internal optimizations. The set of axioms used in the programs from
Figure~\ref{fig:opts} is presented in Appendix~\ref{sec:axioms}.

\subsection{Translation Validation} 
\label{sec:tv}

We used \Peggy to perform translation validation for the Soot
optimizer~\cite{vall99soot}.  In particular, we used Soot to optimize
a set of benchmarks with all of its intraprocedural optimizations
turned on, which include: common sub-expression elimination, lazy code
motion, copy propagation, constant propagation, constant folding,
conditional branch folding, dead assignment elimination, and
unreachable code elimination.  The benchmarks included SpecJVM, along
with other programs, comprising a total of \tvnummethods
methods. After Soot finished compiling, for each method we asked
\Peggy's saturation engine to show that the original method was
equivalent to the corresponding method that Soot produced. The engine
was able to show that \tvsuccess\% of methods were compiled correctly.

Among the cases that \Peggy was unable to validate, we found three
methods that Soot optimized \emph{incorrectly}. In particular, Soot
incorrectly pulled statements outside of an intricate loop,
transforming a terminating loop into an infinite loop. It is a
testament to the power of our approach that it is able not only to
perform optimizations, but also to validate a large fraction of Soot
runs, and that in doing so it exposed a bug in Soot. Furthermore,
because most false positives are a consequence of our coarse
heap model (single $\sigma$ node), a finer-grained model can increase
the effectiveness of translation validation, and it would also enable
more optimizations.

Our equality saturation engine can easily be extended so that, after
each translation validation, it generates a machine-checkable proof of
equivalence. With this in place, the trusted computing base for our
translation validator would only be: (1) the proof checker, (2) the
built-in axioms used in translation validation, most of which we have
proved by hand, and (3) the conversion from Java bytecode to PEG.

\section{Related Work}
\label{sec:rel}

\mypara{Superoptimizers.} Our approach of computing a set of
programs and then choosing from this set is related to the approach
taken by super-optimizers~\cite{massalini, gcc-super-opt,
AikenSuperOpt, burg}. Superoptimizers strive to produce optimal code,
rather than simply improve programs. Although super-optimizers can
generate (near) optimal code, they have so far scaled only to
small code sizes, mostly straight line code. Our approach, on the
other hand, is meant as a general purpose paradigm that can optimize
branches and loops, as shown by the inter-loop optimization from
Section~\ref{sec:loops}.

Our approach was inspired by Denali~\cite{Denali}, a super-optimizer for
finding near-optimal ways of computing a given basic block. Denali
represents the computations performed in the basic block as an
expression graph, and applies axioms to create an E-graph data
structure representing the various ways of computing the values in the
basic block. It then uses repeated calls to a SAT solver to find the
best way of computing the basic block given the equalities stored in
the E-graph.
The biggest difference between our work and Denali is that our
approach can perform intricate optimizations involving branches and
loops. On the other hand, the Denali cost model is more precise than
ours because it assigns costs to entire sequences of operations, and
so it can take into account the effects of scheduling and register
allocation.

\mypara{Rewrite-Based Optimizers.} Axioms or rewrite-rules have
been used in many compilation systems, for example TAMPR~\cite{tampr},
ASF+SDF~\cite{ASF-SDF}, the ML compilation system of Visser {\it et
al.}~\cite{visseretal98}, and Stratego~\cite{stratego}. These systems,
however, perform transformations in sequence, with each axiom or
rewrite rule destructively updating the IR. Typically, such compilers
also provide a mechanism for controlling the application of rewrites
through built-in or user-defined \emph{strategies}. Our approach, in
contrast, does not use strategies -- we instead simultaneously explore
all possible optimization orderings, while avoiding redundant
work. Furthermore, even with no strategies, we can perform a
variety of intricate optimizations. 

\mypara{Optimization Ordering.} Many research projects have been
aimed at mitigating the phase ordering problem, including automated
assistance for exploring enabling and disabling properties of
optimizations~\cite{Whitfield90, WhitfieldSoffa97}, automated
techniques for generating good sequences~\cite{lctes99, lctes04,
comp-for-21st-century}, manual techniques for combining analyses and
optimizations~\cite{click-cooper}, and automated techniques for the
same purpose~\cite{Lerner-etal-02}. However, we tackle the problem
from a different perspective than previous approaches, in particular,
by simultaneously exploring all possible sequences of optimizations,
up to some bound. Aside from the theoretical guarantees from
Section~\ref{sec:formal}, our approach can do well even if every part
of the input program requires a different ordering. 

\mypara{Translation Validation.} Although previous approaches to
translation validation have been explored~\cite{pnueli98translation,
necula:00:trans-valid, voc:jucs}, our approach has the advantage that
it can perform translation validation by using the same technique as
for program optimization.

\mypara{Intermediate Representations.} Our main contribution is an
approach for structuring optimizers based on equality
saturation. However, to make our approach effective, we have also
designed the \EPEG representation. There has been a long line of work
on developing IRs that make analysis and optimizations easier to
perform~\cite{ssa, ZadeckVariableEquality, GSSA, ThinGSSA, PDG, VDG,
GlobalCodeMotionValueNumbering, VFG, DFG}. The key distinguishing
feature of \EPEGs is that a single \EPEG can represent
many optimized versions of the input program, which allows us to
use global profitability heuristics and to perform translation
validation.

We now compare the \PEG component of our IR with previous IRs. \PEGs
are related to SSA~\cite{ssa}, gated SSA~\cite{GSSA} and thinned-gated
SSA~\cite{ThinGSSA}. The $\mu$ function from gated SSA is similar to
our $\theta$ function, and the $\eta$ function is similar to our
$\eval$/$\pass$ pair. However, in all these variants of SSA, the SSA
nodes are inserted \emph{into} the CFG, whereas we do not keep the CFG
around. The fact that \PEGs are not tied to a CFG imposes fewer
placement constraints on IR nodes, allowing us to implicitly
restructure the CFG simply by manipulating the \PEG, as shown in
Section~\ref{sec:local-nonlocal}. Furthermore, the conversion from any
of the SSA representations back to imperative code is extremely simple
since the CFG is already there. It suffices for each assignment $x :=
\phi(a,b)$ to simply insert the assignments $x:=a$ and $x:=b$ at the
end of the two predecessors CFG basic blocks. The fact that our \PEG
representation is not tied to a CFG makes the conversion from \PEGs
back to a CFG-like representation much more challenging, since it
requires reconstructing explicit control information.

The Program Dependence Graph~\cite{PDG} (PDG) represents control
information by grouping together operations that execute in the same
control region. The representation, however, is still
statement-based. Also, even though the PDG makes many analyses and
optimizations easier to implement, each one has to be developed
independently. In our representation, analyses and optimizations fall
out from a single unified reasoning mechanism.

The Program Dependence Web~\cite{PDW} (PDW) combines the PDG with
gated SSA. Our conversion algorithms have some similarities with the
ones from the PDW. The PDW however still maintains explicit PDG
control edges, whereas we do not have such explicit control edges,
making converting back to a CFG-like structure more complex.

Dependence Flow Graphs~\cite{DFG} (DFGs) are a complete and
executable representation of programs based on
dependencies. However, DFGs employ a side-effecting storage model
with an imperative \textit{store} operation, whereas our
representation is entirely functional, making equational reasoning
more natural.

Like \PEGs, the Value Dependence Graph~\cite{VDG} (VDG) is a complete
functional representation. VDGs use $\lambda$ nodes (i.e. regular
function abstraction) to represent loops, whereas we use specialized
$\theta$, $\eval$ and $\pass$ nodes. Using $\lambda$s as a key
component in an IR is problematic for the equality saturation process.
In order to effectively reason about $\lambda$s one must particularly
be able to reason about substitution.  While this is possible to do
during equality saturation, it is not efficient.  The reason is that
equality saturation is also being done to the body of the $\lambda$
expression (essentially optimizing the body of the loop in the case of
VDGs), so when the substitution needs to be applied, it needs to be
applied to all versions of the body and even all future versions of
the body as more axioms are applied. Furthermore, one has to determine
when to perform $\lambda$ abstraction on an expression, that is to
say, turn $e$ into $(\lambda x.e_{\it body})(e_{\it arg})$, which
essentially amounts to pulling $e_{\it arg}$ out of $e$. Not only can
it be challenging to determine when to perform this transformation,
but one also has to take particular care to perform the transformation
in a way that applies to \emph{all} equivalent forms of $e$ and
$e_{\it arg}$.

The problem with $\lambda$ expressions stems in fact from a more
fundamental problem: $\lambda$ expressions use \emph{intermediate
  variables} (the parameters of the $\lambda$s), and the level of
indirection introduced by these intermediate variables adds reasoning
overhead. In particular, as was explained above for VDGs, the added
level of indirection requires reasoning about substitution, which in
the face of equality saturation is cumbersome and inefficient. An
important property of \PEGs is that they have no intermediate
variables.  The overhead of using intermediate variables is also why
we chose to represent effects with an effect token rather than using
the techniques from the functional languages community such as
monads~\cite{WadlerMonads1990,WadlerMonads1995, WadlerMonads1998} or
continuation-passing style~\cite{AppelBookContinuations1991,
  KennedyContinuations2007, HatcliffContinuations1994,
  FlanaganContinuations1993, AppelShrinkingLambda1997}, both of which
introduce indirection through intermediate variables. It is also why
we used recursive expressions rather than using syntactic fixpoint
operators.

\mypara{Dataflow Languages.}  Our \PEG intermediate representation
is related to the broad area of dataflow
languages~\cite{dataflowlangs}. The most closely related is the Lucid
programming language~\cite{lucid}, in which variables are maps from
iteration counts to possibly undefined values, as in our
\PEGs. Lucid's \textbf{first}/\textbf{next} operators are similar to
our $\theta$ nodes, and Lucid's \textbf{as soon as} operator is
similar to our $\eval$/$\pass$ pair. However, Lucid and \PEGs differ
in their intended use and application. Lucid is a programming language
designed to make formal proofs of correctness easier to do, whereas
\Peggy uses equivalences of \PEG nodes to optimize code expressed in
existing imperative languages. Furthermore, we incorporate a
$\monotonize$ function into our semantics and axioms, which guarantees
the correctness of our conversion to and from CFGs with loops.

\mypara{Theorem Proving.} Because most of our reasoning is
performed using simple axioms, our work is related to the broad area
of automated theorem proving. The theorem prover that most inspired
our work is Simplify~\cite{simplify}, with its E-graph data structure
for representing equalities~\cite{NelsonOpenCongruence80}. Our \EPEGs
are in essence specialized E-graphs for reasoning about
\PEGs. Furthermore, the way our analyses communicate through equality
is conceptually similar to the equality propagation approach used in
Nelson-Oppen theorem provers~\cite{Nelson:1979:SCD}.

\mypara{Execution Indices.} Execution indices identify the state of
progress of an execution~\cite{gotoharmful,execindex}. The call stack
typically acts as the interprocedural portion, and the loop iteration
counts in our semantics can act as the intraprocedural portion. As a
result, one of the benefits of \PEGs is that they make intraprocedural
execution indices explicit.

\section{Conclusion and future work}

We have presented a new approach to structuring optimizers that is
based on equality saturation. Our approach has a variety of benefits
over previous compilation models: it addresses the phase ordering
problem, it enables global profitability heuristics, and it performs
translation validation.

There are a variety of directions for future work. One direction is to
extend \Peggy so that it generates a proof of correctness for the
optimizations it performs. Each step in this proof would be the
application of an equality analysis. Since the majority of our
analyses are axiom applications, these proofs would be similar to
standard mathematical proofs. We would then like to use these proofs
as a way of automatically generating optimizations. In particular, by
determining which nodes of the original \PEG the proof depends on, and
what properties of these nodes are important, we will explore
how one can generalize not only the proof but also the
transformation. Using such an approach, we hope to develop a compiler
that can learn optimizations as it compiles.

Another direction involves addressing our heap linearizing issues when
reverting a \PEG to a CFG. One promising solution to this problem
involves adapting our \PEG representation to use \emph{string
  diagrams}~\cite{citeulike:4121553, 1432053}. Expressions are an
excellent theory for non-linear values; string diagrams are a similar
theory, but for linear values. A string diagram is comprised of nodes
with many inputs and many outputs along with strings which connect
outputs of nodes to inputs of other nodes. By default these strings
cannot be forked, capturing the linear quality of the values carried
by the strings; however, strings for non-linear types are privileged
with the ability to fork.
In addition to using string diagrams to encode linearity in our \PEGs,
we could also re-express all of our axioms in terms of string
diagrams, thus preserving the linear qualities of any strings
involved. This prevents the saturation engine from producing \PEGs
which cannot be linearized without additional information. Also,
string diagrams can be used to preserve well-formedness of
PEGs. Well-formedness constraints are the only quadratic component of
our Pseudo-Boolean profitability heuristic formulation, so removing
these constraints could drastically improve the speed of our
Pseudo-Boolean solver.

\newcommand{\meta}[1]{\mathbf{#1}}
\newcommand{\sgap}{\vspace{0.2cm}}
\newcommand{\mgap}{\vspace{0.7cm}}
\newcommand{\bgap}{\vspace{1.2cm}}
\newcommand{\desc}[1]{{\small [{\it #1}]}}
\newcommand{\ax}{\item[$\bullet$]}
\newcommand{\textop}[1]{\textsc{#1}}

\appendix
\section{Axioms}
\label{sec:axioms}

In this section we describe the axioms used to produce the
optimizations listed in Figure~\ref{fig:opts}. We organize the axioms
into two categories: general-purpose and domain-specific. The
general-purpose axioms are useful enough to apply to a wide range of
programming domains, while the domain-specific axioms give useful
information about a particular domain.

The axioms provided below are not a complete list of the ones
generally included in our engine during saturation. Instead, we
highlight only those that were necessary to perform the optimizations in
Figure~\ref{fig:opts}.

\subsection{General-purpose Axioms}

The axioms presented here are usable in a wide range of
programs. Hence, these axioms are included in all runs of \Peggy.


\begin{enumerate}[\hbox to8 pt{\hfill}]

\item\noindent{\hskip-12 pt\bf (Built-in \EPEG ops):}\ This group of axioms relates to the
  special \PEG operators $\theta, \eval,\text{ and }\phi$. Many of
  these axioms describe properties that hold for any operation \textsc{op}.

\begin{enumerate}[$\bullet$]
  \ax if $T=\theta_i(\meta{A},T)$ exists, then $T=\meta{A}$ \\
  \desc{If a loop-varying value always equals its previous value, then it equals its initial value}
  
  
  \ax if $\meta{A}$ is invariant w.r.t. $i$, then $\eval_i(\meta{A},\meta{P}) = \meta{A}$  \\
  \desc{Loop-invariant operations have the same value regardless of the current loop iteration}
  
  
  \ax \begin{tabbing}
    $\textop{op}(\meta{A_1},\ldots,\theta_j(\meta{B_i},\meta{C_i}),\ldots,\meta{A_k}) = $
    $\theta_j($\=$\textop{op}(\eval_j(\meta{A_1},0),\ldots,\meta{B_i},\ldots,\eval_j(\meta{A_k},0)), $ \\
               \>$\textop{op}(\peel_j(\meta{A_1}),\ldots,\meta{C_i},\ldots,\peel_j(\meta{A_k})))$
  \end{tabbing}
  \desc{Any operator can distribute through $\theta_j$}
  
  
  \ax $\phi(\meta{C},\meta{A},\meta{A}) = \meta{A}$ \\
  \desc{If a $\phi$ node has the same value regardless of its condition, then it is equal to that value}
  
  
  \ax $\phi(\meta{C},\phi(\meta{C},\meta{T_2},\meta{F_2}),\meta{F_1}) = \phi(\meta{C},\meta{T_2},\meta{F_1})$ \\
  \desc{A $\phi$ node in a context where its condition is true is equal to its true case}
  
  
  \ax \begin{tabbing} 
    $\textop{op}(\meta{A_1},\ldots,\phi(\meta{B},\meta{C},\meta{D}),\ldots,\meta{A_k}) = \phi(\meta{B},$\=$\textop{op}(\meta{A_1},\ldots,\meta{C},\ldots,\meta{A_k}),$ \\
                                                                                                        \>$\textop{op}(\meta{A_1},\ldots,\meta{D},\ldots,\meta{A_k}))$
  \end{tabbing}
  \desc{All operators distribute through $\phi$ nodes}
  
  
  \ax $\textop{op}(\meta{A_1},\ldots,\eval_j(\meta{A_i},\meta{P}),\ldots,\meta{A_k}) = \eval_j(\textop{op}(\meta{A_1},\ldots,\meta{A_i},\ldots,\meta{A_k}),\meta{P})$, when \\
  $\meta{A_1},\ldots,\meta{A_{i-1}},\meta{A_{i+1}},\ldots,\meta{A_k}$ are invariant w.r.t. $j$ \\
  \desc{Any operator can distribute through $\eval_j$}
\end{enumerate}


\item\noindent{\hskip-12 pt\bf (Code patterns):}\ These axioms are more elaborate and describe
  some complicated (yet still non-domain-specific) code patterns.
  These axioms are awkward to depict using our expression notation, so
  instead we present them in terms of before-and-after source code
  snippets.

\begin{enumerate}[$\bullet$]

  \ax Unroll loop entirely:
\begin{verbatim}
x = B;               ==     x = B;
for (i=0;i<D;i++)           if (D>=0) x += C*D;
   x += C;                  
\end{verbatim}
\desc{Adding C to a variable D times is the same as adding C*D (assuming $D\ge 0$)}


\ax Loop peeling:
\begin{verbatim}
A;                   ==     if (N>0) {
for (i=0;i<N;i++)              B[i -> 0];
   B;                          for (i=1;i<N;i++)
                                  B;
                            } else {
                               A;
                            }
\end{verbatim}
\desc{This axiom describes one specific type of loop peeling, where
  $B[i \rightarrow 0]$ means copying the body of $B$ and replacing all
  uses of $i$ with 0}


  \ax Replace loop with constant:
\begin{verbatim}
for (i=0;i<N;i++){}     ==     x = N;
x = i;
\end{verbatim}
\desc{Incrementing N times starting at 0 produces N}
\end{enumerate}


\item\noindent{\hskip-12 pt\bf (Basic Arithmetic):}\ This group of axioms encodes arithmetic
  properties including facts about addition, multiplication, and
  inequalities. Once again, this is not the complete list of
  arithmetic axioms used in \Peggy, just those that were relevant to
  the optimizations mentioned in Figure~\ref{fig:opts}.

\begin{enumerate}[$\bullet$]
  \ax $(\meta{A}*\meta{B}) + (\meta{A}*\meta{C}) = \meta{A}*(\meta{B}+\meta{C})$
  
  \ax if $\meta{C}\ne 0$, then $(\meta{A}/\meta{C})*\meta{C} = \meta{A}$
  
  \ax $\meta{A}*\meta{B} = \meta{B}*\meta{A}$

  \ax $\meta{A}+\meta{B} = \meta{B}+\meta{A}$
  
  \ax $\meta{A}*1 = \meta{A}$

  \ax $\meta{A}+0 = \meta{A}$
  
  \ax $\meta{A}*0 = 0$
  
  \ax $\meta{A}-\meta{A} = 0$
  
  \ax $\meta{A} \mod 8 = \meta{A} \& 7$
  
  \ax $\meta{A} + (-\meta{B}) = \meta{A}-\meta{B}$

  \ax $-(-\meta{A}) = \meta{A}$
  
  \ax $\meta{A}*2 = \meta{A}<<1$
  
  \ax $(\meta{A}+\meta{B})-\meta{C} = \meta{A}+(\meta{B}-\meta{C})$
  
  \ax $(\meta{A}+\meta{B})+\meta{C} = \meta{A}+(\meta{B}+\meta{C})$
  
  \ax if $\meta{A} \ge \meta{B}$ then $(\meta{A}+1) > \meta{B}$ 
  
  \ax if $\meta{A} \le \meta{B}$ then $(\meta{A}-1) < \meta{B}$ 
  
  \ax $(\meta{A}>\meta{A}) = \false$ 
  
  \ax $(\meta{A}\ge\meta{A}) = \true$
  
  \ax $\lnot(\meta{A}>\meta{B})~=~(\meta{A}\le\meta{B})$
  
  \ax $\lnot(\meta{A} \le \meta{B})~=~(\meta{A} > \meta{B})$

  \ax $(\meta{A}<\meta{B})~=~(\meta{B}>\meta{A})$
  
  \ax $(\meta{A} \le \meta{B})~=~(\meta{B} \ge \meta{A})$
  
  \ax if $\meta{A}\ge\meta{B}$ and $\meta{C}\ge 0$ then $(\meta{A}*\meta{C})\ge(\meta{B}*\meta{C})$
\end{enumerate}


\item\noindent{\hskip-12 pt\bf (Java-specific):}\ This group of axioms describes facts about
  Java-specific operations like reading from an array or field. Though
  they refer to Java operators explicitly, these axioms are still
  general-purpose within the scope of the Java programming language.

\begin{enumerate}[$\bullet$]
\ax $\textop{getarray}( \textop{setarray}(\meta{S},\meta{A},\meta{I},\meta{V}), \meta{A}, \meta{I}) = \meta{V}$  \\
  \desc{Reading $\meta{A}[\meta{I}]$ after writing $\meta{A}[\meta{I}]\leftarrow\meta{V}$ yields $\meta{V}$}

  
  \ax If $\meta{I}\ne \meta{J}$, then $\textop{getarray}(\textop{setarray}(\meta{S},\meta{A},\meta{J},\meta{V}),\meta{A},\meta{I}) = \textop{getarray}(\meta{S},\meta{A},\meta{I})$ \\
    \desc{Reading $\meta{A}[\meta{I}]$ after writing $\meta{A}[\meta{J}]$ (where $\meta{I}\ne \meta{J}$) is the same as reading before the write}

  
  \ax $\textop{setarray}(\textop{setarray}(\meta{S},\meta{A},\meta{I},\meta{V_1}),\meta{A},\meta{I},\meta{V_2}) = \textop{setarray}(\meta{S},\meta{A},\meta{I},\meta{V_2})$ \\
    \desc{Writing $\meta{A}[\meta{I}]\leftarrow\meta{V_1}$ then $\meta{A}[\meta{I}]\leftarrow\meta{V_2}$ is the same as only writing $\meta{V_2}$}

  
  \ax $\textop{getfield}(\textop{setfield}(\meta{S},\meta{O},\meta{F},\meta{V}),\meta{O},\meta{F}) = \meta{V}$ \\
    \desc{Reading $\meta{O}.\meta{F}$ after writing $\meta{O}.\meta{F}\leftarrow\meta{V}$ yields $\meta{V}$}


  \ax If $\meta{F_1}\ne \meta{F_2}$, \\
    \quad then $\textop{getfield}(\textop{setfield}(\meta{S},\meta{O},\meta{F_1},\meta{V}),\meta{O},\meta{F_2}) = \textop{getfield}(\meta{S},\meta{O},\meta{F_2})$ \\
    \desc{Reading $\meta{A}[\meta{I}]$ after writing $\meta{A}[\meta{J}]$ (where $\meta{I}\ne \meta{J}$) is the same as reading before the write}


  \ax $\textop{setfield}(\textop{setfield}(\meta{S},\meta{O},\meta{F},\meta{V_1}),\meta{O},\meta{F},\meta{V_2}) = 
  \textop{setfield}(\meta{S},\meta{O},\meta{F},\meta{V_2})$ \\
    \desc{Writing $\meta{O}.\meta{F}\leftarrow \meta{V_1}$ then $\meta{O}.\meta{F}\leftarrow\meta{V_2}$ is the same as only writing $\meta{V_2}$}
\end{enumerate}

\end{enumerate}


\subsection{Domain-specific} 
Each of these axioms provides useful information about a particular
programming domain. These could be considered ``application-specific''
or ``program-specific'' axioms, and are only expected to apply to that
particular application/program.

\begin{enumerate}[\hbox to8 pt{\hfill}]

\item\item\noindent{\hskip-12 pt\bf (Inlining):}\ Inlining in \Peggy acts like one giant axiom
  application, equating the inputs of the inlined \PEG with the actual
  parameters, and the outputs of the \PEG with the outputs of the
  \textsc{invoke} operator.

\begin{enumerate}[$\bullet$]
  \ax Inlining axiom:
\begin{verbatim}
x = pow(A,B);       ==      result = 1;
                            for (e = 0;e < B;e++)
                               result *= A;
                            x = result;
\end{verbatim}
\desc{A method call to pow is equal to its inlined body}
\end{enumerate}


\item\item\noindent{\hskip-12 pt\bf (Sigma-invariance):}| It is very common for certain Java methods to
  have no effect on the heap. This fact is often useful, and can
  easily be encoded with axioms like the following.

\begin{enumerate}[$\bullet$]
\ax $\rho_{\sigma}(\textop{invoke}(\meta{S},\meta{L},\textop{[{\tt Object List.get()}]},\meta{P}))~=~\meta{S}$ \\
  \desc{{\tt List.get} is $\sigma$-invariant}
  
  
\ax $\rho_{\sigma}(\textop{invoke}(\meta{S},\meta{L},\textop{[{\tt int List.size()}]},\meta{P}))~=~\meta{S}$ \\
  \desc{{\tt List.size} is $\sigma$-invariant}


\ax $\rho_{\sigma}(\textop{invokestatic}(\meta{S},\textop{[{\tt double Math.sqrt(double)}]},\meta{P})) = \meta{S}$ \\
  \desc{{\tt Math.sqrt} is $\sigma$-invariant}

\end{enumerate}


\item\item\noindent{\hskip-12 pt\bf (Vector axioms):}\ In our raytracer benchmark, there are many
  methods that deal with immutable 3D vectors. The following are some
  axioms that pertain to methods of the Vector class. These axioms
  when expressed in terms of \PEG nodes are large and awkward, so we
  present them here in terms of before-and-after source code snippets.

\begin{enumerate}[$\bullet$]
  
\ax $\text{\tt construct(}\meta{A},\meta{B},\meta{C}\text{\tt ).scaled(}\meta{D}\text{\tt )} = \text{\tt construct(}\meta{A}*\meta{D},\meta{B}*\meta{D},\meta{C}*\meta{D}\text{\tt )}$ \\
  \desc{Vector $(A,B,C)$ scaled by $D$ equals vector $(A*D,B*D,C*D)$}

  
\ax $\meta{A}\text{\tt .distance2(}\meta{B}\text{\tt )} = \meta{A}\text{\tt .difference(}\meta{B}\text{\tt ).length2()}$ \\
  \desc{The squared distance between $A$ and $B$ equals the squared length of vector $(A-B)$}
  
  
\ax  $\begin{array}{l} \meta{A}\text{\tt .getX()} = \meta{A}\text{\tt .mX} \\ \meta{A}\text{\tt .getY()} = \meta{A}\text{\tt .mY} \\ \meta{A}\text{\tt .getZ()} = \meta{A}\text{\tt .mZ} \end{array}$ \\
  \desc{Calling the getter method is equal to accessing the field directly}
  
  
\ax $\begin{array}{l}\text{\tt construct(}\meta{A},\meta{B},\meta{C}\text{\tt ).mX} = \meta{A} \\ \text{\tt construct(}\meta{A},\meta{B},\meta{C}\text{\tt ).mY} = \meta{B} \\ \text{\tt construct(}\meta{A},\meta{B},\meta{C}\text{\tt ).mZ} = \meta{C} \end{array}$ \\
  \desc{Accessing the field of constructed vector $(A,B,C)$ is equal to appropriate parameter}
  
  
\ax $\begin{array}{l}
  \text{\tt construct(}\meta{A},\meta{B},\meta{C}\text{\tt ).difference(construct(}\meta{D},\meta{E},\meta{F}\text{\tt ))} =  \\
  \text{\tt construct(}\meta{A}-\meta{D},\meta{B}-\meta{E},\meta{C}-\meta{F}\text{\tt )} \\
\end{array}$ \\
  \desc{The difference of vectors $(A,B,C)$ and $(D,E,F)$ equals $(A-D,B-E,C-F)$}

  
\ax $\text{\tt construct(}\meta{A},\meta{B},\meta{C}\text{\tt ).dot(construct(}\meta{D},\meta{E},\meta{F}\text{\tt ))} = \meta{A}*\meta{D} + \meta{B}*\meta{E} + \meta{C}*\meta{F}$ \\
  \desc{The dot product of vectors $(A,B,C)$ and $(D,E,F)$ equals $A*D+B*E+C*F$}
  
  
\ax  $\text{\tt construct(}\meta{A},\meta{B},\meta{C}\text{\tt ).length2()} = \meta{A}*\meta{A} + \meta{B}*\meta{B} + \meta{C}*\meta{C}$ \\
  \desc{The squared length of vector $(A,B,C)$ equals $A^2+B^2+C^2$}
  
  
\ax $\text{\tt construct(}\meta{A},\meta{B},\meta{C}\text{\tt ).negative()} = \text{\tt construct(}-\meta{A},-\meta{B},-\meta{C}\text{\tt )}$ \\
  \desc{The negation of vector $(A,B,C)$ equals $(-A,-B,-C)$}
  
  
\ax $\text{\tt construct(}\meta{A},\meta{B},\meta{C}\text{\tt ).scaled(}\meta{D}\text{\tt )} = \text{\tt construct(}\meta{A}*\meta{D},\meta{B}*\meta{D},\meta{C}*\meta{D}\text{\tt )}$ \\
  \desc{Scaling vector $(A,B,C)$ by $D$ equals $(A*D,B*D,C*D)$}
  
  
\ax $\begin{array}{l} \text{\tt construct(}\meta{A},\meta{B},\meta{C}\text{\tt ).sum(construct(}\meta{D},\meta{E},\meta{F}\text{\tt ))} =  \\ \text{\tt construct(}\meta{A}+\meta{D},\meta{B}+\meta{E},\meta{C}+\meta{F}\text{\tt )} \end{array}$ \\
  \desc{The sum of vectors $(A,B,C)$ and $(D,E,F)$ equals $(A+D,B+E,C+F)$}
  
  
\ax $\text{\tt getZero().mX} = \text{\tt getZero().mY} = \text{\tt getZero().mZ} = 0.0$  \\
  \desc{The components of the zero vector are 0}

\end{enumerate}


\item\item\noindent{\hskip-12 pt\bf (Design patterns):}\ These axioms encode scenarios that occur when
  programmers use particular coding styles that are common but
  inefficient.

\begin{enumerate}[$\bullet$]
  
\ax Axiom about integer wrapper object: \\
$\meta{A}\text{\tt .plus(}\meta{B}\text{\tt ).getValue()} = \meta{A}\text{\tt .getValue()} + \meta{B}\text{\tt .getValue()}$ \\
\desc{Where {\tt plus} returns a new integer wrapper, and {\tt getValue} returns the wrapped value}


\ax Axiom about redundant method calls when using {\tt java.util.List}:
\begin{verbatim}
Object o = ...                 ==   Object o = ...
List l = ...                        List l = ...
if (l.contains(o)) {                int index = l.indexOf(o);
   int index = l.indexOf(o);        if (index >= 0) {
   ...                                 ...
}                                   }
\end{verbatim}
\desc{Checking if a list contains an item then asking for its index is redundant}

\end{enumerate}


\item\item\noindent{\hskip-12 pt\bf (Method Outlining):}\ Method ``outlining'' is the opposite of
  method inlining; it is an attempt to replace a snippet of code with
  a procedure call that performs the same task. This type of
  optimization is useful when refactoring code to remove a common yet
  inefficient snippet of code, by replacing it with a more efficient
  library implementation.

\begin{enumerate}[$\bullet$]
  
  \ax Body of selection sort replaced with {\tt Arrays.sort(int[])}: 
\begin{verbatim}
length = A.length;                  ==       Arrays.sort(A);
for (i=0;i<length;i++) {
   for (j=i+1;j<length;j++) {
      if (A[i] > A[j]) {
         temp = A[i];
         A[i] = A[j];
         A[j] = temp;
      }
   }
}
\end{verbatim}
\end{enumerate}


\item\item\noindent{\hskip-12 pt\bf (Specialized Redirect):}\ This optimization is similar to Method
  Outlining, but instead of replacing a snippet of code with a
  procedure call, it replaces one procedure call with an equivalent
  yet more efficient one. This is usually in response to some learned
  contextual information that allows the program to use a special-case
  implementation.

\begin{enumerate}[$\bullet$]
  \ax     \begin{tabbing}
    if~\=$I=\textop{invokestatic}(\meta{S},\textop{[{\tt void sort(int[])}]},\textop{params}(\meta{A}))$ exists, \\
    \>then add equality $\text{\it isSorted}(\rho_{\sigma}(I),\meta{A}) = \true$
  \end{tabbing} 
  \desc{If you call {\tt sort} on an array $A$, then $A$ is sorted in the subsequent heap}

  
  \ax if  $\text{\it isSorted}(\meta{S},\meta{A}) = \true$, then \\
      $\textop{invokestatic}(\meta{S},\textop{[{\tt int linearSearch(int[],int)}]},\textop{params}(\meta{A},\meta{B})) = $ \\
      $\textop{invokestatic}(\meta{S},\textop{[{\tt int binarySearch(int[],int)}]},\textop{params}(\meta{A},\meta{B}))$  \\
  \desc{If array $A$ is sorted, then a linear search equals a binary search}
\end{enumerate}

\end{enumerate}

\section*{Acknowledgements}
  We would like to thank Jeanne Ferrante, Todd Millstein, Christopher Gautier, members of
  the UCSD Programming Systems group, and the anonymous reviewers for
  giving us invaluable feedback on earlier drafts of this paper.

\bibliographystyle{plain}
\bibliography{paper}

\end{document}